\documentclass[12pt]{article}
\usepackage{amsmath,amssymb}
\usepackage{graphicx}
\usepackage{amsthm}

\def\AC{\mathcal{A}}

\def\LC{\mathcal{L}}

\def\FC{\mathcal{F}}
\def\HC{\mathcal{H}}

\def\C{\mathbf{C}}
\def\E{\mathbf{E}}

\def\P{\mathbf{P}}
\def\Q{\mathbf{Q}}
\def\R{\mathbf{R}}
\def\Z{\mathbf{Z}}

\def\1{\mathbf{1}}

\def\e{\mathbf{e}}

\def\tr{\rm{tr}}

\def\al{\alpha}
\def\be{\beta}
\def\pa{\partial}
\def\ep{\epsilon}
\def\de{\delta}
\def\ga{\gamma}
\def\ka{\varkappa}

\newtheorem{prop}{Proposition}[section]
\newtheorem{theorem}{Theorem}[section]
\newtheorem{lemma}{Lemma}[section]

\newtheorem{remark}{Remark}[section]

\newcommand{\la}{\lambda}
\newcommand{\si}{\sigma}

\newcommand{\om}{\omega}
\newcommand{\Ga}{\Gamma}
\newcommand{\De}{\Delta}

\newcommand{\Om}{\Omega}

\begin{document}
\title{Quantum filtering and propagation of chaos for open quantum systems,
with applications to quantum feedback control and quantum mean-field games}

\author{Vassili N. Kolokoltsov
\thanks{Department of Mathematical Statistics, Faculty of Computational Mathematics and Cybernetics,
Lomonosov Moscow State University, Leninskie Gory, 1, 119991 Moscow, Russia;
Moscow Center of Fundamental and Applied Mathematics, Leninskie Gory, 1, 119234 Moscow, Russia;
International Laboratory of Stochastic Analysis and its Applications,
National Research University Higher School of Economics, Moscow;
Email: v.n.kolokoltsov@gmail.com}}
\maketitle

\begin{abstract}
The theory of quantum filtering (of quantum continuous measurements) was developed
by V.P. Belavkin about 40 years ago. Since then it attracted attention of numerous
investigators including mathematicians, theoretical and experimental physicists.
However, the rigorous mathematical theory of the filtering equations for mixed
states in basic infinite-dimensional quantum systems
remained an open problem, which was resolved by the author recently.
This survey paper presents in full the mathematical
theory of quantum filtering equations, their rigorous derivation from
basic principles, the corresponding law of large number limits (propagation of chaos),
fractional quantum mechanics of open systems
and related topics. Applications to feedback control, quantum dynamic and mean-field games
are discussed.
\end{abstract}


\vspace{3mm}

{\bf Content.}

1. Introduction.
1.1. Objectives and brief content. 1.2. Basic notations.
1.3. Belavkin's quantum filtering equations: diffusive case.
1.4. Belavkin's quantum filtering equations: counting case.
1.5. SDEs and PDEs. 1.6. Quantum filtering for mixed channels.
1.7. Interaction picture and mild equations.

2. Well-posedness of quantum filtering equations.
2.1. Preliminaries: equations for pure states. 2.2. Counting measurements.
2.3. Diffusive case: linear version. 2.4. Semigroups of linear filtering equations in diffusive case.
2.5. Diffusive case: normalized version. 2.6. Semigroups generated by nonlinear filtering equations.
2.7. Classical Hamiltonians. 2.8. Appendix: some trace inequalities.

3. Derivation of quantum filtering equations.
3.1. Quantum measurement. 3.2. Quantum Zeno paradox and watchdog effect.
3.3. Markov chains of sequential indirect observations. 3.4. Counting observation.
3.5. Diffusive and mixed observation. 3.6. Appendix: convergence of semigroups.

4. Quantum LLN (propagation of chaos) for continuously observed quantum systems.
4.1. Stochastic master equations for mean-field interacting particles.
4.2. Operators of interaction. 4.3. Derivation of mean-field limit: diffusive measurement.
4.4. Extension with control. 4.5. Derivation of mean-field limit: counting measurement.
4.6. Appendix A: a technical estimate. 4.7. Appendix B: stochastic version of Knowles-Pickl estimates.

5. Fractional quantum mechanics of open systems.

6. Examples and properties of solutions.
6.1. Classical quantum particle under continuous observation of its position and/or momentum.
6.2. Long-time behaviour for classical Hamiltonians.
6.3. Path integral representation for classical Hamiltonians.
6.4. Finite-dimensional examples.

7. Quantum feedback control and quantum mean field games.

8. Bibliographic notes and open problems.

\section{Introduction}

\subsection{Objectives and brief content}

The theory of quantum filtering (of quantum continuous measurements) was developed
about 40 years ago  essentially in Belavkin's papers  \cite{Bel87}, \cite{Bel88}, \cite{Bel92}
on the basis of the heavy technique of quantum stochastic calculus.
Since then it attracted attention of numerous investigators including mathematicians,
theoretical and experimental physicists. However, the rigorous mathematical theory
of the basic filtering equations, so-called Belavkin's equations, for mixed
states in basic infinite-dimensional quantum systems remained an open problem.
This was due, at least partially, to the principle difficulty of working with
stochastic equations in the Banach space of trace-class operators (which is not
included in several classes of Banach spaces, where appropriate extensions of Ito's
stochastic calculus was built). Recently, in \cite{Kol25a} and \cite{Kol25b},
 the author managed to overcome these difficulties
and build the rigorous mathematical theory of these equations. In particular,
it led to a possibility to make a rigorous derivation of the quantum filtering equations
from the basic principles for standard infinite-dimensional quantum models.

In another related recent development, see \cite{KolQuantLLN},
 \cite{KolQuantMFGCount} and \cite{KolQuantMFG},
 the author derived the law of large numbers limit for
continuously observed (or open) quantum systems leading to the new class of stochastic
equations for effective states (which can be looked at as kind of nonlinear stochastic
Schr\"odinger equations or as infinite-dimensional operator-valued McKean-Vlasov diffusions)
describing the so-called propagation of chaos for these systems. These theories lead to
numerous other developments including fractional quantum mechanics for open quantum systems,
new results in quantum feedback control and eventually to the theory of quantum dynamic
games and quantum mean-field games. The latter models present a far-reaching extension
of the one of the most popular branch of modern game theory, namely, the theory of
mean-field games initiated about 20 years ago.

In this survey it was planned to present the main ideas of these and related developments in a unified
systematic and (as much as possible) self-contained manner paying main attention to the
rigorous mathematical aspects of the theory. However, during the preparation, several
new important aspects and results popped out around the 3 main topics of the paper: well-posedness
of filtering equations, their rigorous derivation and the many-particle  limit, which eventually took
the main body of the paper presented in Sections \ref{secwellpos}, \ref{secder}, \ref{secLLN}.
In order to remain in a reasonable size, other developments turned out to be presented here
in a more sketchy way.

Continuous measurement and filtering of quantum systems can be organised in two versions:
counting and diffusive type detections.
The corresponding dynamics can be described in terms of the stochastic evolution of pure or mixed states.
The mathematics of the evolution of pure states given by the Belavkin equations (and representing some kind
of stochastic nonlinear Schr\"odinger equation) is fairly well understood by now. In the present paper
the main stress is on a more subtle case of the operator-valued evolution of mixed states.

Solutions to quantum filtering equations are often called quantum trajectories.
Evolutions of quantum trajectories are specified by two ingredients: the Hamiltonian $H$ of the free motion
(self-adjoint operator in a Hilbert space) and the coupling operator $L$ with a reservoir
(or measurement apparatus). In this paper we decided to reduce our attention mostly
(apart from some examples) to the case of bounded coupling operator $L$, as the theory
is already rich enough under this assumption. For the case of unbounded $L$ we refer to
\cite{Kol25b} and references therein. In particular, one can find there the well-posedness results
for quantum filtering equations for mixed states under some basic assumptions on
unbounded $L$.

Equations of quantum filtering can be also looked at as stochastic master or Lindblad
 equation yielding the so-called {\it unravelling} of quantum dynamic semigroups, generated by Lindblad
 operators (also referred to as GKSL operators in honour of contributions \cite{GKS} and \cite{Lindblad}),
see e.g. papers \cite{Gisin}, \cite{GisinPer}, \cite{Dalibard} and
monographs \cite{Carm}, \cite{BarchBook} and \cite{Breuer}.
This unravelling has been widely used for numeric Monte-Carlo evaluations of quantum dynamic
semigroups, see a detailed presentation of algorithms in \cite{Breuer}.

The content of the paper is as follows. In the rest of the introduction we first fix
some notations, then present the main quantum filtering equations for counting,
 diffusive and mixed observations in their various forms, in particular, in both  pure- and
 mix-states versions, then outline the crucial link
 between SDEs (stochastic differential equations) and PDEs (partial differential equations)
 and finally recall the interaction picture of quantum mechanics that is very handy for
 our analysis.

 Section \ref{secwellpos} is devoted to the well-posedness of quantum
 filtering equations. Apart from certain improvements and extensions
 (including the treatment of counting and mixed channels of observations),
 the main new aspects here is the development of the PDEs (partial
 differential equations) version of the corresponding evolutions,
 which, in this infinite-dimensional setting, turned out to be not trivially implied by the SDE theory.
Notice also that, from the pure mathematical point of view, what we are doing here is
building the theory of a certain class of infinite -dimensional SDEs (stochastic differential equations)
with singular coefficients.

Section \ref{secder} is devoted to the rigorous derivation of quantum filtering equations.
The general scheme was already well understood, but rigorous presentation was available
so far only for finite-dimensional quantum mechanics and is presented here in full
 for the first time.

Section \ref{secLLN} deals with
the law of large and the propagation of chaos for continuously observed systems of quantum
particles, under diffusive and counting observations.
In particular, we present the complete result for the
basic case of the multiplication operator of interaction.
In previous papers of the author the interactions given by simpler Hilbert-Schmidt
integral operators were properly analysed. The result with multiplication operators was
given in terms of hardly verifiable quantities. A similar improvement
(with different details) to the main result of \cite{KolQuantMFG} on diffusive case
was obtained in recent paper \cite{Bouard} (which the author learned about when the
present manuscript was essentially ready).

In Sections \ref{secfraceq} we derive the new filtering
equations of fractional quantum mechanics of open systems. This derivation was performed
initially in \cite{KolQuantFrac} for finite-dimensional quantum systems. Here we complete
it for the general case supplying also an SDE formulation and the rates of convergence that became available after
recent progress in the CTRW theory achieved in \cite{KolRates}.

In Section \ref{secexamprop} we briefly present some basic examples of quantum filtering equations,
give natural representations for their solutions and show possibilities of ergodic and scattering
behavior for large times. In Section \ref{seccont} applications to feedback control and quantum games
are briefly described. Section \ref{secbibopen} is devoted to brief bibliographic comments on related
developments and to the formulation of some open problems.

  \subsection{Basic notations}

In this paper the letters $H$ and $L=(L_1, \cdots, L_n)$ will denote linear operators
in a separable Hilbert space $\HC$. Here $H$ is self-adjoint and referred to as
the Hamiltonian. The vector-valued $L$ describes
the coupling operator with the measurement device. We shall use the notation
$L_S=(L+L^*)/2=(L_{S1}, \cdots, L_{Sn})$ and $L_A=(L-L^*)/2i=(L_{A1}, \cdots, L_{An})$
for symmetric and antisymmetric parts of $L$ with the adjoint $L^*$, and $\|L\|=\sum_j \|L_j\|$.
We shall write $\, Re \, z$ and $Im \, z$ for the real and imaginary parts of a complex number
or vector. The brackets $[A,B]$ and $\{A,B\}$
will denote the commutator and anti-commutator of operators $A,B$.

The space of bounded operators in $\HC$ will be denoted by $\LC=\LC(\HC)$
with $\|A\|$ the operator norm of $A\in \LC$.
For a Banach space $B$, we shall denote by $B^*$ the dual Banach space
and by  $\LC(B)$ the space of bounded linear operators in $B$, with the norm $\|A\|_{B}$
for $A\in \LC(B)$. Most important spaces for us will be (1)
the Hilbert space $\HC^2_s$ of self-adjoint Hilbert-Schmidt
operators in $\HC$, which is a closed subspace
in the space of all Hilbert-Schmidt operators $\HC^2$
with the scalar product ${\tr} (A^*B)$, and (2) the Banach space
$\HC^1_s$ of self-adjoint trace-class
operators in $\HC$, which is a closed subspace
in the space of all trace-class operators $\HC^1$
with the norm $\|A\|_{\HC^1}={\tr} |A|$. It is well known that the space $\LC$
is the Banach dual to $\HC^1$ with the duality pairing being given by the trace:
$(B,A)={\tr}(B^*A)$ for $A\in \HC^1$.

As usual, $L^p(\R^d)$ are the spaces of functions on $\R^d$ with integrable $p$th moment
and $W^2(\R^d)$ is the Sobolev space of functions $f\in L^2(\R^d)$
with weak partial derivatives up to order $2$
also belonging to $L^2(\R^d)$.
We denote $C(M)$ the Banach space of bounded continuous
real functions on a metric space $M$ equipped with the usual sup-norm,
and by $C_{luc}(M)$ its subspace consisting of locally uniformly continuous
functions, that is, uniformly continuous on each bounded subset.
For $M$ a closed subset of a Banach space, we shall denote $C_{\infty}(M)$
the closed subspace of $C(M)$ consisting of functions tending to zero at infinity.
We use also the standard notions for Gateaux and Fr\'echet differentiability of
functions in Banach spaces (see e.g. \cite{Toolbox} or Sections 1.2-1.5 in \cite{Kolbook19}).
In particular, $f$ is called Gateaux differentiable at a point $\ga$ of a Banach space $B$
whenever there exists an element $f'(\ga)\in B^*$ such that
\[
f(\ga+h\rho)=f(\ga)+(f'(\ga), h\rho)+o(h),
\]
as $h\to 0$, for any $\rho \in B$. And  $f$ is called Fr\'echet differentiable at $\ga$
if additionally
\[
f(\ga+\rho)=f(\ga)+(f'(\ga), \rho)+o(\|\rho\|_B)
\]
for $\rho \in B$.
We denote by $C^1(B)$ the space of Fr\'echet continuously differentiable
functions $f$ from $C(B)$ with a bounded derivative, that is, when the Fr\'echet
derivative $f'(\ga)$ exists for all $\ga$ and  the mapping $\ga\mapsto f'(\ga)=Df(\ga)$
from $B$ to $B^*$ is bounded and continuous with respect to the norm topologies of $B$ and $B^*$.

The space $C^1(B)$ becomes a Banach space when equipped with the norm
\[
\|f\|_{C^1(B)}=\sup_{\ga} |f(\ga)|+\sup_{\ga} \|f'(\ga)\|.
\]
Similarly we denote by $C^2(B)$ the space of functions from
$C^1(B)$ with the second bounded continuous derivative, that is, when
there exists a bounded continuous mapping $\ga \mapsto f''(\ga)=D^2f(\ga)$ from $B$
to bounded bilinear forms on $B$ such that
\[
f(\ga+\rho)=f(\ga)+(f'(\ga), \rho)+\frac12 (\rho, f''(\ga) \rho) +o(\|\rho\|^2_B).
\]
The norm in this space is defined as
\[
\|f\|_{C^2(B)}=\sup_{\ga} |f(\ga)|+\sup_{\ga} \|f'(\ga)\|+\sup_{\ga}\|f''(\ga)\|_{B},
\]
where for multi-linear forms $F$ of order $k$ on
a Banach space $B$ the norm is defined, as usual, by
\[
\|F\|_B=\sup_{v_1, \cdots, v_k: \forall j \, \|v_j\|\le 1} |F(v_1, \cdots, v_n)|.
\]
Similarly other $C^k(B)$ are defined, as well as the spaces
$C^k(B_1,B_2)$ of $k$ times continuously (Fr\'echet) differentiable mappings
between Banach spaces $B_1$ and $B_2$.

We shall use these notions mostly for Hilbert spaces,
because the problem of constructing smooth functions in
arbitrary Banach spaces is extremely difficult, see review in \cite{Toolbox},
and we had to deal there with Gateaux differentials.


Recall that a general isolated quantum system is described by
a Hilbert space $\HC$ and a self-adjoint operator $H$ in it,
the Hamiltonian. The pure states of the system are vectors in $\HC$
(more precisely, the rays generated by such vectors) and the general (mixed) states are density
operators, also referred to as density matrices (a term preferred by physicists),
that is, non-negative operators in $\HC$ with unit trace.
Let us denote $S(\HC)$ the set of all such mixed states in $\HC$.
 To any pure state $\psi$ there corresponds a one-dimensional (decomposable) density matrix
  $\ga=\psi\otimes \bar \psi$ (denoted $|\psi \rangle \langle \psi|$ in Dirac's notations).

\begin{remark}
In more general models, states are positive functionals on von Neumann algebras, but we do not need such generality.
\end{remark}

In the classical quantum mechanics, pure states evolve in time according to the rule $\psi \to e^{-itH} \psi$
and the mixed states according to the rule  $\ga \to e^{-itH} \ga e^{itH}$.

\subsection{Belavkin's quantum filtering equations: diffusive case}

As we mentioned, quantum filtering distinguishes counting and diffusive observations
(though mixed channels can also be used, see below). Here we start with the description
of diffusive observations. In any case, the corresponding filtering evolutions can be usually
described either by linear non-normalised equations or by nonlinear normalised.

\begin{remark} The meaning of the term 'counting observation' as well as 'diffusive type'
become more concrete in an advanced
treatment of the process of quantum measurement, see e.g.  \cite{BoutHanJamQuantFilt}.
The same concerns also the meaning of terms 'output process' and 'innovation process' introduced below.
\end{remark}

The {\it linear quantum filtering equation} for the diffusive case
or the {\it linear stochastic quantum master equation} or {\it Belavkin's equation for mixed states}
is the equation for the density operators $\ga$ in $\HC$
written down as
\begin{equation}
\label{Lindstoch0}
d\ga(t)=-i[H,\ga(t)] \, dt +\LC_L \ga(t) \, dt +\sum_{j=1}^n (L_j\ga(t)+\ga(t) L^*_j) dY_j(t),
\end{equation}
with the Lindblad operator
\[
\LC_L\ga =\sum_j L_j\ga L^*_j-\frac12 \sum_j L^*_jL_j\ga -\frac12 \sum_j \ga L^*_jL_j
=\sum_j L_j\ga L^*_j-\frac12 \sum_j \{L^*_jL_j,\ga\},
\]
where the process $Y(t)=(Y_1, \cdots, Y_n(t))$, referred to as the {\it output process},
 is a standard Brownian motion (BM) on a certain stochastic basis. If one agrees
 to understand all products of operator expressions as appropriate inner products
(sum over available indices), for instance, writing $L^*L$ instead of $\sum_j L_j^* L_j$,
and $L \chi \, dY(t)$ instead of $\sum_j L_j\chi \, dY_j(t)$,
 this equation gets a much simpler form:

\begin{equation}
\label{Lindstoch}
d\ga(t)=-i[H,\ga(t)] \, dt +\LC_L \ga(t) \, dt +(L\ga(t)+\ga(t) L^*) dY(t),
\end{equation}
with
\[
\LC_L\ga =L\ga L^*-\frac12 L^*L\ga -\frac12 \ga L^*L
=L\ga L^*-\frac12 \{L^*L,\ga\}.
\]

We shall mostly use everywhere such concise form, only occasionally writing the full version
(when it seems to be necessary for clarity).
For finite-dimensional $\HC$, this is a simple linear matrix-valued equation, so that
the well-posedness of the corresponding Cauchy problem follows from the standard results of Ito's calculus.
For general $\HC$ the situation is not so obvious even for bounded $H,L$ and will be discussed later on.

Using Ito's rule one directly derives a stochastic equation for the normalised state
$\rho=\ga/{\tr} \, \ga$:
\begin{equation}
\label{Lindstochnorm1}
d\rho(t)=-i[H,\rho(t)]\, dt+\LC_L \rho (t)\, dt
+[L\rho(t)+\rho(t) L^*-\rho(t)\, {\tr} \, (L\rho(t)+\rho(t) L^*) ] dB(t),
\end{equation}
where
\begin{equation}
\label{eqdefinnov}
dB(t)=dY(t)-2{\tr} \, (\rho(t) L_S) \, dt
\end{equation}
is the Ito process, referred to as the {\it innovation process}.

According to Girsanov's theorem, by an equivalent change of measure one can turn $B(t)$
to the standard Brownian motion (BM). And in fact, one usually chooses $Y$ to be a BM when dealing with a linear
equation and $B$ when dealing with the normalised one.

The evolutions described by \eqref{Lindstochnorm1} and  \eqref{Lindstoch},
their various versions and extensions (including counting evolutions \eqref{eqBeleqcountj}),
are the main objects of studies in this paper.
A derivation of these equations from the point of view of continuous quantum measurement
will be given in Section \ref{secder}.

\begin{remark}
As already mentioned above,
ignoring the noise term in both \eqref{Lindstochnorm1} and \eqref{Lindstoch}
(setting $Y=0$ or $B=0$) leads to
the famous Lindblad equations describing general quantum dynamic semigroups. Therefore, some physicists suggested
equations \eqref{Lindstochnorm1} and \eqref{Lindstoch} independently
as a way of stochastic unravelling of quantum dynamic semigroups,
see numerous references in \cite{Mora13} or \cite{Breuer}.
\end{remark}

Let us make some comments about the evolution of traces.
Equation \eqref{Lindstochnorm1} is best suited for dealing with matrices $\ga$ with unit trace.
A convenient extension to arbitrary traces can be written as
\begin{equation}
\label{Lindstochnorm1m}
d\rho(t)=-i[H,\rho(t)]\, dt +\LC_L \rho(t)\, dt
+[\rho(t)L^*+L\rho(t)-\frac{\rho(t)}{{\tr} \, \rho(t)} \, {\tr} (\rho(t)(L^*+L))]\, dB(t).
\end{equation}
It is seen directly from this equation that $d\, {\tr} \, \rho(t)=0$ and hence the trace is preserved
by evolution \eqref{Lindstochnorm1m} (as long as it is well-posed of course).
For equation \eqref{Lindstochnorm1} one gets
 \begin{equation}
\label{Lindstochnorm1tr}
d\, {\tr} \, \ga(t)= {\tr} (\ga(t)(L^*+L))(1-  {\tr} \, \ga(t)) \, dB(t).
\end{equation}
Thus  evolution \eqref{Lindstochnorm1} does not preserve traces in general. But
${\tr} \, \ga(t)=1$ is a solution of \eqref{Lindstochnorm1tr}. Hence, if the initial
condition has  ${\tr} \, \ga_0=1$ and the equation \eqref{Lindstochnorm1} is well posed, then
their solutions do preserve the trace.

\begin{remark}
In this paper we reduce attention to finite-dimensional noises $B(t)$ or $Y(t)$.
In literature, see e.g. \cite{MoraRebo}, one can find also the corresponding equations
with infinite-dimensional noises. The author believes that most of the results presented
here can be naturally extended to infinite $n$, if carefully looking at the convergence
of the related infinite sums. We decided not to overload this paper
with related technicalities.
\end{remark}

A remarkable property of evolutions \eqref{Lindstochnorm1} and  \eqref{Lindstoch} is the fact that
they preserve pure states. Namely, direct application of Ito's lemma shows that if vectors
 $\chi(t)\in \HC$ evolve according to the {\it linear Belavkin quantum filtering equation
 for non-normalized pure states}
\begin{equation}
\label{eqqufiBlins}
d\chi(t) =-[iH\chi(t) +\frac12 L^*L \chi(t) ]\,dt+L\chi(t) \, dY(t),
\end{equation}
then the decomposable density matrices $\ga(t)=\chi(t)\otimes \bar \chi(t)$
satisfy equation \eqref{Lindstoch}. In fact, from \eqref{eqqufiBlins} it follows that
\[
d\overline{\chi(t)} =[i \overline{H\chi(t)} -\frac12 \overline{L^*L \chi(t)} ]\,dt+\overline{L\chi(t)}\, dY(t).
\]
Consequently, by Ito's product rule,
\[
d\ga(t)=[-iH\chi(t)\otimes \bar \chi -\frac12 L^*L \chi(t)\otimes \bar \chi (t)]\,dt+L\chi(t)\otimes \bar \chi(t) \, dY(t)
\]
\[
+[i \chi(t)\otimes \overline{H\chi(t)} -\frac12 \chi(t) \otimes \overline{L^*L \chi(t)} ]\,dt
+\chi(t)\otimes \overline{L\chi(t)}\, dY(t)+L\chi(t) \otimes \overline{L\chi(t)} \, dt
\]
\[
=(-iH\ga+i\ga H)\, dt -\frac12 \{L^*L,\ga\} \, dt +(L\ga+\ga L^*) \, dY(t)+L\ga L^* \, dt,
\]
as claimed.

On the other hand, as again follows from Ito's formula,
if $\chi(t)$ satisfy \eqref{eqqufiBlins}
then the corresponding normalised vectors $\phi(t)=\chi(t)/\|\chi(t)\|$ satisfy
the {\it nonlinear Belavkin quantum filtering equation for normalized pure states}
\[
d\phi(t)=(L-(\phi(t), L_S \phi(t)))\phi(t) \, dB(t)
\]
\begin{equation}
\label{eqqufiBnonlins}
-[i(H-(\phi(t), L_S \phi(t)) L_A)
+\frac12 (L-(\phi(t), L_S \phi(t)))^*(L-(\phi(t), L_S \phi(t)))]\phi(t) \, dt,
\end{equation}
and the corresponding decomposable normalised
matrices $\rho(t) =\phi(t)\otimes \bar \phi(t)$ satisfy \eqref{Lindstochnorm1},
with $B$ given by \eqref{eqdefinnov}:
\[
dB(t)=dY(t)-2{\tr} \, (\rho(t) L_S) \, dt=dY(t)-2(\phi(t), L_S \phi(t))\, dt.
\]

\begin{remark}
To see the pure mathematical place of equation \eqref{eqqufiBlins} in the theory of SDEs, we can note that
this is a general linear equation in a Hilbert space preserving the expectation of norm squared,
or where the norm squared of a solution is a local martingale. In fact, writing a linear SDE in the form
$d\chi=A\chi \,dt +L\chi \, dY$ with a given operator $L$ and requiring that $d \, \E \|\chi\| ^2(t)=0$,
we get by Ito's formula (at least, its formal application) that $A$ has a structure given in \eqref{eqqufiBlins}.
\end{remark}

\begin{remark}
\label{remtimedep}
Everything given below have a straightforward extension to the case of time dependent families $L(t)$, as
long as these families are continuous (at least in the strong topology)
 and uniformly bounded. In fact, time dependent families arise necessarily
when reducing the case of unbounded $H$ to the case of bounded (in fact, vanishing) $H$ via the interaction
 representation, see below. If $L$ is time-dependent, all estimates below are valid with $\|L\|=\max_t\|L(t)\|$.
Similarly the theory extends to time-dependent Hamiltonian $H_t$ whenever such family generates a well defined
unitary propagator.
\end{remark}

\begin{remark} The square norm $\|\chi\|^2$ of solutions to \eqref{eqqufiBlins} has physical meaning
analogous to the square norm of the solutions to the standard Schr\"odinger equation:
it describes the probability density of observing corresponding values of the output process $Y(t)$,
see detailed discussion in  \cite{BarchBook}.
\end{remark}

As for the case of equations on density matrices,
it is often convenient
to include equations \eqref{eqqufiBnonlins}
in a more general class of norm preserving evolutions, the simplest version being
\[
d\phi(t)=-[i(H-\langle L_S \rangle_{\phi(t)} L_A)
+\frac12 (L-\langle L_S \rangle_{\phi(t)})^*(L-\langle L_S \rangle_{\phi(t)})]\phi(t) \, dt
\]
\begin{equation}
\label{eqqufiBnonlinsn}
+ (L-\langle L_S \rangle_{\phi(t)})\phi(t) \, dB(t),
\end{equation}
where we introduced the (rather standard) notation for the value
of an operator $A$ in a pure state $\phi$:
\[
\langle A \rangle_{\phi}=\frac{(\phi, A \phi)}{(\phi, \phi)}.
\]
Clearly for $\phi(t)$ of unit norm solutions to equations \eqref{eqqufiBnonlins}
and \eqref{eqqufiBnonlinsn} coincide, but \eqref{eqqufiBnonlinsn} is explicitly norm-preserving
for arbitrary $\phi(t)$, which is not the case for equation \eqref{eqqufiBnonlins}.

It is also insightful to write down equation for $\chi$ in terms of the innovation process $B$:
\begin{equation}
\label{eqqufiBlinsB}
d\chi(t) =-[iH\chi(t) +\frac12 L^*L \chi(t) ]\,dt
+L \chi(t) \, (dB(t)+\langle L+L^*\rangle_{\chi(t)} \, dt).
\end{equation}

In the most important case of a self-adjoint $L$ equations
\eqref{eqqufiBlins} and \eqref{eqqufiBnonlinsn} simplify to the equations
\begin{equation}
\label{eqqufiBlinss}
d\chi(t) =-[iH\chi(t) +\frac12 L^2 \chi(t) ]\,dt+L\chi(t) dY(t),
\end{equation}
and, respectively,
\begin{equation}
\label{eqqufiBnonlinss}
d\phi(t)=-[iH+\frac12 (L-(\phi(t), L \phi(t)))^2]\phi(t) \, dt
+ (L-(\phi(t), L \phi(t)))\phi(t) \, dB(t).
\end{equation}

For working with control problems, it is often convenient to work with the Stratonovich rather
than Ito's integral. Using the standard rule for transforming Ito's differential to Stratonovich,
\[
\si(X(t)) \, dB(t)=\si(X(t)) \circ dB(t)-\frac12 \frac{\pa \si}{\pa X}(X(t)) \si(X(t)) \, dt,
\]
equation \eqref{eqqufiBnonlinss} rewrites in terms of the Stratonovich differential $\circ \, dB$ as
\[
d\phi(t)=-[iH+(L-(\phi(t), L \phi(t)))^2 + (\phi(t), L^2 \phi(t))-(\phi(t), L \phi(t))^2 ]\phi(t) \, dt
\]
\begin{equation}
\label{eqqufiBnonlinssStr}
+ (L-(\phi(t), L \phi(t)))\phi(t) \, \circ \, dB(t).
\end{equation}

Notice finally that equations \eqref{Lindstoch} and \eqref{Lindstochnorm1} coincide when $L$ is anti-Hermitian,
that is $L^*=-L$. However, this case is considered as much less relevant from
physical point of view (see discussion in \cite{BarchBook}).

\subsection{Belavkin's quantum filtering equations: counting case}

For the case of {\it counting observation} the {\it Belavkin quantum filtering equations
for normalised density matrices} have the form of the jump type SDE
\begin{equation}
\label{eqBeleqcountj}
d\rho(t)=(- i[H, \rho(t)] -\frac12 \sum_{j=1}^n \{L_j^*L_j,\rho(t)\}+ {\tr} (L_j\rho(t) L_j^*) \rho(t) )\, dt
+\left(\frac{L_j\rho(t) L_j^*}{{\tr} (L_j\rho(t) L_j^*)}-\rho(t)\right) dN_j(t),
\end{equation}
with the counting process $N_j(t)$ having position dependent intensity ${\tr} (L_j^*L_j\rho)$,
so that the compensated processes $M_j(t)=N_j(t)-\int_0^t {\tr} (L_j^*L_j\rho_s) \, ds$ are martingales.

As usual, we shall write these equations in a shorter form
(also omitting for brevity the explicit dependence of unknown functions on $t$):
 \begin{equation}
\label{eqBeleqcount}
d\rho=(- i[H, \rho] -\frac12 \{L^*L,\rho\}+ {\tr} (L\rho L^*) \rho ) \, dt
+\left(\frac{L\rho L^*}{{\tr} (L\rho L^*)}-\rho\right) dN(t).
\end{equation}

In terms of the compensated processes $M_j(t)$, this equation rewrites as
\begin{equation}
\label{eqBeleqcountm1}
d\rho=- i[H, \rho]\, dt +(L\rho L^*-\frac12 \{L^*L,\rho\} ) \, dt
+\left(\frac{L\rho L^*}{{\tr} (L\rho L^*)}-\rho\right) dM(t).
\end{equation}

As in the diffusive case, equation \eqref{eqBeleqcount} is (at least formally) trace preserving
only for operators of unit trace. It is often convenient to extend it to a general class of equations
that is trace preserving for all traces:

\begin{equation}
\label{eqBeleqcounttp}
d\rho=(- i[H, \rho] -\frac12 \{L^*L,\rho\}+ {\tr} (L\rho L^*) \frac{\rho}{{\tr}\, \rho} ) \, dt
+\left( L\rho L^*  \frac{{\tr}\, \rho}{{\tr} (L\rho L^*)}-\rho\right) dN(t).
\end{equation}
Equations \eqref{eqBeleqcount} and \eqref{eqBeleqcounttp} coincide for ${\tr}\,\rho=1$,
but \eqref{eqBeleqcounttp} implies $d\,{\tr}\, \rho=0$ for all trace-class $\rho$.

In equation \eqref{eqBeleqcount} (written in the most common form in the literature)
the driving noise $N(t)$ is itself
position dependent, which makes it not quite standard.
Moreover, a left continuous version of an integrand is not stressed there.
A convenient rigorous
form can be written as follows. Let $K$ be any constant such that $ \| L_j^*L_j \| <K$ for all $j$.
Let $N(dx \, dt)=(N_1(dx \, dt), \cdots, N_n(dx\, dt))$ be a collection of standard independent
Poisson random measures (or corresponding compound Poisson processes)
on $[0,K]\times \R_+$ (with Lebesgue measure
as intensity on $\R_+$). A mathematically rigorous form of \eqref{eqBeleqcount} is as follows:
\[
\rho (t)=\int_0^t (- i[H, \rho(s)] -\frac12 \{L^*L,\rho(s)\}+ {\tr} (L\rho(s) L^*) \rho(s) )\, ds
\]
\begin{equation}
\label{eqBeleqcountind}
+\int_0^t \int_0^K \left(\frac{L\rho(s-) L^*}{{\tr} (L\rho(s-) L^*)}-\rho(s-)\right)
\1(x< {\tr} (L^*L\rho(s-))) N(dx\, ds),
\end{equation}
or, in full version,
 \[
\rho (t)=\int_0^t (- i[H, \rho(s)] -\frac12 \sum_j \{L_j^*L_j,\rho(s)\}
+ \sum_j {\tr} (L_j\rho(s) L_j^*) \rho(s) )\, ds
\]
\[
+\int_0^t \sum_j \int_0^K \left(\frac{L_j\rho(s-) L_j^*}{{\tr} (L_j\rho(s-) L_j^*)}-\rho(s-)\right)
\1(x< {\tr} (L_j^*L_j\rho(s-))) N_j(dx\, ds).
\]

The counting process $N(t)$ from \eqref{eqBeleqcount} is recovered by the formula
\[
N(t)=\int_0^t \int_0^K \1(x< {\tr} (L^*L\rho(s-))) N(dx \, ds).
\]
Talking about equation \eqref{eqBeleqcount}, we shall always mean its rigorous version \eqref{eqBeleqcountind}.


As in the case of diffusive measurements, dynamics \eqref{eqBeleqcountind} preserves the set of pure states.
Namely, if $\phi(t)$ satisfy the equation
\begin{equation}
\label{eqcountpure}
d \phi(t)=-\left(iH +\frac12 (L^*L-\|L\phi(t)\|^2)\right)\phi(t) \, dt
+\left(\frac{L\phi(t)}{\|L\phi(t)\|}-\phi(t)\right)  dN(t),
\end{equation}
then $\rho_t=\phi_t\otimes \bar \phi_t$ satisfies equation \eqref{eqBeleqcountm1}.
In fact, $\|L\phi\|^2={\tr}\, (L\rho L^*)$, and,
 by Ito's product rule $dN(t)dN(t)=dN(t)$, it follows that (omitting argument $t$ for brevity)
\[
d\rho=-i[H,\rho] \, dt
-\frac12 (L^*L-\|L\phi\|^2)\phi \otimes \bar \phi\, dt
- \frac12 \phi\otimes  (\overline{L^*L}-\|L\phi\|^2)\bar \phi\, dt
\]
\[
+\left(\frac{L\phi}{\|L\phi\|}-\phi\right) \otimes \bar \phi \, dN(t)
+\phi \otimes \left(\frac{\overline{L\phi}}{\|L\phi\|}-\bar \phi\right) \, dN(t)
\]
\[
+\left(\frac{L\phi}{\|L\phi\|}-\phi\right)
\times \left(\frac{\overline{L\phi}}{\|L\phi\|}-\bar \phi\right)\, dN(t)
\]
\[
=-i[H,\rho] \, dt-\frac12 \{L^*L,\rho\} \, dt
+{\tr}\, (L\rho L^*) \rho  \, dt
+ \left(\frac{L\rho L^*}{{\tr} (L\rho L^*)}-\rho\right) dN(t),
\]
as claimed.

\begin{remark}For readers not well acquainted with stochastic calculus of jump process,
 let us give some intuition for the product rule $dN(t)\,dN(t)=dN(t)$ for Poisson processes
 with any (even variable) intensity. If processes $X_i(t)$, $i=1,2$, are defined by
 their differentials as $dX_i(t)=(F_i(X_i(t-))-X_i(t-)) dN(t)$, this means that at times $\tau_j$
 of the jumps of $N(t)$ they jump from $X_i(\tau_j-)$ to
 \[
 F_i(X_i(\tau_j-))= X_i(\tau_j-)+[F_i(X_i(\tau_j-))-X_i(\tau_j-)].
 \]
 Then $X_1(t)X_2(t)$ at these moments jumps from $X_1(\tau_j-)X_2(\tau_j-)$
 to $F_1(X_1(\tau_j-))F_2(X_2(\tau_j-))$, and thus this process has to have the differential
 \begin{equation}
 \label{eqprodjump}
 d(X_1(t)X_2(t))=\left( F_1(X_1(\tau_j-))F_2(X_2(\tau_j-))-X_1(\tau_j-)X_2(\tau_j-)\right) dN(t).
 \end{equation}
 And that is what we get from Ito's product rule. In fact, it implies that
 \[
 dX_1(t)dX_2(t)=X_1(t) \, dX_2(t)+dX_2(t) \, dX_1(t)+dX_1(t)\, dX_2(t),
 \]
 which equals, by the rule $dN(t)\,dN(t)=dN(t)$, to
  \[
\bigl[ X_1(\tau_j-)(F_2(X_2(\tau_j-))-X_2(\tau_j-))+X_2(\tau_j-) (F_1(X_1(\tau_j-))-X_1(\tau_j-))
 \]
 \[
 +(F_1(X_1(\tau_j-))-X_1(\tau_j-)) (F_2(X_2(\tau_j-))-X_2(\tau_j-))\bigr] dN(t),
 \]
 that is, exactly \eqref{eqprodjump}.
\end{remark}

Again analogously to the diffusive case, equation \eqref{eqBeleqcount} can be obtained
by normalization from an apparently linear equation. Namely, one checks by a straightforward
application of Ito's formula that if $\ga(t)$ satisfy the equation (omitting argument $t$ for brevity)
 \begin{equation}
\label{eqBeleqcountlin}
d\ga=(- i[H, \ga] -\frac12 \{L^*L,\ga\}+ L\ga L^* ) dt
+(L\ga L^*-\ga) (dN(t)-dt),
\end{equation}
then ${\tr \, \ga}$ satisfies the equations
\[
d\, {\tr}\, \ga={\tr} (L\ga L^*-\ga) \, (dN(t)-dt),
\]
or
\[
d\, {\tr}\, \ga=\sum_j {\tr} (L_j\ga L_j^*-\ga) \, (dN_j(t)-dt)
\]
in the full version, and
\[
d\frac{1}{{\tr}\, \ga}=\left(\frac{1}{{\tr}\,(L\ga L^*)}-\frac{1}{{\tr}\, \ga }\right)
\left(dN(t)-\frac{{\tr}\, (L\ga L^*)}{{\tr}\, \ga} dt\right),
\]
implying that $\ga/{\tr}\,\ga$ satisfies \eqref{eqBeleqcount}.

\begin{remark}
Unlike diffusive case, equation \eqref{eqBeleqcountlin} is not exactly linear,
as the intensity of $N(t)$ depends on $\ga$.
\end{remark}

Again pure states are preserved under \eqref{eqBeleqcountlin}.
In fact, if vectors $\chi(t)$ in $\HC$ satisfy the equation
\begin{equation}
\label{eqcountpurelinpu}
d \chi(t)=-\left(iH +\frac12 (L^*L-1)\right)\chi(t) \, dt
+(L-1)\chi(t)  dN(t),
\end{equation}
then $\ga(t)=\chi(t)\otimes \bar \chi(t)$ satisfy \eqref{eqBeleqcountlin}
 and $\phi(t)=\chi(t)/\|\chi(t)\|$ satisfy  \eqref{eqcountpure}.

Let us note finally that in the important particular case when operator $L$ is unitary,
equation \eqref{eqBeleqcount} becomes linear:

 \begin{equation}
\label{eqBeleqcountunit}
d\rho=- i[H, \rho]  \, dt
+\left(L\rho L^*-\rho\right) dN(t),
\end{equation}
with the standard Poisson process $N(t)$ of unit intensity, in which case it coincides, of course,
with \eqref{eqBeleqcountlin} and writes down in terms of the compensated martingale $M(t)=N(t)-t$ as
 \begin{equation}
\label{eqBeleqcountunit1}
d\rho=(- i[H, \rho] + L\rho L^*-\rho ) \, dt
+\left(L\rho L^*-\rho\right) dM(t).
\end{equation}

\subsection{SDEs and PDEs}

Apart from SDEs, the standard alternative way of describing Markov processes
is via PDEs (partial differential equations) or sometimes $\Psi$DEs (pseudo-differential equations).
  These two descriptions are linked via Ito's
formula. Namely, if $X(t)$ is a (time homogeneous) Markov process in some metric space $S$,
starting with $x=X(0)$ and solving
certain SDE, one defines the corresponding Markov transition operators  $T_tf(x)=\E_x f(X(t))$ on an
appropriate space of bounded functions $f$ ($\E_x$ is used to denote the
expectation with respect to the process started in $x$).
When $T_t$ preserves the space of continuous
functions (as will be the case for all examples of this paper), the process and the semigroup
are usually referred to as $C$-Feller. In this case the operators $T_t$ form a semigroup of positive
contractions in the space $C(S)$. By a generator
of the process $X(t)$ one understands the operator $Af=\lim_{t\to 0} (T_tf-f)/t$, defined on the set of functions,
where this limit exists in some sense (usually in the topology of $C(S)$).
It turns out that, under appropriate regularity assumptions,
the functions $T_tf(x)$ solve the PDE $\pa f/\pa t=Af$. If $A$ is a second order differential operator,
it is referred to as a diffusion operator and the corresponding process is called a diffusion.

For handy future references, let us write now the formal expressions for the generators of all processes
defined by the basic quantum filtering SDEs. They are direct consequences of the corresponding SDEs
and of the formal application of Ito's formula. Of course, only the well-posedness theory for all these equations,
developed further, will make these derivations rigorous.

(1) If $T_tf(\ga)=\E_{\ga} f(\Ga(t))$, where $\Ga(t)$ yields the solution
to the Cauchy problem of SDE \eqref{Lindstoch} with the initial condition $\ga$,
then (for sufficiently smooth functions $f$ on the set of density operators)
the function $T_tf(\ga)=f(t,\ga)$ solves the Cauchy problem
\begin{equation}
\label{eqCauchylinfilt}
\frac{\pa f}{\pa t}=\AC_{lin} f, \quad f(0,\ga)=f(\ga),
\end{equation}
where
\begin{equation}
\label{eqgenlinfilt1}
\AC_{lin} f(\ga)= (f'(\ga), -i[H,\ga] +\LC_L \ga) + \sum_j \frac12 (L_j\ga+\ga L_j^*, f''(\ga) (L_j\ga+\ga L_j^*)),
\end{equation}
or in our usual concise notation
\begin{equation}
\label{eqgenlinfilt}
\AC_{lin} f(\ga)= (f'(\ga), -i[H,\ga] +\LC_L \ga) + \frac12 (L\ga+\ga L^*, f''(\ga) (L\ga+\ga L^*)).
\end{equation}
Therefore, this $\AC_{lin}$ is the generator of the diffusion $\Ga(t)$.

(2) Solutions to SDEs  \eqref{Lindstochnorm1} specify the Markov process on the convex
set of density operators  $S(\HC)$, with the generator
being the diffusion operator
\[
\AC f(\rho)=(f'(\rho),-i[H, \rho] -\frac12 \{L^*L,\rho\}+ L\rho L^*)
\]
\begin{equation}
\label{eqdifgener}
+\frac12 \bigl(\rho L^* + L\rho-{\tr}\,(\rho L^* +L\rho)  \rho,
f''(\rho) (\rho L^* + L\rho-{\tr}\,(\rho L^* +L\rho)  \rho)\bigr).
\end{equation}

(3) Solutions to \eqref{eqBeleqcountlin} specify the Markov process with the generator
\[
\AC_{count,lin}f(\rho)=-(f'(\rho), i[H, \rho]
+\sum_{j=1}^n \left(\frac12 \{L_j^*L_j,\rho\}+ \rho\right)
\]
 \begin{equation}
\label{eqjumpgenerlin}
+ \sum_{j=1}^n {\tr}\, (L_j^*L_j\rho)\left[f(L_j\rho L_j^*)-f(\rho)\right].
\end{equation}

(4) Solutions to SDE \eqref{eqBeleqcountind}
specify a Markov process with the generator
\[
\AC_{count}f(\rho)=-(f'(\rho), i[H, \rho] +\sum_{j=1}^n (\frac12 \{L_j^*L_j,\rho\}-\rho \, {\tr}\, (L_j^*L_j\rho))
\]
\begin{equation}
\label{eqjumpgenerj}
+\sum_j {\tr}\, (L_j^*L_j\rho) \left[f\left(\frac{L_j\rho L_j^*}{{\tr}\, (L_j^*L_j\rho)}\right)-f(\rho)\right],
\end{equation}
or, in concise notations,
 \[
\AC_{count}f(\rho)=-(f'(\rho), i[H, \rho] + \frac12 \{L^*L,\rho\}-\rho \, {\tr}\, (L^*L\rho))
\]
 \begin{equation}
\label{eqjumpgener}
+ {\tr}\, (L^*L\rho)\left[f\left(\frac{L\rho L^*}{{\tr}\, (L^*L\rho)}\right)-f(\rho)\right].
\end{equation}

\subsection{Quantum filtering for mixed channels}

Though one considers usually either counting or diffusive observations, and we shall do so most of the time,
it makes sense to organise mixed channels of observations that include both type of observations simultaneously.
Looking at operators \eqref{eqdifgener} and \eqref{eqjumpgenerj} it is easy to guess the expression for the generator
that combines both type of observation:

\[
\AC_{mix}f(\rho)
=\sum_{j=k}^{m-1} \, {\tr} \, (L_j\rho L_j^*)  \left[f(\frac{ L_j \rho  L_j^*}{ {\tr} \, (L_j\rho L_j^*)})-f(\rho)\right]
\]
\[
+\left(f'(\rho), -i[H, \rho] -\frac12 \sum_{j=0}^{m-1}\{ L_j^* L_j,\rho\}
+\sum_{j=0}^{k-1} L_j \rho  L_j^*
+\sum_{j=k}^{m-1}  \, {\tr} \, (L_j\rho L_j^*) \rho \right)
\]
\begin{equation}
\label{eqmixgener}
+\frac12 \sum_{j=0}^{k-1}  \left(\rho  L_j^* +  L_j\rho- {\tr} (\rho L_j^* +L_j \rho)\rho,
f''(\rho)  (\rho  L_j^* +  L_j\rho- {\tr} (\rho L_j^* +L_j \rho) \rho)\right).
\end{equation}

In this paper we reduce our attention to simplified mixed channels,
when all $L_j$ with $j\ge k$ are unitary,
in which case $\AC_{mix}$ becomes
\[
\AC_{mix}f(\rho)
=\sum_{j=k}^{m-1} \left[f(L_j \rho  L_j^*)-f(\rho)\right]
\]
\[
+\left(f'(\rho), -i[H, \rho] -\frac12 \sum_{j=0}^{m-1}\{ L_j^* L_j,\rho\}
+\sum_{j=0}^{k-1} L_j \rho  L_j^* +(m-k) \rho \right)
\]
\begin{equation}
\label{eqmixgenerun}
+\frac12 \sum_{j=0}^{k-1}  \left(\rho  L_j^* +  L_j\rho- {\tr} (\rho L_j^* +L_j \rho)\rho,
f''(\rho)  (\rho  L_j^* +  L_j\rho- {\tr} (\rho L_j^* +L_j \rho) \rho)\right).
\end{equation}

By Ito's formula, the corresponding SDEs yielding an alternative description
of the Markov process generated by \eqref{eqmixgenerun} have the following form
\[
d\rho(t)=\sum_{j=k}^{m-1} (L_j\rho L_j^*-\rho) \, dN_j(t)
\]
\[
+\left( -i[H, \rho] -\frac12 \sum_{j=0}^{m-1}\{ L_j^* L_j,\rho\}
+\sum_{j=0}^{k-1} L_j \rho  L_j^* +(m-k) \rho \right)\, dt
\]
\begin{equation}
\label{Lindstochmix}
+\sum_{j=0}^{k-1} [L_j\rho+\rho L_j^*-\rho\, {\tr} \, (L_j\rho(t)+\rho L_j^*) ] dB_j(t),
\end{equation}
where $N_j(t)$, $j=k, \cdots, m-1$, are independent standard Poisson processes of unit intensity
and $B_j(t)$, $j=0,\cdots, k-1$, are independent (and independent from all $N_j$) standard Brownian motions.

A rigorous construction of the processes generated by
\eqref{eqmixgener}, \eqref{eqjumpgener} , \eqref{eqdifgener}
will be given in Section \ref{secwellpos},
 and their derivation from basic quantum mechanics postulates in Section \ref{secder}.

\subsection{Interaction picture and mild equations}
\label{secinter}

As was mentioned, in this survey paper we decided to reduce the attention mostly
 to the case of bounded coupling operator $L$,
which will be assumed throughout if not stated otherwise.

For the case of bounded $L$ and unbounded $H$ there exists a standard trick,
the interaction representation,
that allows one to always reduce the analysis to bounded (even vanishing) Hamiltonians, though by moving
to time dependent $L$.  Namely, in terms of the new states $\nu(t)=e^{iHt}\ga(t) e^{-iHt}$,
equation \eqref{Lindstoch} rewrites equivalently as the equation
\begin{equation}
\label{Lindstochin}
d\nu(t)=\LC_{L^{Ht}} \nu(t) \, dt +(L^{Ht}\nu(t)+\nu(t) (L^{Ht})^*) dY(t),
\end{equation}
where $L^{Ht}=e^{iHt} L e^{-iHt}$. Similarly, equation \eqref{Lindstochnorm1} rewrites in terms of
$\mu(t)=e^{iHt}\rho(t) e^{-iHt}$ as the equation
\begin{equation}
\label{Lindstochnormin}
d\mu(t)=\LC_{L^{Ht}} \mu (t)\, dt
+[L^{Ht}\mu(t)+\mu(t) (L^{Ht})^*-\mu(t)\, {\tr} \, (L^{Ht}\mu(t)+\mu(t) (L^{Ht})^*) ] dB(t).
\end{equation}

 This rewriting is of course obtained by the trivial application
of Ito's formula that is valid for bounded $H$ and $L$. But
equations  \eqref{Lindstochin} and \eqref{Lindstochnormin} can be looked at as particular cases
of  \eqref{Lindstoch} and \eqref{Lindstochnorm1} with bounded (vanishing ) $H$,
though with time-dependent coefficients even when the initial equation included unbounded $H$.

In future, we shall speak about equation \eqref{Lindstoch} or \eqref{Lindstochnorm1} generally,
though having in mind this exact equation only for bounded $H$ and equation \eqref{Lindstochin}
or \eqref{Lindstochnormin} in general case. Though after this reduction the coupling operator becomes
time-dependent, we shall not usually pay much attention to this detail, because, as mentioned
above in Remark \ref{remtimedep}, all our results modify automatically to treat such time dependence.

Interaction representation is convenient for abstract analysis. For concrete equations
of classical quantum mechanics, it is often handy to use another trick to get rid of unbounded
Hamiltonians, namely, to use the so-called mild (integral) form of stochastic equations:
\[
\ga(t)=e^{-iHt}\ga_0e^{iHt}
+\int_0^t e^{-iH(t-s)} \LC_L\ga (s)  e^{iH(t-s)} \, ds
\]
\begin{equation}
\label{Lindstochmild}
 +\int_0^t e^{-iH(t-s)} (L\ga(s)+\ga(s) L^*) e^{iH(t-s)}\, dY(s)
\end{equation}
for equation \eqref{Lindstoch} and
\[
\rho(t)=e^{-iHt}\rho_0 e^{iHt}
+\int_0^t e^{-iH(t-s)} \LC_L\rho (s)  e^{iH(t-s)} \, ds
\]
\begin{equation}
\label{Lindstochnormmild}
+\int_0^t e^{-iH(t-s)} \bigl[
L\rho(s)+\rho(s) L^*-\rho(s)\, {\tr} \, (L\rho(s)+\rho(s) L^*) \bigr] e^{iH(t-s)} dB(s)
\end{equation}
for equation \eqref{Lindstochnorm1}.

The equivalence of \eqref{Lindstochin} (resp. \eqref{Lindstochnormin}) and \eqref{Lindstochmild}
(resp. \eqref{Lindstochnormmild}) is straightforward again by formal application of Ito's formula.

For the corresponding pure states the mild versions
 of linear and nonlinear quantum filtering equations have the form
\begin{equation}
\label{eqqufiBlinsm}
\chi(t) =e^{-iH t}\chi_0+\int_0^t e^{-iH(t-s)} [-\frac12 L^*L \chi(s) \,dt+L \chi(s) \, dY(s)],
\end{equation}
and, respectively,
\[
\phi(t)=e^{-iH t}\phi_0+\int_0^t e^{-iH (t-s)} \bigl[ i\langle L_S \rangle_{\phi(s)} L_A \phi(s) \, ds
-\frac12 (L-\langle L_S \rangle_{\phi(s)})^*(L-\langle L_S \rangle_{\phi(s)})\phi(s) \, ds
\]
\begin{equation}
\label{eqqufiBnonlinsm}
+ (L-\langle L_S \rangle_{\phi(s)})\phi(s) \, dB(s)\bigr].
\end{equation}

The same tricks, interaction picture and mild forms,
allows one to reduce general filtering equations for counting observation
to the case of bounded $H$. For instance, equation \eqref{eqBeleqcount} in terms of the states
$\mu(t)=e^{iHt}\rho(t) e^{-iHt}$ rewrites as the equation in interaction picture
 \begin{equation}
\label{eqBeleqcountin}
d\rho=(-\frac12 \{(L^{Ht})^*L^{Ht},\rho\}+ {\tr} (L^{Ht}\rho (L^{Ht})^*) \rho ) \, dt
+\left(\frac{L^{Ht}\rho (L^{Ht})^*}{{\tr} (L^{Ht}\rho (L^{Ht})^*)}-\rho\right) dN(t),
\end{equation}
with the same $L^{Ht}=e^{iHt} L e^{-iHt}$ as above.

\section{Well-posedness of quantum filtering equations}
\label{secwellpos}

This section proves the well-posedness of the main filtering equations.
As a warm-up, we start in subsection \ref{puredif} with a review of known results on filtering SDEs
 for pure states in diffusive case. The study of these equations has a long history,
 where one can distinguish a study of weak and strong solutions. For the former case
 the most advanced final results were obtained in \cite{MoraRebo}, \cite{Mora13} and for the latter
 in \cite{Kol25a}, \cite{Kol25b}, see further references at the end of subsection \ref{puredif}.

In subsection \ref{subwellcount} we present new results for the case of counting observations
proving the well-posedess of both stochastic SDEs and the corresponding
Kolmogorov backward equations identifying also appropriate cores for
 the generators of the corresponding Markov processes on $S(\HC)$.
In subsections \ref{seclineq} and \ref{seclineqpde} we deal with linear filtering equations
for mixed states in diffusive case. Subsection \ref{seclineq} presents essentially the results
from  \cite{Kol25a} on the well-posedness of these equations, and subsection \ref{seclineqpde}
presents new results on the properties of the corresponding Markov semigroup.
Subsection \ref{secnormeq} provides the results from author's paper \cite{Kol25a} on the well-posedness
of normalised filtering SDEs of diffusive case. Two main ideas are used for the proof:
(1) to deal with weak solutions we use the bijection between linear an nonlinear equations
(which traditionally was used only in one direction, namely, from linear to nonlinear equations
see e.g. \cite{BarchHol} and \cite{Mora13}); (2) the reduction of  mixed-state equations to a
pure-state equations in an enlarged Hilbert space.

Subsection \ref{semnonlinfilt} contains new results
on the regularity properties of Markov semigroups specified by the filtering SDEs with
mixed diffusive and counting observations including identifying
an appropriate invariant core for its generator. These results are crucial for the derivation
of filtering SDEs addressed in the next section. A short subsection  \ref{classham}
deals with well-posedness and regularity results for the solutions of filtering SDEs
with the classical Hamiltonian $-\De +V(x)$. These results are possibly new,
but they are very simple and should be more or less obvious for specialists.
These results are needed for further treatment of the corresponding mean-field equations.
The last subsection is an appendix providing some simple trace inequalities used in the main text.

\subsection{Preliminaries: equations for pure states}
\label{puredif}

As a preliminary step in our analysis of filtering equations, we shall briefly look at
filtering equations for pure states reducing attention to the diffusive case only.

If $H$ is bounded,
equations \eqref{eqqufiBlins} are clearly well-posed by the standard Ito theory
in Hilbert spaces, as they are
linear equations with bounded coefficients.

If $H$ is unbounded, we can use the same trick as for the mixed state equation,
namely, the interaction representation, to get rid of unbounded $H$.
It is seen directly that in terms of vectors
$\xi(t)=e^{iHt} \chi(t)$ and $\psi(t)=e^{iHt} \phi(t)$,
equations \eqref{eqqufiBlins} and \eqref{eqqufiBnonlins},
are equivalent (for bounded $H$) to the SDEs
\begin{equation}
\label{eqqufiBlinsin}
d\xi(t) = -\frac12 (L^{Ht})^*L^{Ht} \xi(t) \,dt+L^{Ht} \xi(t) dY(t),
\end{equation}
and, respectively,
\[
d\psi(t)= i\langle L^{Ht}_S \rangle_{\psi(t)} L^{Ht}_A \psi(t) \, dt
-\frac12 (L^{Ht}-\langle L^{Ht}_S \rangle_{\psi(t)})^*(L^{Ht}-\langle L^{Ht}_S \rangle_{\psi(t)})]\psi(t) \, dt
\]
\begin{equation}
\label{eqqufiBnonlinsin}
+ (L^{Ht}-\langle L^{Ht}_S \rangle_{\psi(t)})\psi(t) \, dB(t),
\end{equation}
with $L^{Ht}=e^{iHt} L e^{-iHt}$.

These SDEs in the interaction form have the same structure as the initial equations, the only difference being
the time dependence of $L$. As for mixed states,
we will thus speak about the well-posedness of linear equations \eqref{eqqufiBlins} meaning its interaction version
\eqref{eqqufiBlinsin} in case of unbounded $H$.

As mentioned above, solution to nonlinear equation  \eqref{eqqufiBnonlinsn} can be obtained by the normalisation
  from a solution of the linear one, which yields, in particular, the existence of a weak solution to this equation.
  On the other hand, the well-posedness of  \eqref{eqqufiBnonlinsn}  was not an obvious task to achieve
  (it was noted in 2008, see\cite{MoraRebo}, that the uniqueness, even of a weak solution,
was still open for bounded $H,L$ in infinite-dimensional spaces). However, as was noted in \cite{KolQuantLLN},
the  well-posedness even in the strong probabilistic sense follows from the standard well-posedness result
  for stochastic Ito's equation with Lipschitz coefficients, due to
  the following observation.

\begin{lemma}
 \label{lemboundder}
 For any bounded operator $M$ in $\HC$, the mapping
 $\HC \to \HC$ given by formulas
 \[
 \psi \to f(\psi)= \langle M \rangle_{\psi} \psi
\]
is continuously Fr\'echet differentiable outside  $\psi= 0$, with the derivative mapping
\[
Df(\psi)[\phi]=(Df(\psi),\phi)
=[(M\hat \psi, \phi)+(M^*\hat \psi, \phi)] \hat \psi
 +(\hat \psi, M\hat \psi)\left[\phi-2\hat \psi \,Re\, (\hat \psi, \phi)\right],
 \]
 where $\hat \psi=\psi/\|\psi\|$. The linear operator $\phi\mapsto Df(\psi)[\phi]$
has the norm not exceeding $3\|M\|$.
In particular, the mapping $f$ is globally Lipschitz with Lipschitz constant $3\|M\|$.
\end{lemma}

\begin{proof}
We have:
\[
\left.\frac{d}{dt}\right|_{t=0} f(\psi+t\phi)
=\left.\frac{d}{dt}\right|_{t=0} \frac{(\psi+t\phi, M(\psi+t\phi))}{(\psi+t\phi, \psi+t\phi)} (\psi+t\phi)
\]
\[
=\left.\frac{d}{dt}\right|_{t=0}\bigl[((\psi,M\psi)+t(\phi, M\psi) +t(\psi, M\phi))
\left(1-\frac{2t \, Re\, (\phi,\psi)}{(\psi,\psi)}\right)\frac{\psi+t\phi}{(\psi,\psi)}\bigr]
=Df(\psi)[\phi],
\]
implying Gateaux differentiability of required form. Since
\[
\|\phi-2\hat \psi \,Re\, (\hat \psi, \phi)\|=\|\phi\|
 \]
(as one sees easily by choosing coordinate  with $\hat \psi$ a coordinate vector),
it follows that the norm of $Df(\psi)$ is bounded by $3\|M\|$. We complete the proof by noting
the obvious continuity of the mapping  $\psi\mapsto Df(\psi)$ from $\HC$ to $\LC(\HC)$,
since the continuity of a Gateaux derivative is known to imply that this derivative is Fr\'echet,
see e.g. Theorem 45 from \cite{Toolbox}.
\end{proof}

\begin{remark} Lipschitz continuity of the function $f$ was proved in \cite{BarchBook}
 for finite-dimensional $\HC$ with a bit more lengthier
calculations. Simple general proof via differentiability is taken from \cite{Kol25a}.
\end{remark}

The link between the two descriptions (linear and normalised) of quantum filtering
is summarised in the following statement.

\begin{prop}
\label{linnorrmpure}

(i) If $\chi(t)$ satisfies \eqref{eqqufiBlins},
then $\|\chi(t)\|^2$ satisfies the equations
 \begin{equation}
\label{chisquare}
d\|\chi(t)\|^2=2(\chi(t),L_S \chi(t))dY(t),
\end{equation}
 \begin{equation}
\label{chisquare1}
d\frac{1}{\|\chi(t)\|^2}=-\frac{2}{\|\chi(t)\|^2}
 \frac{(\chi(t),L_S \chi(t))}{\|\chi(t)\|^2}
\left[dY(t)- \frac{(\chi(t),L_S \chi(t))}{\|\chi(t)\|^2} dt\right],
\end{equation}
and the normalised states $\phi(t)=\chi(t)/\|\chi(t)\|$ satisfy the
nonlinear equation \eqref{eqqufiBnonlins}, where
\begin{equation}
\label{eqdefinnov1}
dB(t)=dY(t)-2(\phi(t), L_S \phi(t)) \, dt.
\end{equation}

(ii) Let $\phi(t)$ have unit norms for all $t$ and satisfy the
nonlinear equation \eqref{eqqufiBnonlins}. Define $\|\chi(t)\|^{-2}$ as the solution
(with the initial condition equal to $1$) to the equation
 \begin{equation}
\label{chisquare4}
d\frac{1}{\|\chi(t)\|^2}=-\frac{2}{\|\chi(t)\|^2}
(\phi(t),L_S \phi(t)) dB(t),
\end{equation}
which is seen to be identical with \eqref{chisquare1}
when $B$ and $Y$ are linked via \eqref{eqdefinnov} and $\phi(t)=\chi(t)/\|\chi(t)\|$.
 Then the vectors
$\chi(t)=\phi(t) \|\chi(t)\|$ satisfy the linear equation \eqref{eqqufiBlins}.

(iii) The square norm $\|\chi(t)\|^2$ (respectively its inverse) of a solution
to \eqref{eqqufiBlins} is a positive martingale under
the probability law where $Y(t)$ (resp. $B(t)$) is a Brownian motion. If $B(t)$ is
a Brownian motion and $Y$ is given by \eqref{eqdefinnov}, then
 \begin{equation}
\label{chisquare5}
\E \|\chi(t)\|^2 \le \exp\{2 t \|L\|^2 \} \|\chi_0\|^2,
\end{equation}
where the expectation is with respect to $B(t)$.
\end{prop}

\begin{proof} All formal manipulations are straightforward applications of Ito's lemma and we omit them.
Since well-posedness of all equations is established, these formal manipulations are fully
justified proving (i) and (ii).
Of course, for unbounded $H$ the calculations are performed in the interaction representation.
To prove (iii) we observe from \eqref{chisquare} that $\|\chi\|^2$ is an exponential local martingale solving
the linear equation $d\|\chi\|^2=\|\chi\|^2 dM(t)$ with the martingale
\[
 M(t)=2\int_0^t \left(\frac{\chi(t)}{\|\chi (t)\|},L_S \frac{\chi(t)}{\|\chi (t)\|}\right)dY_j(t).
 \]
 If $Y(t)$ is a BM, this martingale has bounded quadratic variation. Therefore $\|\chi(t)\|^2$ is a true
 (almost sure positive) martingale according to the standard Novikov criterium.
 The same argument applies to $\|\chi(t)\|^{-2}$. Finally, from \eqref{eqdefinnov}, \eqref{chisquare}, it follows that
 \[
 \frac{d}{dt} \E \|\chi(t)\|^2 \le 2\|L\|^2 \, \|\chi(t)\|^2
 \]
 and  \eqref{chisquare5} follows by Gronwall's lemma.
\end{proof}

Let us note that much
more general linear equations in Hilbert spaces preserving the expectation of norm squared,
with bounded coefficients, are treated in detail in \cite{BarchHol}. Nonlinear equations are
more subtle. The well-posedness in strong probabilistic sense for general bounded $L$ was
first obtained in \cite{KolQuantLLN} and for unbounded $L$ in \cite{Kol25b}. In \cite{BarchHol}
 the existence of solutions in a rather general case was proved.
In \cite{MoraRebo} the well-posedness in the weak probabilistic sense was obtained
for unbounded $H,L$ under standard assumptions inspired by the works on the
conservativity of quantum dynamic semigroups from \cite{ChebFagn} and \cite{ChebQuez}.
For other results on the solutions to some generalised
versions of equations \eqref{eqqufiBnonlins} with unbounded $H$ and $L$
under various nontrivial assumptions we can refer also to  \cite{Holevo91}, \cite{Holevo96}, \cite{Fagnola},
\cite{Mora13}. For other classes  of stochastic Schr\"odinger equation we can refer to \cite{BarbRock16},
\cite{GreckGenerNonlinSchr} and references therein.

\subsection{Counting measurements}
\label{subwellcount}

Turning to our main topic, mixed states,
we start with the apparently simpler case of counting observation. The
stochastic calculus of jump processes in Banach spaces is more or less straightforward,
because all stochastic integrals can be understood
in the classical Lebesgue sense.
We shall omit details concerning simpler equation \eqref{eqBeleqcountlin}
and will concentrate on the most important (normalised) version \eqref{eqBeleqcount}
considering it in the (physically natural) Banach space
$\HC^1_s$ of self-adjoint trace-class operators.

As a tool, we will need the following counterpart of Lemma \ref{lemboundder} for density operators.

\begin{lemma}
 \label{lemboundderden}
 (i) For any bounded operator $M$ in $\HC$, the mapping
 $\HC^1 \to \HC^1$ given by the formula
 \[
 \ga \to f(\ga)= \frac{{\tr}\, (\ga M)}{{\tr}\, \ga} \ga
\]
is Fr\'echet differentiable outside the hyperplane ${\tr}\, \ga =0$,  with the bounded derivative mapping
\[
Df(\ga)[\de]=(Df(\ga), \de)
=\hat \ga \, {\tr}\, (M\de) +(\de-\hat \ga \,{\tr}\, \de) \,{\tr}\, (\hat \ga M),
 \]
 where $\hat \ga=\ga/{\tr} \, \ga$. (ii) For positive $\ga$, the linear operator $\de\mapsto Df(\ga)[\de]$
in $\HC^1$ has the norm not exceeding $3\|M\|$.
In particular, the mapping $f$ is globally Lipschitz on $S(\HC)$. (iii) The mapping $f(\ga)$ is infinitely differentiable
outside the hyperplane ${\tr}\, \ga =0$, with the derivative of any order bounded on the set $S(\HC))$.
\end{lemma}

\begin{proof}
(i) We have:
\[
\left.\frac{d}{dt}\right|_{t=0} f(\ga+t\de)
=\left.\frac{d}{dt}\right|_{t=0}
\bigl[\frac{{\tr}\, (\ga M)+t\, {\tr}\, (\de M)}{{\tr}\, \ga +t\, {\tr}\, \de}
 (\ga+t\de)\bigr]
\]
\[
=\left.\frac{d}{dt}\right|_{t=0}\bigl[ ({\tr}\, (\ga M)+t\, {\tr}\, (\de M))\frac{1}{{\tr}\, \ga}
\left(1-t \frac{{\tr}\, \de}{{\tr}\, \ga}\right)(\ga +t \de)\bigr]
=(Df(\ga),\de),
\]
as claimed. The differentiability in the sense of Fr\'echet follows from the continuity,
as in Lemma  \ref{lemboundder} above. (ii) It follows from the well known inequality
$\|\rho M\|_{\HC^1}\le \|M\| \, \|\rho\|_{\HC^1}$ for any operators $\rho, M$. (iii)
It is seen that, when differentiating, we always get in the denominator
the powers of ${\tr}\, \ga$.
\end{proof}

Let us stress again that, if not stated otherwise, we always assume that $L$ is bounded. Moreover,
usually we perform calculations for bounded $H$ having in mind that this can always be achieved
via the interaction representation, see Section \ref{secinter}, for an arbitrary self-adjoint Hamiltonian.

\begin{theorem}
\label{JumpWellSDE}
SDE \eqref{eqBeleqcount}, or more precisely \eqref{eqBeleqcountind},
is globally well-posed in $S(\HC)$, so that for any $\rho_0\in S(\HC)$ there exists
a unique solution $\rho(t)\in S(\HC)$ defined for all $t>0$.
\end{theorem}

\begin{proof}
Let us look first at the deterministic part of evolution \eqref{eqBeleqcountind}:
$\dot \rho= b(\rho)$ with
\[
b(\rho)= -i[H, \rho] -\frac12 \{L^*L,\rho\}+ {\tr} (L^*L\rho) \rho.
\]
Since this evolution is explicitly trace preserving only for matrices of unit trace, it is convenient
to look at a modified version of this equation (deterministic part of \eqref{eqBeleqcounttp}):
\[
\dot \rho=\tilde b(\rho)=-i[H, \rho] -\frac12 \{L^*L,\rho\}+ {\tr} (L^*L\rho) \frac{\rho}{{\tr}\, \rho}.
\]
It is seen that the equations with $b(\rho)$ and $\tilde b(\rho)$ coincide
on operators of unit trace, but the equation with $\tilde b(\rho)$ is explicitly
trace preserving for arbitrary trace-class matrices. What is even more important,
according to Lemma \ref{lemboundderden}, the function $\tilde b(\rho)$ is Lipschitz
continuous in the set of positive trace-class operators (or in any neighborhood of
this set). This implies that the ODE $\dot \rho=\tilde b(\rho)$, and hence also
$\dot \rho=b(\rho)$, are well-posed globally as long as solutions start in $S(\HC)$,
if one can show that they always remain in the closed convex set of positive (non-negative) operators.
One can show that they really do it by checking
the general criteria for solutions of ODEs staying inside a closed subset of a Banach space from papers
 \cite{Martin} and \cite{Lakshm}. However, in our case, this conclusion can be derived directly. Namely, if
 $\rho(t)$, started in $S(\HC)$ touches the boundary at some time $t$, this means that there appears
  a vector $v$ such that $(v,\rho(t) v)=0$. But then $(v, b(\rho(t)) v)=(v,\tilde b(\rho(t))v)=0$ meaning
  that this property remains for all future times $t$. Thus $(v,\rho(t)v)$ can never become negative.

\begin{remark} This argument actually shows more, than just preservation of $S(\HC)$
by the deterministic part of evolution \eqref{eqBeleqcountind}. It shows that if
we start away from the boundary of $S(\HC)$, then we shall never touch it (because
the argument above can be used also in the inverse time). And moreover,
the set (actually the subspace) of $v$ such that $(v, \rho(t) v)=0$
(and thus $\rho(t)v=0$ by positivity of $\rho$)
is invariant under the evolution $\dot \rho=b(\rho)$.
\end{remark}

Once the deterministic part of  evolution \eqref{eqBeleqcountind} is constructed,
the full evolution can be built in a standard way by interlacing: trajectory jumps
at times $\tau_j$ of the process $N_t^K$ and moves according to the deterministic
part between the jumps. The jumps are seen to preserve $S(\HC)$. Theorem is proved.
\end{proof}

\begin{theorem}
\label{JumpWellPDE}
The operator  \eqref{eqjumpgener} generates a
$C$-Feller semigroup $T_t$ in $C(S(\HC))$ with the spaces
 $C^1(S(\HC))$ and $C^2(S(\HC))$ being invariant subspaces of the domain,
 and $T_s$ are bounded in these spaces
 uniformly for $s\in [0,t]$ with any $t>0$.
\end{theorem}

\begin{proof}
We start again with the deterministic part of
operator \eqref{eqjumpgener}:
\[
\AC_{count, det}f(\rho)=(f'(\rho), b(\rho)=-(f'(\rho), i[H, \rho]
+ \frac12 \{L^*L,\rho\}-\rho \, {\tr}\, (L^*L\rho)).
\]
As was shown at the beginning of the proof of Theorem \ref{JumpWellSDE} above the ODE
$\dot \rho=b(\rho)$ is globally well-posed in $S(\HC)$. Moreover, as the coefficients
of this equation are infinitely smooth, its solutions
depend infinitely smooth on the initial condition. Hence, by the standard link between ODEs and
first order PDEs (see e.g. details of this link for Banach space in Section 2.10 \cite{Kolbook19}),
resolving operators of the Cauchy problem for equation $\dot f=\AC_{count,det}f$ define the semigroup
$T_t^{det}$ of contractions in the space $C(S(\HC))$ such that all spaces $C^k(S(\HC))$ are invariant and
$T_t$ act as a bounded semigroup in each of this spaces.  Next, it is straightforward to see that
the jump part $\AC_{jump}=\AC_{count, det}-\AC_{count}$ of the operator $\AC_{count}$ represents a
bounded operator in $C(S(\HC))$ and in each space $C^k(S(\HC))$ (we need only $k=1,2$). For instance,
the latter can be seen by finding the derivatives explicitly. For instance,
\[
D\AC_{jump}f(\rho)[\de]
= {\tr}\, (L^*L\de)\left[f(\frac{L\rho L^*}{{\tr}\, (L^*L\rho)})-f(\rho)\right]
\]
\[
+{\tr}\, (L^*L\rho)\left(Df(\frac{L\rho L^*}{{\tr}\, (L^*L\rho)})
\left[ \frac{L\de L^*}{{\tr}\, (L^*L\rho)}-\frac{L\rho L^*}{{\tr}\, (L^*L\rho)}\,
\frac{{\tr}\, (L^*L\de)}{{\tr}\, (L^*L\rho)} \right]
-Df(\rho)[\de]\right).
\]

Since $\AC_{jump}$ is bounded in all spaces $C^k(S(\HC))$, where $\AC_{count,det}$ generates a bounded semigroup,
it follows that $\AC_{count}$ generates there a bounded semigroup $T_t$ as well, due to the basic perturbation theory for
semigroups (see e.g. \cite{ReedSimon} or Sec 4.6 of \cite{Kolbook19})). It is seen directly that the operators $T_t$
in the space $C(S(\HC))$ act as positivity preserving contractions (of course, it follows also from
the fact that the operator $\AC_{jump}$ is conditionally positive).
Finally, the fact that  \eqref{eqjumpgener} generates the process given
by SDEs  \eqref{eqBeleqcount} follows from Ito's formula.
\end{proof}

For the case of finite-dimensional $\HC$ it follows that the spaces $C^k(S(\HC))$
are dense in $C(S(\HC))$ and we can conclude that $C^1(S(\HC))$ and $C^2(S(\HC))$ represent invariant
cores for \eqref{eqBeleqcount} yielding the proper analytic characterisation of the semigroup $T_t$.
In infinite-dimensional case an annoying detail does not allow for a similar conclusion.
It turns out to be a very hard problem to assess the richness of the class $C^1(S(\HC))$ of smooth functions.
Seemingly, it is not dense in $C(S(\HC))$, see lengthy discussions of similar questions in
\cite{Toolbox} and \cite{Phelps}.

Therefore, to complete the analytic characterization of  $T_t$ for infinite-dimensional $\HC$
we supplement Theorem \ref{JumpWellPDE} with the following claim:

\begin{theorem}
\label{JumpWellPDE1}
An invariant core for the generator of $T_t$ can be chosen as the space
$C^1_{Gat*}(S(\HC))$ consisting of continuous uniformly Gateaux differentiable functions on $S(\HC)$
with the mapping $\ga\mapsto f'(\ga)=Df(\ga)$
from $S(\HC)$ to $(\HC^1_s)^*$ being bounded and continuous
when $S(\HC)$ is equipped with its norm topology
and $(\HC^1_s)^*$ with its weak* topology.
\end{theorem}

\begin{proof}
The fact that $C^1_{Gat*}(S(\HC))$ is an invariant subspace of the domain of \eqref{eqBeleqcount}
is obtained as above. Notice only that the required continuity is sufficient to ensure that
$L_{count}f \in C(S(\HC))$ for any $f\in C^1_{Gat*}(S(\HC))$. The nontrivial fact is that
$C^1_{Gat*}(S(\HC))$ is dense in $C(S(\HC)$. To see that this is true we first note that,
similar to the theory of $l^1$ spaces (Theorem 5.12 of \cite{Phelps}), one can show the existence
of an equivalent Gateaux differentiable norm on $\HC^1_s$. Equivalently one can refer to a more
result, see \cite{Hayek}, that in any separable Banach space there exists a uniformly
Gateaux differentiable equivalent norm.
Next, by Proposition 2.8 of \cite{Phelps}, the derivative of such a norm is norm-weak${}^*$
continuous mapping from $\HC^1_s$ to $(\HC^1_s)^*$. Finally, by the well known theory
(see e.g. \cite{Eells}), the existence of norms of certain smooth class in a separable
Banach space implies the possibility of approximating continuous functions by the
smooth functions of the same class.
\end{proof}

If the coupling operator $L$ is unitary yielding the linear filtering equation
\eqref{eqBeleqcountunit}, a more straightforward and a more detailed analytic
characterization of the core for the generator of $T_t$ is available (yielding
cores of arbitrary smoothness classes)
by two different approaches.
The first one is given by the following.

\begin{prop}
\label{JumpWellPDE2}
If $L$ is unitary, the semigroup  $T_t$ preserves the subspace of bounded weakly* continuous functions
on $S(\HC)$, and their subspaces of smooth functions of first or second order are invariant cores
for the generator of so restricted $T_t$.
\end{prop}

\begin{proof} The preservation of weakly* continuous functions by $T_t$ follows from linearity
of \eqref{eqBeleqcountunit} (which is not at all obvious for the general case). The fact that smooth
functions are dense in the space of weakly* continuous functions follows from Stone-Weierstrasse theorem.
\end{proof}

The second method can be developed via the Hilbert-space lifting of the dynamics of \eqref{eqBeleqcountunit}.
It will be presented below for more general mixed channels.

\subsection{Diffusive case: linear version}
\label{seclineq}

The first thing to decide for dealing with the equations on mixed states is the choice
of an appropriate Banach space of operators, where the corresponding SDEs will be analysed.

For the diffusive case, we shall consider these equations in the Hilbert space $\HC^2_s$ of self-adjoint Hilbert-Schmidt
operators in $\HC$, which is a closed subspace
in the space of all Hilbert-Schmidt operators $\HC^2$ with the scalar product
${\tr} (A^*B)={\tr} (AB)$. Note that $\HC^2_s$ is a real Hilbert space
(though consisting of operators in the complex Hilbert space $\HC$).

Since we are interested in trace-class operators, a more natural
space from physical point of view  would be of course the Banach space $\HC^1_s$ of self-adjoint
trace-class operators in $\HC$ (that we used above for the counting case).
However, the classes of Banach spaces,
for which a satisfactory extension of Ito stochastic calculus was developed,
namely the so-called UMD spaces, spaces of martingale type 2 and spaces with a smooth norm
(see review \cite{VanNeerven}) do not include $\HC^1$, and therefore we work in the larger space
$\HC^2$. In this space the key linear functional of taking trace is unbounded,
and we are led to work with SDEs with singular coefficients. This complication
is the price to pay for working in a convenient Hilbert setting of $\HC^2$.

\begin{theorem}
\label{LindStochLin1}
Assume $Y(t)$ is a ($n$-dimensional) standard BM.
 Then the following claims hold:

(i) Equation \eqref{Lindstoch} is well-posed in $\HC^2_s$, that is,
it has a unique global solution for any $\ga_0\in \HC^2_s$.
These solutions have the growth estimates
 \begin{equation}
\label{Lindstochquadex1}
\E [{\tr}\, \ga^2(t)] \le {\tr}\, \ga_0^2 \exp\{ 4 t\|L\|^2\}.
\end{equation}

(ii) The solution $\ga (t)$ to \eqref{Lindstoch}
is positive-definite for all $t$ whenever $\ga_0$ so is.

(iii) If the initial condition $\ga_0$ is of trace-class,
then so is the solution $\ga(t)$, with the trace given by the formula
\begin{equation}
\label{Lindstochmart}
{\tr} \,\ga(t)={\tr} \, \ga_0+\int_0^t {\tr} (L\ga(s)+\ga(s) L^*) dY(s).
\end{equation}
Moreover, ${\tr} \, \ga(t)$ is a
square integrable martingale such that
\begin{equation}
\label{Lindstochmart1}
\E ({\tr} \,\ga(t))^2\le [({\tr} \, \ga_0^+)^2 +({\tr} \, \ga_0^-)^2]\exp\{ 4t \|L\|^2\},
\end{equation}
where $\ga_0^{\pm}$ denote positive and negative parts of $\ga_0$, and
\begin{equation}
\label{Lindstochmart2}
\E ({\tr} \,|\ga(t)|)\le {\tr} \, |\ga_0|.
\end{equation}

\end{theorem}


\begin{proof}
(i) The existence of a unique solution is straightforward, as the coefficients at $dt$ and $dY$ are
bounded linear operator of $\ga$ in $\HC^2_s$.

We shall work with equation \eqref{Lindstoch}, extension to  \eqref{Lindstochin} being automatic.

We derive from
\eqref{Lindstoch} by Ito's formula that
\begin{equation}
\label{Lindstochquadtr}
d \, {\tr}\, \ga^2={\tr} \,(2\ga L\ga L^*+ L \ga L \ga + \ga L^* \ga L^*) \, dt
+2\, {\tr} \, [(L+L^*)\ga^2] \, dY(t),
\end{equation}
and thus
\[
\E \, {\tr}\, \ga^2(t)
={\tr}\, \ga_0^2 +\E \, {\tr} \,
\int_0^t (\ga(s) L^*\ga(s) L+\ga(s) L\ga(s) L^*+ L \ga(s) L \ga(s) + \ga(s) L^* \ga(s) L^*) ds.
\]
Applying inequality \eqref{eqmyineqtr1} (proved below), we obtain
\[
\E \, {\tr}\, \ga^2(t) \le {\tr}\, \ga_0^2+ 4 \|L\|^2 \int_0^t  \E \, {\tr} \,(\ga^2(s)) \, ds,
\]
and \eqref{Lindstochquadex1} follows by By Gronwall's lemma.

 (ii) Since $\ga_0$ is a positive operator from $\HC^2_s$, it follows that
  there exists a orthonormal basis $\{e_k\}$ in $\HC$ such that $\ga_0$
 can be presented as a convergent (in $\HC^2_s$) series
 \[
 \ga_0=\sum_{k=1}^{\infty} p_k e_k\otimes \bar e_k
 \]
 with non-increasing non-negative sequence $\{p_k\}$ from $l^2$. Hence $\ga_0=\lim \ga_{0n}$
 with finite-dimensional operators
 \[
 \ga_{0n}=\sum_{k=1}^n p_k e_k\otimes \bar e_k.
 \]
By linearity (and uniqueness of solutions), the solution $\ga_n(t)$ with the initial condition $\ga_{0n}$ is the finite
convex combination of the solutions with the initial conditions $e_k\otimes \bar e_k$, the latter
being given by $e_k(t)\otimes \bar e_k(t)$ with $e_k(t)$ solving the linear filtering equation for
pure states \eqref{eqqufiBlins}, and thus being positive definite. Therefore, $\ga_n(t)$ are positive-definite.
On the other hand, by \eqref{Lindstochquadex1}, $\ga_n(t)-\ga(t)$ tend to zero, as $n\to \infty$,
because the solutions with initial condition $\ga_0-\ga_{0t}$ tend to zero. Hence all $\ga (t)$ are
also positive-definite.

(iii) First assume that $\ga_0$ is positive. We then use the approximations
$\ga_n(t)$ as defined in (ii) above.
Since all $e_k(t)\otimes \bar e_k(t)$ are positive-definite operators,
the sequence $\ga_n(t)$ is monotonically increasing in $n$.

Since
\begin{equation}
\label{tracenormsq}
{\tr} \, e_k(t)\otimes \bar e_k(t)=\|e_k(t)\|^2,
\end{equation}
it follows that, for $n>m$,
\[
{\tr}\, |\ga_n(t)-\ga_m(t)|={\tr}\, (\ga_n(t)-\ga_m(t))
=\sum_{k=m+1}^n p_k \|e_k(t)\|^2,
\]
and thus the sequence $\ga_n(t)$ converges not only in $\HC^2$, but also in
the space of trace-class operators $\HC^1$. Hence, $\ga(t)$ is of trace class and
${\tr }\, \ga(t)=\lim {\tr}\, \ga_n(t)$.

From \eqref{chisquare} and \eqref{tracenormsq} it follows that
\[
{\tr} \,\ga_n(t)={\tr} \, \ga_{0n}+\int_0^t {\tr} (L\ga_n(s)+\ga_n(s) L^*) dY(s).
\]
Passing to the limit in this equation
we obtain \eqref{Lindstochmart}.

By \eqref{Lindstochmart}, Ito's formula and the estimate
$|{\tr} (\ga L)|\le {\tr}\, \ga \,  \|L\|$, it follows that
 \[
 \E ({\tr}\, \ga(t))^2
 \le  ({\tr} \, \ga(0))^2+4 \|L\|^2 \int_0^t \E ({\tr} \, \ga(s))^2\,  ds
\]
implying \eqref{Lindstochmart1} by Gronwall's lemma.

Finally, if $\ga_0$ is not positive, we decompose $\ga_0$
as the difference of its positive and negative parts:
$\ga_0=\ga_0^+-\ga_0^-$. Applying  \eqref{Lindstochmart}
 to the corresponding solutions $\ga^{\pm}(t)$ , we get  \eqref{Lindstochmart} for $\ga(t)$
by linearity. Similarly estimates \eqref{Lindstochmart1} are obtained. It remains estimate \eqref{Lindstochmart2}.
Firstly it clearly holds with the sign of equality for positive $\ga_0$. For general Hermitian $\ga_0$,
we can write
\[
{\tr} \, |\ga(t)| \le {\tr} \, \ga^+(t) +{\tr} \, \ga^-(t).
\]
The sign of inequality is due to the fact that, though the solutions $\ga^{\pm}(t)$ are positive operators,
 they are not necessarily positive and negative parts of the operator $\ga(t)$).
Applying  \eqref{Lindstochmart2} to $\ga^{\pm}(t)$ we get \eqref{Lindstochmart2} to $\ga(t)$.
 \end{proof}

Important part of well-posedness of an equation is the continuous dependence of its solution on
 initial conditions and parameters. By linearity, continuous dependence on initial condition
 follows from \eqref{Lindstochquadex1}. Next result establishes the continuous dependence on
 a bounded part of the Hamiltonian.

\begin{theorem}
\label{LindStochLin2}
Under the assumptions of Theorem \ref{LindStochLin1},
consider two equations of type \eqref{Lindstoch}:
\begin{equation}
\label{Lindstochtwo}
d\ga(t)=-i[H,\ga] \, dt -i[H_j,\ga] \, dt+\LC_L \ga(t) \, dt +(L\ga(t)+\ga(t) L^*) dY(t),
\end{equation}
$j=1,2$, where $H_1$ and $H_2$ are two bounded self-adjoint operators in $\HC$.
Then for their solutions $\ga_j(t)$, $j=1,2$, with one and the same positive initial
condition $\ga_0$ of trace-class
we have the estimates for the deviations in the norms of $\HC^1$ and $\HC^2$:
\begin{equation}
\label{eqLindstochLin21}
\E \, {\tr}\, |\ga_1(t)-\ga_2(t)|\le 2 t \|H_2-H_1\| \, {\tr} \, \ga_0,
\end{equation}
\begin{equation}
\label{eqLindstochLin22}
\sqrt{\E \, {\tr}\, (\ga_1(t)-\ga_2(t))^2}
\le 2 t \|H_2-H_1\|\sqrt{ {\tr} \, \ga_0^2} \exp\{2t \|L\|^2\}.
\end{equation}
\end{theorem}

\begin{proof}
By subtracting the two equations we get the following:
\[
d(\ga_1(t)-\ga_2(t))=-i[H_1,\ga_1(t)-\ga_2(t)] \, dt
+\LC_L (\ga_1(t)-\ga_2(t)) \, dt +(L(\ga_1(t)-\ga_2(t))
\]
\begin{equation}
\label{Lindstochtwo1}
+(\ga_1(t)-\ga_2(t)) L^*) dY(t)
+i[H_2-H_1, \ga_2(t)] \, dt.
\end{equation}
Denoting by $\Phi_t$ the operator giving solution to the Cauchy problems for equation \eqref{Lindstochtwo}
with $j=1$, we can express solution to \eqref{Lindstochtwo1} with the vanishing
initial condition in the following standard Du Hamel form:
\begin{equation}
\label{Lindstochtwo2}
\ga_1(t)-\ga_2(t)=i\int_0^t \Phi_{t-s} [H_2-H_1, \ga_2(s)] \, ds.
\end{equation}
Hence
\[
\E \, {\tr}\, |\ga_1(t)-\ga_2(t)|\le \int_0^t \E \, {\tr} \, |\Phi_{t-s} [H_2-H_1, \ga_2(s)]| \, ds.
\]
By the chain rule,  one can insert the conditional expectation with
respect to $\FC_s$ inside the expectation on the r.h.s. of the inequality
and then apply \eqref{Lindstochmart2} leading to the estimate
\[
\E \, {\tr}\, |\ga_1(t)-\ga_2(t)|\le \int_0^t \E \, {\tr} \, | [H_2-H_1, \ga_2(s)]| \, ds
\]
\[
\le 2\|H_2-H_1\| \int_0^t \E \, {\tr} \, \ga_2(s) \, ds
\le 2 t \|H_2-H_1\|  \, {\tr} \, \ga_0,
\]
implying \eqref{eqLindstochLin21}.

Similarly, using \eqref{Lindstochtwo2}, insertion of conditional expectation
and estimates  \eqref{Lindstochquadex1} and \eqref{eqmyineqtr2} (proved below), we write
\[
\sqrt{\E \, {\tr}\, (\ga_1(t)-\ga_2(t))^2}
\le \int_0^t \sqrt{ \E \, {\tr} \, (\Phi_{t-s} [H_2-H_1, \ga_2(s)])^2} \, ds
\]
\[
\le \int_0^t \sqrt{\E \, {\tr} \, ([H_2-H_1, \ga_2(s)])^2} \exp\{2(t-s)\|L\|^2\} \, ds
\]
\[
\le 2\|H_2-H_1\| \int_0^t \sqrt{\E \, {\tr} \, (\ga_2(s))^2} \exp\{2(t-s)\|L\|^2\} \, ds
\]
\[
\le 2 \|H_2-H_1\| \sqrt{{\tr} \, \ga_0^2} \int_0^t \exp\{2t\|L\|^2\} \, ds
\]
implying \eqref{eqLindstochLin22}.
\end{proof}

\begin{remark} An alternative proof of \eqref{eqLindstochLin22} can be given by writing
down the equation for ${\tr} \, (\ga_1(t)-\ga_2(t))^2$ and then applying Gronwall's lemma.
\end{remark}

\subsection{Semigroups of linear filtering equations in diffusive case}
\label{seclineqpde}

Let us look at the analytic properties of the semigroup $T_tf(\ga)=\E f(\Ga_t \ga)$,
where $\Ga_t$ are random linear operators yielding the solution
to the Cauchy problem of SDE \eqref{Lindstoch} (given by Theorem \ref{LindStochLin1}).
Since $\Ga_t(\ga)$ are continuous, the semigroup $T_t$ is a contraction semigroup in $C(\HC^2_s)$.
It also preserves positivity, and therefore, it is a $C$-Feller semigroup.

\begin{theorem}
\label{Linsdesem}

(i) $T_t$ preserve the spaces $C^k(\HC^2_s)$ for all $k$ and act there as a bounded quasi-contraction semigroup,
that is, their norm is bounded by $e^{ct}$ with a constant $c$.

(ii) If $f\in C_{luc}(\HC^2_s)$, then $T_tf(\ga) \to f(\ga)$, as $t\to 0$,
uniformly on $\ga$ from any bounded set. The space $C_{luc}(\HC^2_s)$ is invariant under $T_t$.

(iii) The space $C_{\infty}(\HC^2_s)$ is invariant under $T_t$.
If $f\in C_{luc}(\HC^2_s)\cap C_{\infty}(\HC^2_s)$,
then $T_tf(\ga) \to f(\ga)$, as $t\to 0$, uniformly in all $\ga$.

(iv) The semigroup $T_t$ is strongly continuous in the closed subspace
$C_{luc, \infty}(\HC^2_s)= C_{luc}(\HC^2_s)\cap C_{\infty}(\HC^2_s)$
of $C(\HC^2_s)$.
\end{theorem}

\begin{proof}
(i) Assuming $f\in C^1(\HC^2_s)$, we have
\[
T_tf(\ga+\rho)=\E f(\Ga_t(\ga+\rho))=\E f(\Ga_t \ga+\Ga_t\rho)
\]
\[
=T_tf(\ga)+ \E (f'(\Ga_t\ga), \Ga_t \rho)+\E \, o(\|\Ga_t\rho)\|
=T_t f(\ga)+\E (\Ga_t^* f'(\Ga_t \ga), \rho)+o(\|\rho\|).
\]
Consequently, $(T_tf)'=\E \Ga_t^* f'(\Ga_t\ga)$ and,
by \eqref{Lindstochquadex1},
\[
\|T_tf\|_{C^1(\HC^2_s)}\le e^{2\|L\|t} \|f\|_{C^1(\HC^2_s)}
\]
implying the claim for $k=1$.
Similarly
\begin{equation}
\label{eqLinsdesem}
(T_tf)''=\E \Ga_t^* f''(\Ga_t\ga) \Ga_t
\end{equation}
and other cases follow analogously.

(ii) Let  $f\in C_{luc}(\HC^2_s)$, and let $B_R$ denote the ball of radius $R$ in $\HC^2_s$.
By \eqref{Lindstochquadex1} and \eqref{Lindstoch}, for any bounded $M\subset \HC^2_s$ and any $T$,
there exists a constant $C(T)$ such that $\E \|\Ga_t \ga \|_{\HC^2_s}\le C(T)$ and
\[
\E \|\Ga_t\ga-\ga\|_{\HC^2_s}\le t C(T),
\]
\[
\E \|\Ga_t\ga_1-\Ga_t \ga_2\|_{\HC^2_s}\le \|\ga_1-\ga_2\| C(T),
\]
for all $t\le T$ and all solutions with $\ga,\ga_1,\ga_2 \in M$.
Hence, by the Markov inequality, for any $R$,
 $\P (\Ga_t\ga \notin B_R)<C(T)/R$, and
\[
\P (\|\Ga_t\ga-\ga\|_{\HC^2_s}>\sqrt t) \le \sqrt t C(T),
\]
\[
\P (\|\Ga_t\ga_1-\Ga_t\ga_2\|_{\HC^2_s}>\sqrt{\|\ga_1-\ga_2\|}) \le \sqrt {\|\ga_1-\ga_2\|}  C(T).
\]

Next, by local continuity of $f$, there exists $\de$ such that
$\|\ga_1-\ga_2\|_{\HC^2_s}\le\de $ and $\ga_1,\ga_2\in B_R$ imply
that $|f(\ga_1)-f(\ga_2)|<1/R$. Consequently, for $\sqrt t <\de$,
\[
\sup_{\ga\in M} |\E f(\Ga_t \ga)-f(\ga)|
\le 2\|f\| (\frac{C(T)}{R}+\sqrt t C(T)) +\frac{1}{R},
\]
which can be made arbitrary small by choosing large $R$ and small $t$. This proves the first claim.

On the other hand, for $\sqrt{\|\ga_1-\ga_2\|}<\de$,
\[
|\E f(\Ga_t \ga_1)-\E f(\Ga_t \ga_2)|\le  \E |f(\Ga_t \ga_1)-f(\Ga_t \ga_2)|
\le 2\|f\| (2 \frac{C(T)}{R}+\sqrt {\|\ga_1-\ga_2\|}  C(T)) +\frac{1}{R},
\]
which can be made arbitrary small by choosing large $R$ and small $\|\ga_1-\ga_2\|$. This proves the second claim.

(iii) The invariance of $ C_{\infty}(\HC^2_s)$ follows from \eqref{Lindstochquadex1}.
If $f\in C_{\infty}(\HC^2_s)\cap C_{luc}(\HC^2_s)$, then, using \eqref{Lindstochquadex1},
we can make all values $T_tf(\ga)$ small for small enough $t$ and $\|\ga\|_{\HC^2_s}>R$
with sufficiently large $R$, and then apply (ii) for $M=B_R$.

(iv) It is clear that $C_{luc}(\HC^2_s)\cap C_{\infty}(\HC^2_s)$
is closed in $C(\HC^2_s)$. The strong continuity is a consequence of (iii).
\end{proof}

Knowing the conservation of smoothness by $T_t$, we can justify the application of Ito's formula,
and conclude rigorously
that for any $f\in C^2(\HC^2_s)$,
the function $T_tf(\ga)=f(t,\ga)$ solves the Cauchy problem
\eqref{eqCauchylinfilt} with the diffusion operator $\AC$ given by
\eqref{eqgenlinfilt}.

However, bounded functions are neither sufficient nor natural in our quantum mechanics setting,
where quadratic or linear functions are mostly met. To deal with them notice firstly that
linear functions of the type $\hat \phi: \ga \mapsto (\phi, \ga)={\tr}\, (\phi \ga)$ with a $\phi \in \HC^2_s$
are invariant under $T_t$. Namely,
\[
T_t \hat \phi (\ga)=\E (\phi, \Ga_t(\ga))=(\E \Ga_t^* \phi, \ga),
\]
where $\Ga_t^*$ is the resolving operator for the dual equation to \eqref{Lindstoch}:
\begin{equation}
\label{Lindstochdu}
dm(t)=i[H,m(t)] \, dt +\LC_{L^*} m(t) \, dt +(L^* m(t)+ m(t) L) dY(t).
\end{equation}

A natural space for $T_t$ is a space $C_{quad}(\HC^2_s)$ of continuous functions $f$
of at most quadratic growth, that is, when the norm
\[
\|f\|_{quad}=\sup_{\ga \in \HC^2_s} \frac{|f(\ga)|}{1+{\tr} \, \ga^2}
\]
is bounded. A convenient subspace of this space is the space $C^2_{quad}(\HC^2_s)$
of twice continuously differentiable functions with bounded second derivatives.
Equipped with the norm
\[
\|f\|_{C^2_{quad}(\HC^2_s)}=\|f\|_{quad}+\sup_{\ga} \frac{\|f'(\ga)\|}{\sqrt{1+{\tr} \, \ga^2}}
+\sup_{\ga} \|f''(\ga)\|,
\]
this space is Banach and such that the generator $\AC$ is bounded as an operator
from $C^2_{quad}(\HC^2_s)$ to $C_{quad}(\HC^2_s)$.

\begin{theorem}
\label{Linsdesem1}

(i) $T_t$ preserve the spaces $C_{quad}(\HC^2_s)$ and $C^2_{quad}(\HC^2_s)$
and act there as a bounded quasi-contraction semigroup.

(ii) The semigroup $T_t$ is strongly continuous in the closed subspace $C_{quad,uc,\infty}(\HC^2_s)$
of $C_{quad}(\HC^2_s)$ consisting of functions $f$ such that the function $f(\ga)/(1+{\tr}\, \ga^2)$
belongs to  $C_{luc,\infty}(\HC^2_s)$.

(iii) If $f\in C^2_{quad}(\HC^2_s)$, then$T_tf$ solves the Cauchy problem \eqref{eqCauchylinfilt}.

(iv)) The subspace of $C_{quad}^2(\HC^2_s)$ consisting of functions with the second derivative from
$C_{luc,\infty}(\HC^2_s)$ is an invariant core for the generator of the semigroup $T_t$ in $C_{quad,uc,\infty}(\HC^2_s)$.
 \end{theorem}

\begin{proof}
(i) The invariance of $C_{quad}^2(\HC^2_s)$ follows from
\eqref{eqLinsdesem}. The preservation of the quadratic growth
 follows from \eqref{Lindstoch}. Part (ii) is a consequence of (i) and Theorem \ref{Linsdesem}.
 Part (iii) is a consequence of Ito's lemma. (iv) Taking into account all previous statements,
it remains to note (see e.g. \cite{Eells}) that smooth functions are dense in the space of continuous functions
on any open bounded subset of a Hilbert space.
\end{proof}

\begin{remark}
The statements of Theorem \ref{Linsdesem1} remain true, if instead of
the space $C_{quad}(\HC^2_s)$ one works with its subspace $C_{quad, weak}^2(\HC^2_s)$ of weakly continuous functions.
This space is preserved by $T_t$, because the solutions of  SDE \eqref{Lindstoch} are given by continuous linear operators.
\end{remark}


\subsection{Diffusive case: normalized version}
\label{secnormeq}

 Recall that a strong
solution of an SDE like \eqref{Lindstochnorm1} is a process $\rho(t)$ that can
 be expressed as a measurable function of a given Brownian motion $B(t)$. By a
weak solution one means a pair of processes $(\rho(t),B(t))$ (where $B(t)$ is a Brownian
motion) defined on a certain stochastic basis, adapted to its filtration and satisfying
\eqref{Lindstochnorm1}. One says that weak solution is  unique in law
if for any two solutions $(\rho^1,B^1)$ and $(\rho^2,B^2)$ (possibly defined on
different probability spaces) the processes $\rho^1$ and $\rho^2$ have the same distribution.

As was mentioned, solutions of nonlinear equation can be built
formally from the linear one. Namely, from \eqref{Lindstochmart}
we can derive by Ito's formula that
\[
d\frac{1}{{\tr} \, \ga(t)}=-\frac{1}{({\tr} \, \ga(t))^2} \, {\tr} \, (L\ga(t) +\ga(t) L^*) dY_t
+\frac{1}{({\tr} \, \ga(t))^3} [{\tr} \, (L\ga(t)+\ga(t) L^*)]^2 \, dt.
\]
Hence by Ito's product rule we check that the normalised density operator
$\rho(t)=\ga(t)/{\tr} \, \ga(t)$ satisfies the equation
 \[
d\rho=-i[H, \rho(t)] \, dt +\LC_L \rho(t)\, dt
\]
 \begin{equation}
\label{Lindstochnorm}
+(L\rho(t)+\rho(t) L^*-\rho(t)\, {\tr} \, (L\rho(t)+\rho(t) L^*) )
[dY_t-{\tr} \, (L\rho(t)+\rho(t) L^*) dt].
\end{equation}

Therefore, in terms of the {\it innovation process}
 \begin{equation}
\label{outputinnovation}
B(t)=Y(t)-\int_0^t {\tr} \, (L\rho(s)+\rho(s) L^*) \, ds
\end{equation}
 the equation for the inverse trace rewrites as
  \begin{equation}
\label{eqtrinnov}
d\frac{1}{{\tr} \, \ga(t)}
=-\frac{1}{{\tr} \, \ga(t)} \, {\tr} \, (L\rho(t) +\rho(t) L^*) dB(t)
\end{equation}
and the equation for the normalized density operator \eqref{Lindstochnorm} rewrites in the
standard form \eqref{Lindstochnorm1} of the nonlinear filtering equation.
 Notice that, for positive solutions, ${\tr} \, \ga(t)$ is a positive martingale
 that specifies the density between the measures $\P$ and $\Q$ that make $Y(t)$
 and $B(t)$ Brownian motions, respectively:
\[
\E_{\Q}\xi=\E_{\P}(\xi \, {\tr}\, \ga(t))
\]
for $\FC_t$-adapted $\xi$.
Therefore, Theorem \ref{LindStochLin1} implies the existence of
weak solutions to nonlinear equation \eqref{Lindstochnorm1}.

Moreover, let a continuous pair of processes $(\rho(t), B(t))$ be a
weak solution to \eqref{Lindstochnorm1} with a positive initial condition
$\rho_0$ and with ${\tr}\, \rho(t)=1$ for all $t$. Then,
if the process ${\tr }\, \ga(t)$ is defined as the solution of equation \eqref{eqtrinnov}
(with any positive initial condition), the process
$\ga(t)$ is defined as $\ga(t)= \rho(t) {\tr }\, \ga(t)$ and the process
$Y(t)$ is defined via \eqref{outputinnovation}, then these processes satisfy
the linear equation \eqref{Lindstoch}. Therefore, any weak solution of the
nonlinear problem can be obtained by normalization from a solution to the linear one.
Hence the uniqueness for the latter implies the uniqueness (in distribution) for the former.

Summarizing we can conclude the following well-posedness
result for weak solutions of nonlinear equation \eqref{Lindstochnorm1}.

\begin{theorem}
\label{LindStochnonLin}
Let $\rho_0$ be a positive operator of unit trace.
Then there exists a unique in law weak solution of equation \eqref{Lindstochnorm1} in $\HC^2_s$
with $B(t)$ a Brownian motion, the initial data $\rho_0$ and such that all $\rho(t)$
are positive-definite operators of unit trace.
\end{theorem}

Let us turn to the strong solutions of \eqref{Lindstochnorm1}. The difficulty with its direct analysis
is due to the fact that though, according to Lemma \ref{lemboundderden}, it can be written
as an equation with Lipschitz coefficients in $\HC^1_s$, the stochastic integrals are defined in $\HC^2_s$,
where the coefficients are singular.
To overcome this difficulty we employ an additional idea:
it turns out that the equation for mixed states can be rewritten as a filtering equation for pure
 states in some enhanced Hilbert space.

 \begin{theorem}
\label{LindStochnonLin1}
Let $\rho_0$ be a positive operator of unit trace,
and $B(t)$ a BM. Then there exists a unique strong solution of equation
\eqref{Lindstochnorm1} in $\HC^2_s$, with the initial data $\rho_0$
and such that all $\rho(t)$ are positive trace class operators of unit trace.
\end{theorem}

\begin{proof}
From the proof of Theorem \ref{LindStochnonLin} we know that $\rho(t)$ solves \eqref{Lindstochnorm1}
 if and only if $\ga(t)={\tr}\, \ga(t)\, \rho(t)$ solves the equation
  \begin{equation}
\label{Lindstochnormlinmix}
d\ga(t)=-i[H, \ga(t)] dt +\LC_L \ga(t) dt +(L\ga(t)+\ga(t) L^*)
[dB(t)+ \pi(t) \, dt],
\end{equation}
with
\[
\pi(t)=\frac{{\tr} \, (L\ga(t)+\ga(t) L^*)}{{\tr}\, \ga(t)}.
\]

As in the proof of Theorem \ref{LindStochLin1} above, we can
expand $\rho_0=\ga_0$ in a series
\[
 \ga_0=\sum_{k=1}^{\infty} p_k e_k\otimes \bar e_k
 \]
 with a non-negative sequence $\{p_k\}$ summing up to one
 and an orthonormal basis $\{e_k\}$.
Hence we can represent $\ga(t)$ as the convergence series of pure states
 \[
 \ga(t)=\sum_{k=1}^{\infty} p_k e_k(t)\otimes \bar e_k(t),
 \]
 with $e_k(t)$ solving the linear filtering equation for
pure states \eqref{eqqufiBlins}:
\begin{equation}
\label{eqinfdim}
de_k(t)=(-iH e_k(t)-\frac12 L^*L e_k(t))\,dt +Le_k(t) [dB(t)+\pi(t)\, dt].
\end{equation}
Here
\[
\pi(t)=\frac{\sum_{k=1}^{\infty} p_k (e_k(t), (L+L^*) e_k(t))}{\sum_{k=1}^{\infty} p_k \|e_k(t)\|^2}.
\]

It is a key observation that the infinite-dimensional system of SDEs \eqref{eqinfdim}
can be considered as a single SDE
with values in the  Hilbert space $l^2_{\HC}(\{p_k\})$ consisting of infinite sequences
$\e=(e_1, e_2, \cdots )$ of vectors from $\HC$ and equipped
with the norm
\[
\|\e\|^2=\sum_{k=1}^{\infty} p_k (e_k, e_k).
\]
Bounded operators in $\HC$ extend naturally (acting identically on each coordinate)
to bounded operators in $l^2_{\HC}(\{p_k\})$ with the preservation of norm.
In this notation system \eqref{eqinfdim} writes down as the SDE
\begin{equation}
\label{eqinfdim1}
d \e(t)=(-iH \e(t)-\frac12 L^*L \e(t))\, dt +L\e(t) \left[dB(t)+\frac{(\e, (L+L^*) \e)}{(\e, \e)} \, dt\right].
\end{equation}

This equation is the same as \eqref{eqqufiBlinsB} (though written in an enhanced Hilbert space). Hence
the coefficients of this equation are globally Lipschitz due to Lemma \ref{lemboundder}.
Therefore, it has the unique solution. And consequently, \eqref{Lindstochnormlinmix} has a unique solution $\ga(t)$,
and hence  $\rho(t)=\ga(t)/T(t)$ is the unique solution to the Cauchy problem for \eqref{Lindstochnorm1}.
\end{proof}

The full well-posedness of a problem includes also a statement on a continuous dependence of
the solution on initial data and parameters of the problem. We prove here the continuous dependence
on the Hamiltonian, which would be crucial for the next Section.

 \begin{theorem}
\label{LindStochnonLin2}
Under the assumption of Theorem \ref{LindStochnonLin1}
let us consider the Cauchy problem for equations
\begin{equation}
\label{Lindstochnorm11}
d\rho(t)=-i[H+H_j(t),\rho(t)]\, dt+\LC_L \rho (t)\, dt
+[L\rho(t)+\rho(t) L^*-\rho(t)\, {\tr} \, (L\rho(t)+\rho(t) L^*) ] dB(t),
\end{equation}
$j=1,2$, where $H_1(t), H_2(t)$ are continuous families of bounded self-adjoint
 operators in $\HC$.
Then, for the solutions $\rho_j(t)$, $j=1,2$, of these equations with the same
positive initial data $\rho_0$ of unit trace, one has the following estimate:
\begin{equation}
\label{Lindstochnorm12}
\E \|\rho_1(t)-\rho_2(t)\|_{H^{1,2}_s} \le \sqrt t C(t) \sup_{s\in [0,t]}\|H_1(s)-H_2(s)\|,
\end{equation}
where $C(t)$ is an increasing continuous function depending on $\|L\|$,
and where $\HC^{1,2}_s$ means of course either $\HC^1_s$ or $\HC^2_s$.
\end{theorem}

\begin{proof} By changing to the "interaction picture", that is, to equations
on the variable $e^{-itH}\rho(t) e^{itH}$ we can reduce the discussion to the case
of vanishing $H$. Time-dependence of $H_j$ and $L$ arising from this change does
not affect the argument. Therefore, without loss of generality we can set $H=0$.

Making the transformation to the equations in $l^2_{\HC}(\{p_k\})$,
as in the proof of the previous theorem,
 we can rewrite equations \eqref{Lindstochnorm11} as the equations
\begin{equation}
\label{eqinfdim11}
d \e^j(t)=(-iH_j \e^j(t)-\frac12 L^*L \e^j(t))\, dt +L\e^j(t)
\left[dB(t)+\frac{(\e^j, (L+L^*) \e^j)}{(\e^j, \e^j)} \, dt\right],
\end{equation}
with $j=1,2$.

As shown in the previous theorem, these equations are SDEs in a Hilbert space
with globally Lipschitz coefficients. Moreover, these two equations differ by
bounded linear terms. Hence it is a standard procedure
(see e.g. Proposition 7.1 in \cite{KolQuantLLN}) to derive an
estimate for $\E \|(\e^1-\e^2)(t)\|^2$ in terms of $\|H_1-H_2\|$.
We skip details here referring to \cite{Kol25a}, where they are given in full.
\end{proof}

\subsection{Semigroups generated by nonlinear filtering equations}
\label{semnonlinfilt}

Let us look at the analytic properties of the semigroup $\Phi_tf(\rho)=\E_{\rho} f(\rho(t))$,
where $\rho(t)$ are solutions to the Cauchy problem of SDE \eqref{Lindstochnorm1}. As follows
from Theorem \ref{LindStochnonLin}, the semigroup $\Phi_t$ is a $C$-Feller semigroup in $C(S(\HC))$
(a semigroup of positivity preserving contractions). For finite-dimensional $\HC$ it implies directly
that smooth functions (of any order) represent invariant cores for generator \eqref{eqdifgener}.
However, for infinite-dimensional $\HC$, when searching for an invariant core
(where generator \eqref{eqdifgener} is defined), we face the same problem
as for the case of countable observation above, namely that smooth functions
 are not dense in $C(S(\HC))$. To tackle this issue, we shall reduce the action
 of $\Phi_t$ to the subspace $\hat C(S(\HC))\subset C(S(\HC))$ consisting of functions that
 are continuous in the topology of $\HC^2_s$. In other words, functions from $\hat C(S(\HC))$
 are restrictions to $S(\HC)$ of the functions from the space $C(\HC^2_sP^+)$ of continuous functions
 on the projective space $\HC^2_sP^+$ of lines in $\HC^2_s$ going through nonnegative operators. Therefore, yet equivalently,
 the space $\hat C(S(\HC))$ is isomorphic to the space $C(S^+(\HC^2_s))$ of continuous functions on
 the positive part $S^+(\HC^2_s)$ of the unit sphere in $\HC^2_s$.

It is an easy exercise to derive the evolution of states normalised in the sense of $\HC^2_s$
(not $\HC^1_s$, as we did earlier) from evolution \eqref{Lindstoch}. Namely,
from \eqref{Lindstochquadtr} we derive that
\[
d\|\ga\|_2^{-1}=-\frac{1}{\|\ga\|_2^3}(L+L^*,\ga^2)\, dY(t)
\]
\[
-\frac{1}{2\|\ga\|_2^3} \,{\tr}\, (2\ga L\ga L^*+L \ga L \ga + \ga L^* \ga L^*) \, dt
 +\frac32 \frac{1}{\|\ga\|_2^5} (L+L^*,\ga^2)^2\, dt,
 \]
 where we use scalar product in the space $\HC^2_s$, so that $(L+L^*,\ga^2)={\tr}((L+L^*)\ga^2)$.
Therefore for $\hat \ga=\ga/\|\ga\|_2$ we obtain
\[
d \, \hat  \ga=(-i[H,\hat \ga(t)] \, dt +\LC_L \hat \ga(t)) \, dt
\]
\[
-\frac12 \hat \ga \, {\tr} \,(2\hat \ga L\hat \ga L^*+ L \hat \ga L \hat \ga + \hat \ga L^* \hat \ga L^*) \, dt
\]
\[
-(L+L^*, \hat \ga^2) (L\hat \ga +\hat \ga L^*) \, dt+\frac32 \hat \ga \, (L+L^*, \hat \ga^2) ^2 \, dt
\]
\[
+[L\hat \ga(t)+\hat \ga(t) L^* -\hat \ga (L+L^*, \hat \ga^2)]\, dY(t).
\]

From this one derives
\[
d(\hat \ga, \hat \ga)=2(1-(\hat \ga, \hat \ga)) (L+L^*, \hat \ga^2)\, dY(t)
+2( \hat \ga , \LC_L \hat \ga(t)) \, dt
\]
\[
-(\hat \ga, \hat \ga) \, {\tr} \,(2\hat \ga L\hat \ga L^*+ L \hat \ga L \hat \ga + \hat \ga L^* \hat \ga L^*) \, dt
\]
\[
-2(L+L^*, \hat \ga^2)^2 \, dt +3 (\hat \ga, \hat \ga) \, (L+L^*, \hat \ga^2) ^2 \, dt
\]
\[
+ (L\hat \ga(t)+\hat \ga(t) L^* -\hat \ga (L+L^*, \hat \ga^2),L\hat \ga(t)+\hat \ga(t) L^* -\hat \ga (L+L^*, \hat \ga^2)) dt.
\]
\[
=(1-(\hat \ga, \hat \ga))\bigl[[2 (L+L^*, \hat \ga^2)\, dY(t)-4  (L+L^*, \hat \ga^2)^2\, dt
\]
\[
+ \, {\tr} \,(2\hat \ga L\hat \ga L^*+ L \hat \ga L \hat \ga + \hat \ga L^* \hat \ga L^*) \, dt\bigr].
\]
This preserves $(\hat \ga, \hat \ga)$, when it equals to one.

If we modify the above equation for $\hat \ga $ to (normalising $\ga$ in the nonlinear terms)
 \[
d \, \hat  \ga=(-i[H,\hat \ga(t)] \, dt +\LC_L \hat \ga(t)) \, dt
\]
\[
-\frac12 \hat \ga \, {\tr} \,(2\hat \ga L\hat \ga L^*+ L \hat \ga L \hat \ga + \hat \ga L^* \hat \ga L^*)
 \frac{1}{\|\hat \ga\|^2_2}\, dt
\]
\[
-\left(L+L^*, \frac{\hat \ga^2}{{\|\hat \ga\|^2_2}}\right) (L\hat \ga +\hat \ga L^*) \, dt
+\frac32 \hat \ga \, \left(L+L^*, \frac{\hat \ga^2}{{\|\hat \ga\|^2_2}}\right) ^2 \, dt
\]
\begin{equation}
\label{eqquantfilonhilbertsphere}
+\left[L\hat \ga(t)+\hat \ga(t) L^* -\hat \ga \left(L+L^*, \frac{\hat \ga^2}{{\|\hat \ga\|^2_2}}\right)\right]\, dY(t),
\end{equation}
then we get $d(\hat \ga, \hat \ga)=0$,
and the conservation of the norm becomes universal, but the evolution of operators
with unit Hilbert-Schmidt norm would not be changed.
Moreover, by Lemma \ref{lemboundderden},
equation \eqref{eqquantfilonhilbertsphere} has smooth coefficients
in $S^+(\HC^2_s)$. It follows that evolution given by \eqref{eqquantfilonhilbertsphere} is well-defined
and preserves smooth functions (of arbitrary order) on  $S^+(\HC^2_s)$. Since  $S^+(\HC^2_s)$
is a Hilbert-space-based infinite-dimensional manifold, the classes of smooth functions
$C^k(S^+(\HC^2_s))=C^k(\HC^2_sP^+)$ are dense
in the space of continuous functions for any $k$.
But evolution  \eqref{Lindstochnorm1} is just the restriction
of evolution \eqref{eqquantfilonhilbertsphere} on  $S^+(\HC^2_s)$ , or equivalently on $\HC^2_sP^+$.
Identifying, as noted above, the spaces
 $\hat C(S(\HC))$, $C(\HC^2_sP^+)$ and $C(S^+(\HC^2_s))$ we get dense classes of invariant smooth
 (in the topology of $\HC^2_s$, of course) functions on $S(\HC)$.

Thus we have proved the following result.

\begin{theorem}
\label{nonLinsdesem}
The semigroup $\Phi_t$ is strongly continuous in the space $\hat C(S(\HC))$, having the spaces
$C^k( S^+(\HC^2_s))=C^k(\HC^2_sP^+)$ as invariant cores for any $k\ge 2$.
\end{theorem}

Since generator \eqref{eqmixgenerun} of quantum filtering process for mixed observation (and unitary coupling
operators with counting measurements) is seen to be obtained from \eqref{eqdifgener} by adding a bounded perturbing
term, the standard perturbation theory allows one to directly extend the previous theorem to the case of operators
\eqref{eqmixgenerun}  yielding the following result.

\begin{theorem}
\label{nonLinsdesemmix}
The semigroup $\Phi_t^{mix}$ generated by operator \eqref{eqmixgenerun}
is strongly continuous in the space $\hat C(S(\HC))$, having the spaces
$C^k(S^+(\HC))=C^k(\HC^2_sP^+)$ as invariant cores for any $k\ge 2$.
\end{theorem}

\subsection{Classical Hamiltonians}
\label{classham}

When $H$ is unbounded, one can work with the interaction representation
 or with mild forms of filtering equations,
as we did above, reducing the analysis to bounded operators. However, it is of course desirable to solve
the corresponding filtering equations in their original form. For this to be possible, one has to ensure that
the evolution remains inside the domain of $H$ for all times.

For completeness, we briefly present here the simplest basic examples (which will be used
also when dealing with the propagation of chaos) referring to \cite{Kol25b}
for the full story including unbounded $L$.

The classical Hamiltonian of quantum mechanics acts in $L^2(\R^d)$ as the operator
\[
Hf(x)=-\frac12 \De f(x) +V(x) f(x),
\]
where $V(x)$ is a function in $\R^d$, which is also identified (as usual) with the operator
of multiplication by this function.
Recall that a standard example of the condition ensuring that $H$ is self-adjoint
in $L^2(\R^d)$ on the domain being the Sobolev space $W^2(\R^d)$ is that $V$ is bounded
for $d\le 2$ and  that $V=V_1+V_2$,
where $V_2$ is bounded and $V_1\in L^p(R^d)$ for $d\ge 3$ and $p>d/2$, see e.g. \cite{ReedSimon}
and more general situations in \cite{Yajima}.

\begin{theorem}
\label{thclassichamfilt}
Let $V$ be a measurable function on $\R^d$ such that the operator $H=-\frac12 \De +V(x)$ is self-adjoint
in $L^2(\R^d)$ on the domain being the Sobolev space $W^2(\R^d)$. Let $L$ be a multiplication operator
on the real function $L$ that belongs to $C^2(\R^d)\cap W^2(\R^d)$. For dimensions $d=1,2,3$ it
is sufficient to assume $L\in C^1(\R^d)\cap W^2(\R^d)$. Then the equations
\eqref{eqqufiBlins} and \eqref{eqqufiBnonlinsn} are well-posed in the sense that for any initial condition
from $W^2(\R^d)$ there exists a unique global solution of the corresponding Cauchy problem that belongs
$W^2(\R^d)$ for all times.
\end{theorem}

\begin{proof}
The simplest proof is via the mild forms
\eqref{eqqufiBlinsm} and \eqref{eqqufiBnonlinsm}. Though the well-posedness
was proved via the interaction representation, we note that it can be also obtained directly from
these mild forms via the standard Banach fixed-point principle. Moreover,
from the condition of the theorem it follows that the semigroup $e^{iHt}$
is also bounded and strongly continuous in $W^2(\R^d)$ (see e.g. \cite{Yajima}).
Moreover, the operators of multiplication on $L$ and $L^2$ are bounded in $W^2(\R^d)$.
The same holds under simplified assumptions for $d=1,2,3$, because in these dimensions
$W^2(\R^d)\in L^{\infty}(\R^d)$ (see e.g. \cite{ReedSimon}).  Therefore,
the same Banach fixed-point principle applied in $W^2(\R^d)$ allows one
to conclude that if the initial condition was in $W^2(\R^d)$, then it will stay there for all times.
Consequently, the integrands on the r.h.s. of equations \eqref{eqqufiBlinsm} and \eqref{eqqufiBnonlinsm}
become differentiable in $t$, and consequently, solutions to these equations become also solutions to
\eqref{eqqufiBlins} and \eqref{eqqufiBnonlinsn}.
\end{proof}

Analogous results hold for the equations on mixed states. Recall that we proved
the well-posedness of the interaction representation (and hence also mild form)
of equation  \eqref{Lindstoch} in the Hilbert space of self-adjoint
Hilbert Schmidt operators $\HC^2_s$. In the setting of Theorem \ref{thclassichamfilt} this space
coincides with the Hilbert subspace  $L^2_s(\R^{2d})$ of symmetric functions in
$L^2(\R^{2d})$ (integral kernels such that $\ga(x,y)=\overline{\ga(y,x)}$).
Moreover, the domain of the generator $-i[H, \ga]$ of the semigroup of operators
$\ga \mapsto e^{-iHt}\ga e^{iH t}$ in $\HC^2_s$  is given by the Sobolev space $W^2(\R^{2d})$.
Consequently, the same argument as above gives the following.

\begin{theorem}
\label{thclassichamfiltmix}
Under the assumptions of Theorem \ref{thclassichamfilt},  equations
  \eqref{Lindstoch} and \eqref{Lindstochnorm1} are well posed in $\HC_s^2$ in the sense that for any
  $\ga_0\in W^2(\R^{2d})$ (resp. $\ga_0\in W^2(\R^{2d})\cap S(L^2(\R^d))$ there
 exist unique global solutions of the corresponding Cauchy problems that belong
$W^2(\R^{2d})$ for all times.
\end{theorem}

Similar results hold for the filtering equations of counting observations.

\subsection{Appendix: some trace inequalities}

\begin{prop}
If $A$ is a self-adjoint Hilbert-Schmidt operator and $B$ a bounded operator, then
\begin{equation}
\label{eqmyineqtr}
2|{\tr} \, (ABAB^*)|\le {\tr} \, [A^2(BB^*+B^*B)],
\end{equation}
and
\begin{equation}
\label{eqmyineqtr1}
|{\tr} \, (AB AB +AB^*AB^*)|\le {\tr} \, [A^2 (BB^*+B^*B)].
\end{equation}
\end{prop}

\begin{proof}
By approximation it is reduced to finite-dimensional situation. The diagonalization procedure reduces
the problem to the case when $A$ is a diagonal matrix with real numbers $a_i$ on the diagonal. Then
 \[
 2{\tr} \, (ABAB^*)
=2\sum a_ib_{ij}a_j\bar b_{ij}=\sum a_i a_j (|b_{ij}|^2+|b_{ji}|^2)
\]
\[
=2\sum_i a_i^2 |b_{ii}|^2+2\sum_{i<j}a_ia_j(|b_{ij}|^2+|b_{ji}|^2).
\]
The r.h.s. of \eqref{eqmyineqtr} equals
\[
\sum a_i^2 (|b_{ij}|^2+|b_{ji}|^2)
=2\sum_i a_i^2 |b_{ii}|^2+ \sum_{i < j}(a_i^2+a_j^2)(|b_{ij}|^2+|b_{ji}|^2).
\]
Thus \eqref{eqmyineqtr} holds, because $2|a_ia_j| \le a_i^2 +a_j^2$.

Inequality \eqref{eqmyineqtr1} rewrites as
\[
2\sum_i a_i^2 |Re (b_{ii}^2)|+4\sum_{i<j}|a_ia_j| \, |Re (b_{ij} b_{ji})|
\le 2\sum_i a_i^2 |b_{ii}|^2+ \sum_{i < j}(a_i^2+a_j^2)(|b_{ij}|^2+|b_{ji}|^2),
\]
which easily seen to hold.
\end{proof}

In particular, for self-adjoint $B$ it follows that

\begin{equation}
\label{eqmyineqtr2}
|{\tr} \, (ABAB)|\le {\tr} \, (A^2 B^2).
\end{equation}

\section{Derivation of quantum filtering equations}
\label{secder}

As we noted, the theory of quantum filtering was built initially in papers
\cite{Bel87}, \cite{Bel88}, \cite{Bel92}
on the basis of quantum stochastic calculus.
A well written review of this development is given in \cite{BoutHanJamQuantFilt}.
Afterwards, there appeared alternative approaches to the derivations of the filtering equations,
see \cite{Attal}, \cite{BelKol},
\cite{Pellegrini}, \cite{BarchBel}, \cite{Holevo91}, \cite{KolQuantFrac} and references therein.
Some further extensions can be found in \cite{BarndLoub}, \cite{Loubenets}.
The most straightforward (and intuitive) derivation seems to be via the limit
of successive instantaneous measurements,
when the times between these measurements tend to zero. The complete rigorous derivation
of the filtering equations via this method  in both diffusive and counting observations
was performed for finite-dimensional quantum mechanics in \cite{Pellegrini} and
\cite{Pellegrini10}. In \cite{KolQuantFrac} the author added the explicit rates of convergence
and indicated the ways to infinite-dimensional extensions.
Using the well-posedness results above, we present here, seemingly for the first
time, the complete rigorous derivation of the filtering equations for the general
infinite-dimensional case including the rates of convergence. For completeness,
we also discuss in the initial subsections \ref{contmeas} and \ref{Zenopar} the classical
principles of quantum measurements and the famous quantum Zeno paradox.
The last subsection is an appendix, where we recall some results on the rates of
convergence of semigroups in terms of their generators.

\subsection{Quantum measurement}
\label{contmeas}

If two quantum systems living in Hilbert spaces $\HC$ and $\HC_1$
are brought to interaction, the combined system
has the tensor product Hilbert space  $\HC \otimes \HC_1$ as the state space.
We will be interested now only in the case,
when the Hilbert space $\HC_1$ is finite-dimensional, that is, $\HC_1=\C^m$.
In this case $\HC\otimes \HC_1$ can be identified with the Hilbert-space $\HC^m$
of vectors $(h_0, h_1, \cdots, h_{m-1})$ with
coordinates $h_j\in \HC$, so that if $h\in \HC$ and $f=(f_0, \cdots, f_{m-1})\in \C^m$, then
$h\otimes f=(f_0 h, \cdots, f_{m-1} h)$.
Similarly, operators $A$ in $\HC \otimes \HC_1$ are
given by matrices $\{A^i_j\}$, $i,j=0, \cdots , {m-1}$,
with each entry $A^i_j$ being an operator in $\HC$.
A product $A\otimes B$ of two operators $A$ and $B$ acting in $\HC$ and $\HC_1$ respectively
is defined generally by its action on the tensor products as
\[
(A\otimes B)(e\otimes f)=Ae\otimes Bf.
\]
In case $\HC=\C^m$, it means that if $B$ is given by an $m\times m$-matrix with entries $B_{ij}$,
then $A\otimes B$ is given by the matrix with the entries $AB_{ij}$.

An operator $A$ in $\HC$ has the natural lifting $A\otimes I$ (where $I$ is the identity operator) to $\HC\otimes \HC_1$.
Similarly an operator $B$ in $\HC_1$ has the natural lifting $I\otimes B$ to $\HC\otimes \HC_1$.

The key notion of the theory of interacting systems is that of the {\it partial trace}. For an operator
$A$ in $\HC \otimes \HC_1$ given by its matrix $\{A^k_j\}$, the partial trace with respect to $\HC_1$
is the operator in $\HC$ defined as
 \begin{equation}
 \label{eqdefparttr}
 {\tr}_{p1} A=\sum_k A^k_k.
 \end{equation}
This  partial trace is interpreted as the state of the first system given the state of the coupled one.
 Similarly, the partial trace with respect to $\HC$ is the operator
 in $\HC_1$ defined as
 \[
 ({\tr}_{p0} A)^i_j={\tr} A^k_j
 \]
 (of course, whenever all these traces exist).
 Clearly,
 \[
 {\tr} ({\tr}_{p0} A)= {\tr} ({\tr}_{p1} A)= {\tr} (A).
 \]

{\it Physical observables} in the Hilbert space of a quantum system are
given by self-adjoint operators $A$ in it. If $A$ has a discrete spectrum,
then $A$ has the spectral decomposition $A=\sum_j \la_j P_j$, where
 $P_j$ are orthogonal projections on the eigenspaces of $A$ corresponding to
 the eigenvalues $\la_j$. According to the {\it basic postulate of quantum measurement}
 \index{basic postulate of quantum measurement}, measuring observable $A$ in a state $\ga$
(often referred to as the {\it Stern-Gerlach experiment}\index{ Stern-Gerlach experiment})
can yield each of the eigenvalue $\la_j$ with the probability
   \begin{equation}
\label{eqantiunitdress1}
{\tr} \, (\ga P_j)={\tr} \, (P_j \ga P_j),
\end{equation}
 and, if the value $\la_j$ was obtained, the state
of the system changes (instantaneously) to the reduced state
\[
P_j\ga P_j/ {\tr} \, (\ga P_j).
\]
In particular, if the state $\rho$ was pure, $\ga=|\psi\rangle \langle \psi|$, then the
probability to get $\la_j$ as the result of the measurement becomes $(\psi_,P_j\psi)$
and the reduced state also remains pure and is given by the vector $P_j\psi$.

\begin{remark}
The rule \eqref{eqantiunitdress1} is often referred to as von Neumann-L\"uders postulate.
It is a particular performance of a more general theory of quantum instruments, initiated by
Kraus, Ozawa and Holevo, see e.g. \cite{Hol97}, \cite{Ozawa}, \cite{Gudder}
and references therein.
\end{remark}

It is physically more natural to make indirect measurement of a chosen quantum system,
which we sometimes refer to as an atom, where this system is coupled (brought to interaction) with
some measuring apparatus, whose pointer positions (states) we can read.  More precisely,
{\it indirect measurements} of an atom in the initial space $\HC$,
are organised as follows.
One couples the atom with another quantum system, a measuring devise,
also referred to an ancilla system or a reservoir, specified
by another Hilbert space $\HC_1$, which here we choose to be finite-dimensional: $\HC_1=\C^m$.
 Namely the combined system lives in the tensor product
 Hilbert space $\HC\times \HC_1$ and its evolution is given by  certain self-adjoint
 operator $H$ in  $\HC\times \HC_1$. In the measuring device some fixed vector $\varphi \in \HC_1$
 is chosen, called the vacuum and interpreted as the stationary state of the devise when
 no interaction is involved. The corresponding density matrix will be denoted $\Om=|\varphi \rangle \langle \varphi|$.
 Indirect measurements of the states of the atom
 are performed by measuring the coupled system via an observable of the second system
 and then projecting the resulting state to the atom via the partial trace.

Namely, it is given by an operator $R$ in $\HC_1$ with the spectral decomposition
$R=\sum_j \la_j P_j$, where $P_j=f_j\otimes \bar f_j$ and $f_j$ the normalised eigenvectors of $R$,
 and it is performed in two steps: given a mixed state $\ga$ in  $\HC\times \HC_1$
one performs a measurement of $R$ lifted as $I\otimes R$ to $\HC\times \HC_1$
yielding values $\la_j$ and new states
\[
(I\otimes P_j)\ga (I\otimes P_j)/ {\tr} \, (\ga  (I\otimes P_j))=(f_j, \ga f_j) \otimes P_j
\]
with the probabilities
\[
p_j= {\tr} \, (\ga  (I\otimes P_j))={\tr} \, (f_j, \ga f_j),
\]
and then one projects these states to $\HC$
via the partial trace producing the states
\begin{equation}
\label{eqindirmeas}
\frac{{\tr}_{p1} [(I\otimes P_j)\ga (I\otimes P_j)]}{ {\tr} \, (\ga  (I\otimes P_j))}
=\frac{(f_j, \ga f_j)}{{\tr}\, (f_j, \ga f_j)},
\end{equation}
where, for an operator-valued matrix $\ga$ and a vector $f=\{f^l\}\in \C^m$ we denoted
\[
(f,\ga f)=\sum_{k,l} \bar f^k\ga_{kl} f^l.
\]

\subsection{Quantum Zeno paradox and watchdog effect}
\label{Zenopar}

A natural idea of organising continuous measurement is via a limit of sequential instantaneous
measurements with times between these measurements tending to zero. Let us see what comes out of it.
Suppose a starting point $\psi_0$ of our atom in $\HC$ is chosen, between measurements
 the atom evolves according to the Hamiltonian $H$ and each measurement is made according
 to the operator $P_0$, which is the projection on $\psi_0$. In other words,
 we check every time whether our state is $\psi_0$ or is orthogonal to it.
 Assume $H$ is bounded. Then starting with $\psi_0$ we move in a small time $t$ to
 \[
 e^{iHt}\psi_0\sim \frac{\psi_0+iHt \psi_0}{\|\psi_0+iHt \psi_0\|}\sim (\psi_0+iHt \psi_0)(1+O(t^2)).
 \]
 Therefore the probability to return to $\psi_0$ after the first measurement at time $t$ is
 \[
 |(e^{iHt}\psi,\psi_0)|^2(1+O(t^2))\sim |(\psi_0+iHt \psi_0,\psi_0)|^2(1+O(t^2))=1+O(t^2).
 \]
 Fix now a time interval of length $T$, choose natural $n$,  and let us make instantaneous
 measurements at times $t_k=(k/n)T$. Then the probability to remain in $\psi_0$ after all $n$
 measurements is $(1+O(k^2T^2/n^2))^n$, which clearly tends to $1$, as $n\to \infty$.
 This is the {\it quantum Zeno paradox} -- when continuous measurement is performed in this way,
 the atom just remains in its initial state independently of the law of its evolution.
 This paradox was suggested in \cite{Zeno}, where we refer to for the corresponding
  argument in case of unbounded (in fact, semibounded) Hamiltonian operator.
 By obvious reasons this effect was also named {\it watch-dog effect}, or yet alternatively
 expressed by  wording "watched kettle never boils".

 To avoid this effect, one has to scale appropriately the dynamics in order to
 make the effect of evolution comparable with the effect of measurements.
 We show below how this is done in the framework of indirect measurements.

 Let us mention however that Zeno paradox above was linked to the assumptions that
 we keep measuring the same projection operator $P_0$ all the times. An alternative way to deviate
 from Zeno paradox is to smoothly change the measuring observable from $P_0$ to $P_t=U(t)P_0U^*(t)$
 with a smooth family of unitary operators $U(t)$. Then similar calculations will show that
 in the limit of continuous measurement the system will move through the eigenstates of $P_t$
 with eigenvalue $1$ o that "watched kettle will always boils". This effect was discovered in
 \cite{Anti-Zeno} and called {\it Anti-Zeno paradox}.

\subsection{Markov chains of sequential indirect observations}
\label{secMarchainindirodserv}

We describe now the
Markov chains of sequential indirect observations, rather standard by now,
in discrete and continuous time.

The discrete time {\it Markov chain of successive indirect observations} (or measurements) evolves according
to the following procedure specified by a triple: a self-adjoint  operator $A$ in $\HC\times \HC_1$,
a self-adjoint operator $R$ in $\HC_1$ and the vacuum vector $\Om$ in $\HC_1$.
(i) Starting with an initial state $\rho$ of $\HC$ one couples it with the device in its vacuum state $\Om$
producing the state $\ga=\rho\otimes \Om$ in $\HC\times \HC_1$,
(ii) During a fixed period of time $h$ one evolves the system according to the operator $A$ producing the state
$\ga_h=e^{-ihA}\ga e^{ihA}$ in $\HC\times \HC_1$,
(iii) One performs the indirect measurement with the state $\ga_h$ yielding the states
\begin{equation}
\label{eqMarkchain}
\rho_h^j=\tilde \rho_h^j/ p_j(h)
\end{equation}
where
\[
\tilde \rho_h^j ={\tr}_{p1} (I\otimes P_j)\ga_h (I\otimes P_j)
\]
\begin{equation}
\label{eqMarkchain0}
={\tr}_{p1} (I\otimes P_j)e^{-ihA}(\rho\otimes \Om) e^{ihA} (I\otimes P_j)
=(f_j, e^{-ihA}(\rho\otimes \Om) e^{ihA} f_j) \otimes P_j,
\end{equation}

are the corresponding non-normalised states, and
\begin{equation}
\label{eqMarkchain1}
p_j(h)={\tr} \, \tilde \rho_h^j
={\tr} \, (e^{-ihA}(\rho\otimes \Om) e^{ihA}  (I\otimes P_j))
={\tr}\, (f_j, e^{-ihA}(\rho\otimes \Om) e^{ihA} f_j)
\end{equation}
are the probabilities of obtaining these states as the result of the measurement.

Then the same repeats starting with $\rho_h$ as the initial state.  Let us denote $U_h$ the transition
operator of this Markov chain that acts on the set of continuous functions on $S(\HC)$ as
\begin{equation}
\label{eqMarkchain2}
U_h f(\rho)=\E f(\rho_h)=\sum_j p_j(h) f(\rho_h^j).
\end{equation}
Similarly one can define the continuous time
{\it Markov chain of successive indirect observations} (or measurements)
$O^{\rho}_{t,h}$ and the corresponding Markov semigroup $T_t^h$ on $C(S(\HC)$
evolving according to the same rules, with only difference that the times between successive measurements
are not fixed, but represent independent exponential random  variables $\tau$ with
the intensity $1/h$: $\P(\tau>t)=e^{-t/h}$.
The generator $L^h$ of this Markov process is bounded in $C(S(\HC))$ and acts as
\begin{equation}
\label{eqMarkchain3}
L^hf(\rho)=\frac{(U_hf-f)(\rho)}{h}=\frac{1}{h}\sum_j p_j(h) (f(\rho_h^j)-f(\rho)).
\end{equation}

All "quantum content" of the theory is now captured in the explicit formula \eqref{eqMarkchain}.
What follows will be
the pure classical probability analysis of these Markov chains, their scaling limits and control.

\begin{remark}
\label{redtimedepder}
Let us comment on the straightforward extension of the scheme above to the case of time dependent family
of operators $A$. Namely, suppose that instead of $A$ we have a family of self-adjoint
operators $A(t)$ in $\HC\otimes \HC_1$. Then in the discrete scheme we perform measurements at times
 $t_k=kh$ and between times $t_k$ and $t_{k+1}$ the system evolves according to $A(t_k)$ producing the state
$\ga_h=e^{-ihA(t_k)}\ga e^{ihA(t_k)}$ in $\HC\times \HC_1$. The rest is the same.
Of course, the corresponding  operators \eqref{eqMarkchain2} and \eqref{eqMarkchain3}
become time-dependent. We mention this extension, because the reduction
of the case of unbounded Hamiltonian via the interaction representation leads
necessarily to time-dependent families $A(t)$.
\end{remark}

In what follows, as the vacuum vector we shall choose the first basis vector $e_0=(1,0,\cdots , 0)$
 of our measuring devise Hilbert space $\HC=\C^m$, so that $\Om =e_0\otimes e_0$
 is the orthogonal projection to $e_0$.
Moreover, to shorten formulas we shall further on write down $m\times m$-matrices
 $D$ representing operators in $\HC_0\otimes \C^m$
 in the concise $2\times 2$-form, that is, as
 \[
R= \begin{pmatrix}
A & C   \\
D & B
\end{pmatrix},
\]
where $B$ is an operator-valued  $(m-1)\times (m-1)$-matrix, $C$ and $D$ are a row
and a column operator-valued vectors respectively.
 For instance, for an operator $D$ in $\HC_0$, we shall write
 \[
D\otimes \Omega
=\begin{pmatrix}
D & 0   \\
0 & 0
\end{pmatrix}.
\]

As all sequential measurements go through projection to the vacuum, the parts of a Hamiltonian operator $H$
in $\HC_0\otimes \C^m$ (describing the dynamics that do not involve the vacuum) are irrelevant for the limiting
dynamics (that we are interested in). Therefore, for the purpose of deriving the limiting dynamics
of continuous observation, we can and will consider only  Hamiltonians
with matrices of the form
\[
A= \begin{pmatrix}
H & -iL_1^* & -iL_2^* &\cdots & -iL_{m-1}^*  \\
iL_1 & 0 & 0 & \cdots & 0 \\
&  &  \cdots &  & \\
iL_{m-1} & 0 & 0 & \cdots & 0
\end{pmatrix},
\]
or better in our concise notations,
\[
A= \begin{pmatrix}
H & -iL^*   \\
iL & 0
\end{pmatrix}.
\]
Here $H$ is a self-adjoint operator in $\HC$ describing
the dynamics of our atom without observation and the vector-valued operator $L=(L_1, \cdots , L_{m-1})$
in $\HC$ describes the interaction of the atom with the measuring devise.

We are aiming at calculating the small $h$ asymptotics of the
Markov transition operators defined by \eqref{eqMarkchain}.

We shall performed all calculations as if both $H$ and $L$ are bounded. However, as was mentioned earlier,
the case of unbounded $H$ reduces to the analysis of vanishing $H$ but with the time-dependent
$L(t)=e^{itH}Le^{-iHt}$ (see Remark \ref{redtimedepder}). Therefore, the only real restriction
in our derivation is the boundedness of $L$.

\begin{remark}
One can check that if one would write some arbitrary matrix $B$ instead
of the zero block in the expression for $H$, then
this $B$ would not contribute to the limiting evolution obtained below, see \cite{KolQuantFrac}.
\end{remark}

The main idea for obtaining sensible asymptotic limits (thus avoiding Zeno paradox)
is to enhance the interaction part $C$ of the Hamiltonians
 replacing it by the scaled version $C/\sqrt h$. Thus we choose the Hamiltonian in the form
\begin{equation}
\label{interham}
A= \begin{pmatrix}
H & -iL^*/\sqrt h   \\
iL/\sqrt h & 0
\end{pmatrix}
\end{equation}

 \begin{remark}
 The necessity of scaling is clear from Zeno's paradox. In principle one can suggest
 scaling $L$ as $L/h^{\al}$ with some $\al>0$. However, as calculations would show, choosing $\al>1/2$ would force $L$
 to disappear from the limit (measurement would have no effect in the limit), and choosing $\al<1/2$
 would force $H$ to disappear from the limit,
  only with $\al=1/2$ a sensible contributions from both $H$ and $L$ remain.
 \end{remark}

In the calculations below we shall use the following simple general small time asymptotic
formula for the evolutions $e^{-itA}$:
\[
e^{-itA}\ga e^{itA}
=(1-it A -\frac12 t^2 A^2)\ga (1+it A-\frac12 t^2 A^2)+O(t^3 \|A\|^3)
\]
\[
=\ga-it [A,\ga]-\frac12 t^2 A^2\ga -\frac12 t^2 \ga A^2+t^2 A\ga A+O(t^3 \|A\|^3)
\]
\begin{equation}
\label{smalltimegroup}
=\ga-it [A,\ga]+t^2 (A\ga A-\frac12 \{A^2,\ga\})+O(t^3\|A\|^3).
\end{equation}

To use this formula, for a mixed state $\rho$ in $\HC$, we calculate its ingredients as follows:

\[
\left[A, \rho \otimes \Om \right]
=\begin{pmatrix} [H,\rho] & i\rho L^*/\sqrt h \\ iL\rho/\sqrt h & 0 \end{pmatrix},
\]
\[
A (\rho \otimes \Om) A
=\begin{pmatrix} H\rho H & -i H\rho L^*/\sqrt h\\ iL\rho H/\sqrt h & L\rho L^*/h \end{pmatrix},
\]
\[
A^2 =\begin{pmatrix} H^2+L^*L/h & -iH L^* /\sqrt h \\ iL H/\sqrt h & LL^*/h \end{pmatrix},
\]
\[
\{A^2, \rho \otimes \Om\}=\begin{pmatrix} \{H^2+L^*L/h,\rho\} & -i\rho H L^*/\sqrt h \\ iL H \rho/\sqrt h  & 0 \end{pmatrix}.
\]
Using \eqref{smalltimegroup},
we get the approximation
\begin{equation}
\label{eqdressedrho1}
e^{-ihA} (\rho\otimes \Om) e^{ihA}
=\begin{pmatrix} \rho -ih[H, \rho] -\frac12 h\{L^*L,\rho\} &  \sqrt h \rho L^* \\ \sqrt h L\rho & hL\rho L^* \end{pmatrix}
+O(h^{3/2}),
\end{equation}
which is the key formula for what follows.

Let us stress for clarity that in this concise formula it is meant
 that $L^*L=\sum_j L_j^*L_j$ and $L\rho L^*$ is the square matrix
with the entries $L_k \rho L_j^*$.

\subsection{Counting observation}
\label{seccountcase}

As it turns out (may be not intuitively quite clear why),
the limiting processes are quite different depending on whether
the set of eigenvectors of the operator $R$ (observable of the measuring devise) includes the vacuum or not.
Let us start with the first case, when $e_0$ is an eigenvector of $R$.
Then the unitary operator $U$ in $\C^m$ transferring the standard basis
of $\C^m$ to the basis of eigenvectors of $R$ can be chosen in the form

\begin{equation}
\label{eqdiagonproj}
U= \begin{pmatrix} 1 & 0 \\ 0 & u \end{pmatrix},
\end{equation}
with a unitary operator $u$ in $\C^{m-1}=e_0^{\perp}$ -- the orthogonal complement of the vacuum vector.

It is then seen that changing the standard basis of $\C^m$ to the new one via the transformation $U$ would
transfer $A$ to the operator of the same form with the same $H$, but with the vector $uL$ instead of $L$.
This means that in the case, when $e_0$ is an eigenvector of $R$, we can consider $R$ to be diagonal without
essential loss of generality, as more general case would result only in a linear
transformation of $L$. Therefore, let us assume that $R=\sum_j \la_j P_j$ is diagonal, so that each $P_j$
is the projection on the standard basis vector $e_j$.

 Then the non-normalized new states are
\[
\tilde \rho_0=(I\otimes P_0) e^{-ihH} (\rho \otimes \Om) e^{ihH} (I\otimes P_0)
= \rho -ih[H, \rho] -\frac12 h\{L^*L,\rho\},
\]
\[
\tilde \rho_j=(I\otimes P_j) e^{-ihH} (\rho \otimes\Om) e^{ihH} (I\otimes P_j)
= hL_j\rho L_j^*, \quad j>0,
\]
occurring with the probabilities
\[
p_0=1-h \, {\tr} \, (L^*L \rho), \quad p_j=h \, {\tr}\, (L_j^*L_j \rho), j>0.
\]
Thus equation \eqref{eqMarkchain2} becomes
\[
U_h f(\rho)=\E f(\rho_h)=(1-h \, {\tr} \, (L^*L \rho))f
\left( \frac{\rho -ih[H, \rho] -\frac12 h\{L^*L,\rho\}}{1-h \, {\tr} \, (L^*L \rho)}\right)
\]
\begin{equation}
\label{eqMarkchaindiag}
+h \sum_{j>0} {\tr} (L_j^*L_j \rho) f\left(\frac{L_j \rho L_j^*}{{\tr}\, (L_j^*L_j \rho)}\right).
\end{equation}

Aiming at using Proposition \ref{propconvsemigr} (ii) we are looking for
the limit of the operator $(U_h-1)/h$ for $h \to 0$.

Denoting $T_j=  {\tr}\, (L_j^*L_j \rho)$ for $j>0$ and $T=\sum_j T_j$ we obtain
\[
\frac{U_h-1}{h} f(\rho)=\frac{1}{h} (1-hT)
\left[f \left( \frac{\rho -ih[H, \rho] -\frac12 h\{L^*L,\rho\}}{1-h T}\right)-f(\rho)\right]
\]
\[
+\sum_{j>0} T_j \left[f(\frac{L_j\rho L_j^*}{T_j})-f(\rho)\right].
\]
Since
\[
\frac{\rho -ih[H, \rho] -\frac12 h\{L^*L,\rho\}}{1-h T}-\rho
=\frac{-ih[H, \rho] -\frac12 h\{L^*L,\rho\}+h T \rho}{1-h T},
\]
we get by the Taylor expansion that, up to terms of order $\sqrt h$,
\[
\frac{U_h-1}{h} f(\rho)\approx \AC_{count} f(\rho)
\]
with
\begin{equation}
\label{eqjumpgenerrep}
\AC_{count}f(\rho)=-(f'(\rho), i[H, \rho] +\frac12 \{L^*L,\rho\}-\rho T)
+\sum_j T_j \left[f(\frac{L_j\rho L_j^*}{T_j})-f(\rho)\right],
\end{equation}
which is exactly operator $\AC_{count}$ from \eqref{eqjumpgener}.

Summarising, we conclude the following.

\begin{lemma}
\label{lemmaongencount}
Under the setting considered,
\begin{equation}
\label{eqjumpgener0}
\frac{U_h-1}{h} f \to \AC_{count}f
\end{equation}
for $f\in C^1_{Gat*}(X)$ (see Theorem \ref{JumpWellPDE1} for this notation),
with $\AC_{count}$ given by \eqref{eqjumpgenerrep}. Moreover,
\begin{equation}
\label{eqjumpgener1}
\|\frac{U_h-1}{h} f -\AC_{count}f\|\le \sqrt h \ka \|f\|_{C^2(S(\HC_0))}
\end{equation}
for $f\in C^2(S(\HC_0))$ and a constant $\ka$.
\end{lemma}

We can now complete our derivation of quantum filtering equations for continuous observations.

\begin{theorem}
\label{thBeleqcount}

(i) The scaled discrete semigroups $(U_h)^{[t/h]}$ converge to the semigroup $T_t$,
generated by operator \eqref{eqjumpgener} or \eqref{eqjumpgenerrep}
(according to Theorems \ref{JumpWellPDE} and \ref{JumpWellPDE1})
as $h\to 0$, so that the corresponding
processes converge in distribution; the scaled semigroups $T_t^h$ generated
 \eqref{eqMarkchain3} converge to the semigroup $T_t$, as $h\to 0$,
so that the corresponding
processes converge in distribution.

(ii) The rates of convergence can be given for smooth functions:
\begin{equation}
\label{eq1thBeleqcount}
\|(U_h)^{[t/h]}f -T_tf\| \le \sqrt h t M e^{Mt} \|f\|_{C^2(S(\HC))},
\end{equation}
and
\begin{equation}
\label{eq2thBeleqcount}
\|T_t^h f -T_tf\| \le \sqrt h t Me^{Mt} \|f\|_{C^2(S(\HC))}.
\end{equation}
with a constant $M$ depending on the norms of $H$ and $L$.
\end{theorem}

\begin{proof}
(i) This is a consequence of Lemma \ref{lemmaongencount}, Theorem \ref{JumpWellPDE1}
 and the general result on the convergence
of semigroups in terms of the convergence of their generators, see e.g. Theorems 19.27 and  19.28 of \cite{Kallen}.

(ii) This is a consequence of Lemma \ref{lemmaongencount}, Theorem \ref{JumpWellPDE}, Proposition \ref{propconvsemigr},
and the observation that \eqref{eq5propconvsemigr}
holds here with  the triple of spaces $C^2(S(\HC))\subset C^1(S(\HC)) \subset C(S(\HC))$.
\end{proof}

In order for Part (ii) of the Theorem to have a content, it is necessary for the classes
of smooth functions  $C^2(S(\HC))$  to be dense in $C(S(\HC))$. As we noted already, this is not clear in general.
Thus the rates of convergence in (ii) really make sense only either for finite-dimensional $\HC$ or for the case
of unitary operator $L$ (as seen from Proposition \ref{JumpWellPDE2}).

\subsection{Diffusive and mixed observation}
\label{secdifcase}

Let us turn to the second case, when $e_0$ is not an eigenvector of $R$.
Then the number $k$ of eigenvectors of $R$ having a non-vanishing projection
on $e_0$ is not less than $2$. Ordering eigenvectors of $R$ in such a way that
these $k$ eigenvectors take the first $k$ places, we conclude that in the
coordinate representation of the eigenvectors of $R$
\[
r_j=\sum_{l=0}^{m-1} r_j^l e_l,
\]
we have $r_j^0\neq 0$ exactly for $j=0, \cdots, k-1$.
Just for simplicity of writing assume that all coordinates here are real.

Recalling \eqref{eqdressedrho1} we derive the
new non-normalised states
\[
\tilde \rho_j =(r_j, e^{-ihA} (\rho\otimes \Om) e^{ihA} r_j)
\]
\[
=(r_j^0)^2 (\rho -ih[H, \rho] -\frac12 h\{L^*L,\rho\})
+r^0_j \sqrt h (\rho \tilde L_j^*+\tilde L_j \rho)+h\tilde L_j \rho \tilde L_j^*
+O(h^{3/2})
\]
where
\[
\tilde L_j= \sum_{l=1}^{m-1} r^l_j L_l.
\]

Introducing, as above, the notations $T_j={\tr}\, (\tilde L_j^* \tilde L_j \rho)$ for $j>0$ and $T=\sum_j T_j$
and denoting $\Om_j=  {\tr} (\rho \tilde L_j^* + \tilde L_j\rho)$ and observing that
$T=\sum_j {\tr} \, L_j^* L_j\rho $ we obtain
the probabilities of the occurrence of these states (up to higher order in $h$):
\[
p_j={\tr} \, \tilde \rho_j=(r_j^0)^2  (1-hT)+\Om_j r^0_j \sqrt h  +h T_j,
\]

In particular, for $j\ge k$,
\[
\tilde \rho_j =h\tilde L_j \rho \tilde L_j^*, \quad p_j=hT_j, \quad \rho_j=\frac{\tilde L_j \rho \tilde L_j^*}{T_j}.
\]

Next, for arbitrary numbers $a,b,c$, one can write up to terms of order $t$, that
\[
\frac{1}{a+b\sqrt t +ct}
=\frac{1}{a} \frac{1}{1+(b/a) \sqrt t+(c/a) t}
=\frac{1}{a}(1-(b/a) \sqrt t-(c/a) t+(b/a)^2 t).
\]
Consequently, with this order of approximation, it follows for $j<k$ that
\[
\frac{1}{p_j}=\frac{1}{(r_j^0)^2}
\left[1- \frac{\Om_j}{r^0_j} \sqrt h  +\left(T+\frac{\Om_j^2-T_j}{(r^0_j)^2}\right)h \right],
\]
and the normalised states are
\[
\rho_j=\frac{\tilde \rho_j}{p_j}=[\rho -ih[H, \rho] -\frac12 h\{L^*L,\rho\})
+\frac{\sqrt h}{r^0_j} (\rho \tilde L_j^*+\tilde L_j \rho)+\frac{h}{(r^0_j)^2}\tilde L_j \rho \tilde L_j^*]
\]
\[
\times \left[1- \frac{\Om_j}{r^0_j} \sqrt h  +\left(T+\frac{\Om_j^2-T_j}{(r^0_j)^2}\right)h \right]
\]
\[
=\rho +  \frac{\sqrt h}{r^0_j}  (\rho \tilde L_j^* + \tilde L_j\rho-\Om_j \rho)+tB_j
\]
with
\[
B_j=-i[H, \rho] -\frac12 \{L^*L,\rho\}+ T \rho
 +(r_j^0)^{-2} (\tilde L_j\rho \tilde L_j^*-(\rho \tilde L_j^* + \tilde L_j\rho) \Om_j -T_j\rho +\Om_j^2  \rho).
\]

Therefore, the first $k$ terms in the expression for $(U_h-1)/h f(\rho)$, up to orders $h$ become
\[
\sum_{j=0}^{k-1} \frac{1}{h} ((r_j^0)^2 +\Om_j r^0_j \sqrt h)
[f(\rho +  \frac{\sqrt h}{r^0_j}  (\rho \tilde L_j^* + \tilde L_j\rho-\Om_j \rho)+hB_j)-f(\rho)]
\]
\[
=\sum_j  \frac{1}{h} ((r_j^0)^2 +\Om_j r^0_j \sqrt h)
f'(\rho) \left(\frac{\sqrt h}{r^0_j}  (\rho \tilde L_j^* + \tilde L_j\rho-\Om_j \rho)+hB_j\right)
\]
\[
+\frac12 \sum_j  \left(\rho \tilde L_j^* + \tilde L_j\rho-\Om_j \rho,
f''(\rho)  (\rho \tilde L_j^* + \tilde L_j\rho-\Om_j \rho)\right).
\]
The terms of order $1/\sqrt h$ cancel, because the coefficient at $h^{-1/2}$ equals
\[
\sum_{j=0}^{k-1} r_j^0   (\rho \tilde L_j^* + \tilde L_j\rho-\Om_j \rho)
\]
\[
=\sum_{l=1}^{k-1} \sum_{j=0}^{k-1} r^0_j r^l_j [\rho L_l +L_l \rho -{\tr}\, (\rho L_l+L_l \rho)]=0,
\]
due to the orthogonality of the vectors $r_j$. Consequently, disregarding terms of order $\sqrt h$, we have
\[
\sum_{j=0}^{k-1} \frac{1}{h} p_j (f(\rho_j)-f(\rho))
=\sum_j \left(f'(\rho), \Om_j  (\rho \tilde L_j^* + \tilde L_j\rho-\Om_j \rho)+(r^0_j)^2B_j\right)
\]
\[
+\frac12 \sum_j  \left(\rho \tilde L_j^* + \tilde L_j\rho-\Om_j \rho,
f''(\rho)  (\rho \tilde L_j^* + \tilde L_j\rho-\Om_j \rho)\right)
\]
\[
=\left(f'(\rho), -i[H, \rho] -\frac12 \{L^*L,\rho\} +\sum_{j=0}^{k-1}\tilde L_j \rho \tilde L_j^*
+\sum_{j=k}^{m-1} T_j \rho \right)
\]
\[
+\frac12 \sum_j  \left(\rho \tilde L_j^* + \tilde L_j\rho-\Om_j \rho,
f''(\rho)  (\rho \tilde L_j^* + \tilde L_j\rho-\Om_j \rho)\right).
\]
Noticing that
\[
L^*L=\sum_{j=1}^{k-1} L_j^*L_j=\tilde L^* \tilde L=\sum_{j=0}^{k-1}\tilde L_j^* \tilde L_j,
\]
we see that everything in the last expression rewrites in terms of $\tilde L$.

Including the terms with $j\ge k$ yields
\[
\lim_{h\to 0} \frac{U_h-1}{h} f(\rho)=\sum_{j=k}^{m-1} T_j \left(f(\frac{\tilde L_j \rho \tilde L_j^*}{T_j})-f(\rho)\right)
\]
\[
+\left(f'(\rho), -i[H, \rho] -\frac12 \{\tilde L^*\tilde L,\rho\}
+\sum_{j=0}^{k-1}\tilde L_j \rho \tilde L_j^*
+\sum_{j=k}^{m-1} T_j\right)
\]
\begin{equation}
\label{eqapprdif}
+\frac12 \sum_{j=0}^{k-1}  \left(\rho \tilde L_j^* + \tilde L_j\rho-\Om_j \rho,
f''(\rho)  (\rho \tilde L_j^* + \tilde L_j\rho-\Om_j \rho)\right),
\end{equation}
which is exactly the operator \eqref{eqmixgener} (for $\tilde L$).

Turning from $m-1$ operators $L_j$
to $m$ operators $\tilde L_j$ produces natural degeneracy meaning
that one can expect that often the number of terms in the last expression
can be made less than $k$. For instance, if $k=2$, $r_j=e_j$ for $j>1$ and $r_0,r_1$ belong to the
space generated by $e_0,e_1$, then $\tilde L_0=r_0^1 L_1$ and $\tilde L_1=r^1_1 L_1$
are proportional and
\[
\sum_{j=0}^{k-1} \tilde L_j \rho \tilde L_j^*=L_1\rho L_1^*,
\]
so that the corresponding diffusive part of \eqref{eqapprdif}
can be written as a single term (not as the sum of two terms):
\[
\sum_{j=0}^1 \frac{1}{h} p_j (f(\rho_j)-f(\rho))
=f'(\rho) \left(-i[H, \rho] -\frac12 \{\tilde L_1^*\tilde L_1,\rho\} +\tilde L_1 \rho \tilde L_1^* \right)
\]
\[
+\frac12 (\rho \tilde L_1^* + \tilde L_1\rho-{\tr} \, (\rho \tilde L_1^* +\tilde L_1 \rho) \rho) f''(\rho)
 (\rho \tilde L_1^* + \tilde L_1\rho-{\tr}\, (\rho \tilde L_1^* +\tilde L_1 \rho)  \rho).
\]

In \cite{KolQuantFrac} we suggested a different way of organising mixed observations via
explicit different channels, where the number of terms in the corresponding generator
is explicitly fixed from the initial model.

It is remarkable that the expression for the generator is continuous
in $r_j^0$ and has a finite limit as these coefficients
tend to zero, but this limit does not equal the generator
obtained in the case when all but one of these coefficients vanish.
The mere fact that they do not vanish creates the diffusive (second order) term in the generator.

Everything is ready for the main result of this Section:

\begin{theorem}
\label{thBeleqdif}

Let the operators $\tilde L_j$ be unitary for $j\ge k$. Then

(i) The scaled discrete semigroups $(U_h)^{[s/h]}$ converge to the semigroup $\Phi^{mix}_s$
from Theorem \ref{nonLinsdesemmix},
as $h\to 0$, so that the corresponding
processes converge in distribution,
 with the following rates of convergence:
\begin{equation}
\label{eq1thBeleqdif}
\|(U_h)^{[t/h]} -\Phi^{mix}_tf\| \le \sqrt h t M e^{Mt} \|f\|_{C^4(S(\HC))},
\end{equation}
with a constant $M$.

(ii) The scaled semigroups $T_t^h$ converge to the semigroup $\Phi^{mix}_t$,
as $h\to 0$, so that the corresponding
processes converge in distribution, with the following rates of convergence:
\begin{equation}
\label{eq2thBeleqdif}
\|T_t^hf -\Phi_t^{mix}f\| \le \sqrt h t M e^{Mt} \|f\|_{C^3(S(\HC))}.
\end{equation}
\end{theorem}

\begin{proof}
This is a consequence of \eqref{eqapprdif}, Theorem \ref{nonLinsdesemmix},
Proposition \ref{propconvsemigr}, and the observation that \eqref{eq5propconvsemigr}
holds here with  the triple of spaces $C^4(S(\HC))\subset C^2(S(\HC)) \subset C(S(\HC))$.
\end{proof}

\begin{remark}
As usual, in Theorems \ref{thBeleqcount} and \ref{thBeleqdif} above
 we assumed that $L$ was bounded and $H$ self-adjoint, where the case of unbounded $H$
is reduced to a non-autonomous time-dependent setting to operators $L$.
\end{remark}

\subsection{Appendix: convergence of semigroups}
\label{appsem}

Here we collect the results on the convergence of Markov semigroups, which form
the theoretical basis for our derivations of the filtering equations.

It is well known that the convergence of the generators on the core of the limiting generator
implies the convergence of semigroups. We shall use a version of this result with the rates,
namely the following result, given in Theorem 8.1.1 of \cite{Kolbook11}.

\begin{prop}
\label{propconvsemigr}

 Let $F_t=e^{tL}$ be a strongly continuous semigroup in a Banach space $B$ with a norm $\|.\|_B$,
 generated by an operator $L$,
having a core $D$, which is itself a Banach space with a norm $\|.\|_D\ge \|.\|_B$ so that $L\in \LC(D,B)$.
Let $F_t$ be also a bounded semigroup in $D$
such that $\|F_t\|_{D\to D} \le C_D(T)$ with a constant $C_D(T)$ uniformly for $t\in [0,T]$.

(i) Let $F_t^h$, $h>0$,  be a family of strongly continuous contraction semigroups in a Banach space $B$
with bounded generators $L_h$ such that
\[
\|L_hf-Lf\|_B \le \ep_h \|f\|_D
\]
for all $f\in D$ and some $\ep_h$ such that $\ep_h\to 0$ as $h\to 0$.
 Then the semigroups $F_t^h$ converge strongly to the semigroup $F_t$, as $h\to 0$, and
\begin{equation}
\label{eq1propconvsemigr}
\|F_t^hf -F_tf\|_B \le t \ep_h C_D(T)\|L\|_{D\to B}.
\end{equation}

(ii) Let $U_h$ be a family of contractions in $B$ such that
\begin{equation}
\label{eq2propconvsemigr}
\|\left(\frac{U_h-1}{h} -L\right)f\|_B \le \ep_h \|f\|_D,
\end{equation}
and
 \begin{equation}
\label{eq3propconvsemigr}
\|\left(\frac{F_h-1}{h} -L\right)f\|_B \le \ka_h \|f\|_D,
\end{equation}
with $\ep_h \to 0$ and $\ka_h\to 0$, as $h\to 0$.
Then the scaled discrete semigroups $(U_h)^{[t/h]}$ converge to the semigroup $F_t$
and moreover
 \begin{equation}
\label{eq4propconvsemigr}
\sup_{s\le t}\|(U_h)^{[s/h]} -F_sf\|_B \le (\ka_h+\ep_h)t C_D(T)\|f\|_D.
\end{equation}
\end{prop}

Additional condition \eqref{eq3propconvsemigr} makes working with discrete approximation a bit more subtle,
than with the continuous chain approximations. Effectively to get \eqref{eq3propconvsemigr}
one needs a deeper regularity. Namely one should have another core $\tilde D$ such that $D\subset \tilde D\subset B$
with $L\in \LC(D,\tilde D) \cap \LC(\tilde D,B)$. In this case
\begin{equation}
\label{eq5propconvsemigr}
\|\left(\frac{F_h-1}{h} -L\right)f\|_B \le h \|L\|_{D,\tilde D} \|L\|_{\tilde D,B}\|f\|_D.
\end{equation}

\section{Quantum LLN (propagation of chaos) for continuously observed quantum systems}
\label{secLLN}

Here we derive the dynamic law of large numbers for continuously observes systems of quantum particles
that is described by the mean-field Belavkin filtering equations, which were initially introduced in \cite{KolQuantLLN},
and \cite{KolQuantMFG}.
The exposition is based essentially on author's papers \cite{KolQuantMFGCount}, \cite{KolQuantLLN},
\cite{KolQuantMFG}, \cite{Kol25a}. The only new moments is the part (2) of Theorem \ref{thmynonlinSch}
devoted to the case of interaction specified by a multiplication operator, which is an essential improvement
to the result of \cite{KolQuantMFG} that was given with hardly verified conditions. As was already mentioned
in introduction, a similar improvement to this result of \cite{KolQuantMFG} was given in the recent
important contribution \cite{Bouard}.

This section is organised as follows. In subsection \ref{secMFeq} we provide the well-posedness result
for the mean-field Belavkin quantum filtering equations. A brief subsection \ref{operinter} makes clear
the two facets  of an interaction operator
$A$ that can be considered either as an integral operator in $L^2(X^2)$ or as a linear mapping in the space $\LC$
of bounded operators in $\HC$, these actions being linked by the important identity
\eqref{eqmaincancel}. Subsection \ref{secmynonlinSchrod} gives a detailed proof of the main
result of this section, Theorem \ref{thmynonlinSch}, establishing the propagation of chaos property
and the dynamic law of large numbers for continuously observed quantum systems under
diffusive type observations. A brief subsection \ref{extcont} extends this result to
the case of controlled Hamiltonians, which is crucial for applications to quantum games.
 But for the proof of this extension we refer to the original paper \cite{KolQuantMFG}.
Subsection \ref{seccounting} provides the propagation of chaos result for counting observations.
Appendices deal with some technical issues.

\subsection{Stochastic master equations for mean-field interacting particles}
\label{secMFeq}

In \cite{KolQuantLLN} and \cite{KolQuantMFG}, the author derived the
effective quantum filtering equations for the quantum law of large number limit of interacting particles
under continuous measurement. As above, these equations can be written either for pure states
as a new kind of stochastic nonlinear Schr\"odinger equation, or for mixed states, as stochastic master equations for
 mean-field interacting particles, which can be looked at as an infinite-dimensional complex McKean-Vlasov
 diffusion in the space of positive trace-class operators. These limiting equations provide the forward
 part for the forward-backward system of equations governing the quantum mean-field games.
 We start with the well-posedness of these effective quantum filtering equations,
 a derivation being presented in the next section.

 The stochastic master equations for mean-field interacting particles
 can be formally obtained by adding an interaction term into the Hamiltonian.
 Namely, equation \eqref{Lindstoch} enhanced by mean-field interaction takes the form

 \[
d\ga(t)=-i[H,\ga(t)] \, dt -i[A(\bar \eta(t)), \ga(t)] , dt +\LC_L \ga(t) \, dt
\]
 \begin{equation}
\label{LindstochnewBel}
+(L\ga(t)+\ga(t) L^*) dY(t), \quad \eta(t) =\E (\ga(t)/{\tr}\, \ga (t)),
\end{equation}
which can be called the {\it mean-field Belavkin's equation}.
Here, $H$, $L$ are as above, $Y(t)$ is a $n$-dimensional BM,
the expectation $\E$ is with respect to $Y$ and
\[
A: \nu \to A(\nu)
\]
is a linear mapping in the space of bounded linear operators in $\HC$.
Not aiming at the most general situations
we shall assume that $A$ satisfies one of the two assumptions:
either $A$ is a
bounded linear mapping $\HC_s^2 \to \HC_s^2$ so that
\begin{equation}
\label{eqinterterm}
\|A(\nu)\|_{\HC^2_s} \le C_A \|\nu\|_{\HC^2_s}
\end{equation}
with a constant $C_A$, or $A$ is a bounded mapping from the trace-class operators to bounded
operators so that
\begin{equation}
\label{eqinterterm1}
\|A(\nu)\| \le C_A \, {\tr}\, |\nu| =C_A\|\nu\|_{\HC^1}
\end{equation}
with a constant $C_A$,

For instance, if $\HC$ is realised as the space $L^2(X, dx)$ of
square integrable functions on some Borel measure space $(X,dx)$,
$A$ satisfying \eqref{eqinterterm} can be given by an integral kernel
$A(x,y;x',y')$ so that, for $\nu\in \HC^2_s$ given by a kernel $\nu(x,y)$,
$A(\nu)$ is the integral operator in $L^2(X,dx)$ with the integral kernel
\begin{equation}
\label{eqinterterm2}
A(\nu)(x;y)=\int_{X^2} A(x,y;x',y')\nu (y,y') \, dydy'.
\end{equation}
In this case
\[
C_A^2=  \int_{X^4} |A(x,y;x',y')|^2 dx dy dx'dy'.
\]
On the other hand, $A$ satisfying \eqref{eqinterterm1} can be given by a
bounded function $A(x,y)$ (interaction potential) so that, for $\nu \in \HC^1_s$ given by a kernel $\nu(x,y)$,
$A(\nu)$ is the operator of multiplication by the function
$\int A(x, y)\nu(y,y) \, dy$.
In this case
\[
C_A= \sup_{x,y} |A(x,y)|.
\]

Notice that the case of multiplication operator by function $A(x,y)$ can be formally considered as
the integral operator with the singular kernel
\begin{equation}
\label{multint}
A(x,y)\de (x-x')\de (y-y')=A(x',y')\de (x-x')\de (y-y').
\end{equation}

Notice also that \eqref{eqinterterm} implies \eqref{eqinterterm1} for $\nu$ of trace class, because
\[
 \|\nu\|_{\HC^2_s} <  \|\nu\|_{\HC^1_s}.
\]

Similarly, equation \eqref{Lindstochnorm1} enhanced by a mean-field interaction takes the form

\[
d\rho(t)=-i[H,\rho(t)] \, dt -i[A(\bar \eta(t)), \rho(t)] , dt +\LC_L \rho(t) \, dt
\]
\begin{equation}
\label{eqmainnonlinBel}
+[L\rho(t)+\rho(t) L^*-\rho(t)\, {\tr} \, (L\rho(t)+\rho(t) L^*) ] dB(t),
\quad \eta (t) =\E \rho (t),
\end{equation}
with a $n$-dimensional Brownian motion $B(t)$.

As in the case without interaction, the same link between equations \eqref{eqmainnonlinBel}
and \eqref{LindstochnewBel} holds. Namely, as one checks by Ito's formula, (i)
if $\ga(t)$ satisfies \eqref{LindstochnewBel}, then $\rho(t)=\ga(t)/{\tr}\, \ga(t)$ satisfies
 \eqref{eqmainnonlinBel}, with $B$ and $Y$ connected via \eqref{outputinnovation}, and (ii) if
 $\rho(t)$ satisfies \eqref{eqmainnonlinBel} and ${\tr} \, \ga(t)$ is chosen as a solution
 to \eqref{eqtrinnov}, then $\ga(t)={\tr}\, \ga (t) \rho(t)$ satisfies  \eqref{LindstochnewBel}.

 As above, we shall work with these McKean-Vlasov-type SDEs
 as with SDEs in the Hilbert space $\HC^2_s$.

 \begin{theorem}
\label{mainMFmaster2}
Let $\rho_0$ be a positive-definite operator of unit trace,
 $B(t)$ a Brownian motion and $A$ satisfy \eqref{eqinterterm} or \eqref{eqinterterm1}.
Then there exists a unique strong solution of equation \eqref{eqmainnonlinBel}
in $\HC^2_s$, with the initial data $\rho_0$
and such that all $\rho(t)$ are positive-definite operators of unit trace.
\end{theorem}

\begin{proof}
By employing the interaction representation
we can and will assume that $H=0$ without loss of generality.

Let $C_{\rho_0}^{1+}([0,T], \HC^1_s)$ be the space of continuous mapping
$\eta: [0,T] \to \HC^1_s$ such that $\eta(0)=\rho_0$ and all $\eta(t)$ are positive
trace class operators of trace not exceeding $1$. It is not difficult to see that
$C_{\rho_0}^{1+}([0,T], \HC^1_s)$ is a complete metric space,
considered as a closed subset of the Banach space of curves in $\HC^1_s$ with the norm
$\sup_{t\in [0,T]} \|\eta(t)\|_{\HC^1_s}$.

Let us define the mapping
\[
\Phi: C_{\rho_0}^{1+}([0,T], \HC^1_s)
\to C_{\rho_0}^{1+}([0,T], \HC^1_s)
\]
by the following rule. To an $\eta \in C_{\rho_0}^{1+}([0,T], \HC^1_s)$
let us assign the solution  of equation
 \[
d r(t)=-i[H,r(t)] \, dt-i[A(\bar \eta(t)), r(t)] , dt
+\LC_L r(t) \, dt
\]
\begin{equation}
\label{eqmainnonlinBela}
+[r(t) L^*+L r(t)-r(t) \, {\tr} (r(t)(L+L^*))] \, dB(t),
\end{equation}
with the initial condition $r_0=\rho_0$ and then define $(\Phi (\eta))(t)= \E r(t)$.
Clearly, $\rho(t)$ is the solution of the Cauchy problem for equation
\eqref{eqmainnonlinBel} with the initial data $\rho_0$ if and only if
$\eta=\E \, \rho$ is a fixed point of the mapping $\Phi$.

By \eqref{Lindstochnorm12} and \eqref{eqinterterm1},
\[
\| \E \, r_1(t)-\E \, r_2(t)\|_{\HC^1_s}
= {\tr}\, |\E \, r_1(t)-\E \, r_2(t)|\le {\tr}\, \E |r_1(t)-r_2(t)|
\]
\[
\le
\sqrt t C(t)\sup_{s\in [0,t]} \|A(\eta_1(s))-A(\eta_2(s))\|
\le \sqrt t C(t) C_A \sup_{s\in [0,t]}\|\eta_1(s)-\eta_2(s)\|_{\HC^1_s}.
\]
Hence, for sufficiently small $t$, the mapping $\Phi$ is a contraction and thus
has a unique fixed point. As usual, existence and uniqueness extends to arbitrary $t$
by iteration.
\end{proof}

The corresponding version of the filtering equation with mean-field interaction for pure states
is obtained from \eqref{eqqufiBnonlinsn} by adding the corresponding interaction term yielding the equation
(where we omit explicit dependence on $t$)
\[
d\phi=-[i(H-\langle L_S \rangle_{\phi} L_A)+A(\E (\bar \phi \otimes \phi))
+\frac12 (L-\langle L_S \rangle_{\phi})^*(L-\langle L_S \rangle_{\phi})]\phi \, dt
\]
\begin{equation}
\label{eqmainnonlinBelpu}
+ (L-\langle L_S \rangle_{\phi})\phi \, dB(t),
\end{equation}
with the same operator $A$ as above. This is another example of McKean-Vlasov type equation in a Hilbert space,
which can be also looked at as a special type of nonlinear stochastic Schr\"odinger equation.
Its well-posedness under the same assumptions of self-adjoint $H$, bounded $L$ and $A$ satisfying
\eqref{eqinterterm} or \eqref{eqinterterm1} can be obtained via fixed-point argument as above.
Details of this proof (with explicit bounds for growth and continuity) can be found in \cite{KolQuantLLN}.

Similarly, for the of counting observation, where we shall work only with unitary coupling operator $L$
equation  \eqref{eqBeleqcountunit1} enhanced by interaction takes the form
 \begin{equation}
\label{eqmainnonlinBelcount}
d\rho=(- i[H, \rho]  -i[A(\bar \eta(t)), \rho] , dt -\frac12 \{L^*L,\rho\}+ L\rho L^* ) \, dt
+\left(L\rho L^*-\rho\right) dM(t),
\end{equation}
in terms of the martingale $M(t)=N(t)-t$,  with $\eta(t) =\E \ga(t)$.

The following result is an analogue of Theorem \ref{mainMFmaster2} for counting observation.

 \begin{theorem}
\label{mainMFmaster3}
Let $L$ be unitary, $\rho_0$ a positive-definite operator of unit trace,
 $M(t)=N(t)-t$ with the standard Poisson process $N(t)$, and $A$ satisfy \eqref{eqinterterm} or \eqref{eqinterterm1}.
Then there exists a unique strong solution of equation \eqref{eqmainnonlinBelcount}
in $\HC^1_s$, with the initial data $\rho_0$
and such that all $\rho(t)$ are positive-definite operators of unit trace.
\end{theorem}

\begin{proof}
It is a complete analogue of Theorem \ref{mainMFmaster2}. Namely, we define the mapping
\[
\Phi: C_{\rho_0}^{1+}([0,T], \HC^1_s)
\to C_{\rho_0}^{1+}([0,T], \HC^1_s)
\]
by the following rule. To an $\eta \in C_{\rho_0}^{1+}([0,T], \HC^1_s)$
let us assign the solution  of the equation
\[
d r(t)=(- i[H, r(t)]  -i[A(\bar \eta(t)), r(t)] , dt -\frac12 \{L^*L, r(t)\}+ L r(t) L^* ) \, dt
+\left(L r(t) L^*- r(t)\right) dM(t),
\]
with the initial condition $r_0=\rho_0$ and then define $(\Phi (\eta))(t)= \E \, r(t)$.
Clearly, $\rho(t)$ is the solution of the Cauchy problem for equation \eqref{eqmainnonlinBelcount}
with the initial data $\rho_0$ if and only if
$\eta=\E \, \rho$ is a fixed point of the mapping $\Phi$.
The existence and uniqueness of a fixed point follows from
the Banach contraction principle, as above.
\end{proof}

\subsection{Operators of interaction}
\label{operinter}

Let $X$ be a Borel space with a fixed Borel measure that we denote $dx$.
For a sequence $\nu_m$ of operators in $L^2(X)$ we shall identify $\nu_m$
with the operator in $\HC^{\otimes N}=L^2(X^N)$
that acts on the $m$th coordinate of functions $f(x_1, \cdots, x_N)$.
If just one operator $O$ in $L^2(X)$ is given we shall denote $O_m$
the operator in $\HC^{\otimes N}=L^2(X^N)$ that acts as $O$ on the variable $x_m$.
Similarly, if $A$ is a bounded operator in $\HC^{\otimes 2}=L^2(X^2)$, then
$A_{jk}$ denote the operators on $\HC^{\otimes N}=L^2(X^N)$ that act as $A$ on
 the $j$th and $k$th coordinates of an $f(x_1, \cdots, x_N)\in L^2(X^N)$.

From now the interaction between particles will be specified by an integral operator
in $\HC^{\otimes 2}=L^2(X^2)$ given by the real integral kernel
$A(x,y;x',y')$, including the case of singular kernel \eqref{multint}.

Thus, on the one hand side,  $A$ acts in $L^2(X^2)$ as the integral operator
\begin{equation}
\label{intopL2}
f(x,y) \mapsto Af(x,y)=\int_{X^2} A(x,y;x',y')f (x',y') \, dx'dy',
\end{equation}
and, on the other hand, $A$ specifies a linear mapping in the space of bounded operators in $\HC$
so that, for an operator $\nu$ given by the kernel $\nu(x,y)$,
$A(\nu)$ is the integral operator in $L^2(X)$ with the integral kernel \eqref{eqinterterm2}:
\[
A(\nu)(x;y)=\int_{X^2} A(x,y;x',y')\nu (y,y') \, dydy'.
\]

These two facets of the kernel $A$ are linked by the following identity:
\begin{equation}
\label{eqmaincancel}
\ga_m A_{jm}\ga_m =\ga_m A_j(\overline{\ga_m}),
\end{equation}
where $\ga_m=\psi_m\otimes \bar \psi_m$ is an arbitrary collection of one-dimensional projectors.
In fact, the operator $\ga_m A_{jm}\ga_m$ acts as
\[
\ga_m A_{jm}\ga_m f(x_j,x_m)
\]
\[
=\int_{X^4} \psi_m (x_m)\bar \psi_m (z_m) A(x_j,z_m; x'_j,w_m)\psi_m(w_m)\bar \psi_m (x'_m)
f(x'_j,x'_m) dx'_jdx'_m dz_m dw_m,
\]
and the operator $\ga_m A_j(\overline{\ga_m})$ acts on $f(x_j,x_m)$ as
\[
(\ga_mA_j(\overline{\ga_m}) f)(x_j,x_m)
\]
\[
= \int_{X^4} \psi_m(x_m) \bar \psi_m(x'_m) A(x_j,z_m;x'_j,w_m)
\bar \psi_m(z_m) \psi_m(w_m) f(x'_j,x'_m) dz_m dw_m dx'_j dx'_m,
\]
and \eqref{eqmaincancel} follows.

We shall further assume that $A$
is self-adjoint and takes symmetric functions to symmetric. In terms of the kernel
these properties write down as
\begin{equation}
\label{eq1athmynonlinSch}
A(x,y;x'y')=A(y,x;y',x'), \quad A(x,y;x',y')=A(x',y';x,y).
\end{equation}

\subsection{Derivation of mean-field limit: diffusive measurement}
\label{secmynonlinSchrod}

Let $H$ be a self-adjoint operator in $\HC=L^2(X)$
and $A$ be an integral operator \eqref{intopL2} in $L^2(X^2)$.
Let us consider the quantum evolution of $N$ particles in $\HC^{otimes N}$
driven by the interaction Hamiltonian
\begin{equation}
\label{eqHambinaryinter0}
 H(N)f(x_1, \cdots , x_N)=\sum_{j=1}^N H_jf(x_1, \cdots , x_N)
 + \frac{1}{N}\sum_{i<j\le N} A_{ij}f(x_1, \cdots , x_N),
\end{equation}
where, as pointed out above, $H_j$ denotes the action of $H_j$ on the variable $x_j$
and $A_{ij}$ denotes the action of $A$ on the variables $x_i,x_j$.

Assume further that this quantum system is observed
via coupling with the collection of identical one-particle operators $L$.
For simplicity of notations we assume here
that $L$ is not vector-valued.
 Thus we consider the filtering equation of the type \eqref{eqqufiBnonlinsn}:
\[
d\Psi_N = \sum_j [i \langle (L_j)_S\rangle_{\phi} (L_j)_A
 -\frac12 (L_j-\langle (L_j)_S \rangle_{\Psi_N})^*(L_j-\langle (L_j)_S\rangle_{\Psi_N})]\Psi_N \, dt
\]
\begin{equation}
\label{eqmainNpartBelnonls}
-iH(N) \Psi_N \, dt+\sum_{j=1}^N (L_j-\langle (L_j)_S\rangle_{\Psi_N})\Psi_N \, dB^j(t),
\end{equation}
omitting the argument $t$ in $\Psi_N(t)$ for brevity, as we often do.

Clearly, if the initial condition is invariant under the permutation of variables, as we always assume,
the distribution of $\Psi_N(t)$ is also invariant under such permutations.

Our aim is to show that, as $N\to \infty$, the solutions of these equations
and the corresponding density matrices $\Ga_N=\Psi_N\otimes \overline{\Psi_N}$
are close to the product of the solutions of one-particle nonlinear stochastic
equations \eqref{eqmainnonlinBelpu} and \eqref{eqmainnonlinBel}.

Our analysis will be carried out via the extension of the method suggested
by Pickl in a deterministic case,
see \cite{Pickl} and  \cite{KnowlesPickl}, to the present stochastic framework.
In Pickl's approach the main measures of the deviation of the solutions $\Psi_{N,t}$ to $N$-particle
systems from the product of the solutions $\psi(t)$ to the Hartree equations are the following positive numbers
from the interval $[0,1]$:
\[
\al_N(t)=1-(\psi(t), \Ga_N(t) \psi(t)).
\]


In the present stochastic case, these quantities depend not just on the number of particles
in the product, but on the concrete choice of these particles. The proper stochastic analog
of the quantity $\al_N(t)$ is the collection of random variables
\begin{equation}
\label{eqforalpha}
\al_{N,j}(t)=  1-(\psi_{j,t}, \Ga_{N,t} \psi_{j,t})=1-{\tr}(\ga_{j,t} \Ga_{N,t})=1-{\tr}(\ga_{j,t} \Ga^{(j)}_{N,t}),
\end{equation}
where the latter equation holds by the definition of the partial trace.
Here $\ga_{j,t}$ is identified with the operator in $L^2(X^N)$ acting on the $j$th variable and
$\Ga^{(j)}_{N,t}$ denotes the partial trace of $\Ga_{N,t}$ with respect to all variables except for the $j$th.

Since the solutions to equations \eqref{eqmainNpartBeldensnonl}
and \eqref{eqmainnonlinpartBel1dens} preserve the set of operators with the unit trace,
equation \eqref{eqforalpha} rewrites as

\begin{equation}
\label{eqforalpha1}
\al_{N,j}(t)={\tr}((\1-\ga_j(t)) \Ga_N(t))={\tr}((\1-\ga_j(t)) \Ga^{(j)}_N(t)).
\end{equation}

Due to the i.i.d. property of the solutions to \eqref{eqmainnonlinco} and the symmetry of $\Psi_N(t)$,
the expectations $\E \al_N(t)=\E \al_{N,j}(t)$ are well defined (they do not
depend on a particular choice of particles).

Expressions $\al_{N,j}$ can be linked with the traces by the following
inequalities, due to Knowles and Pickl:
\begin{equation}
\label{ineqKnPi}
\al_{N,j}(t)\le {\tr} |\Ga^{(j)}_N(t)-\ga_j(t)|\le 2\sqrt{2\al_{N,j}(t)},
\end{equation}
see Lemma 2.3 from \cite{KnowlesPickl}. For our stochastic setting it follows that
\begin{equation}
\label{ineqKnPistoch}
\E \, \al_N(t)\le \E \, {\tr}\, |\Ga^{(j)}_N(t)-\ga_j(t)|=\E \, \|\Ga^{(j)}_N(t)-\ga_j(t)\|_{\HC^1_s}
\le 2\sqrt{2 \E \, \al_N(t)},
\end{equation}

\begin{theorem}
\label{thmynonlinSch}
Let $A$ be a self-adjoint bounded operator in $L^2(X^2)$, which is either

(1) a Hilbert-Schmidt operator with
the kernel $A(x,y;x',y')$ such that \eqref{eq1athmynonlinSch} and
\begin{equation}
\label{eq1thmynonlinSch}
\|A\|^2_{HS}=\int_{X^4} |A(x,y;x',y')|^2 \, dx dy dx'dy' <\infty
\end{equation}
hold, or

(2) an operator of multiplication by a function $A(x-y)$
such that $A\in C(\R^d)\cap L^p(\R^d)$
with some $p>1$.

 Let $\psi_j=\psi_j(t)$ be solutions to the equations
\[
d\psi_j(x) =-i[H  +A(\bar \eta)
-\langle L_S \rangle_{\phi} L_A]  \psi_j(x) \, dt
\]
\begin{equation}
\label{eqmainnonlin}
-\frac12 (L-\langle L_S \rangle_{\psi_j})^* (L-\langle L_S \rangle_{\psi_j}) \psi_j(x)\,dt
+(L-\langle L_S \rangle_{\psi_j})\psi_j \, dB^j(t),
\end{equation}
with $\eta=\eta(t)=\E \psi_j(t)\otimes \bar \psi_j(t)$ and with the i.i.d. initial conditions
$\psi_{j,0}$, $\|\psi_{j,0}\|=1$, and $\ga_j(t)=\psi_j(t)\otimes \bar \psi_j(t)$.
Let $\Psi_N(t)$ be the solution to the $N$-particle equation \eqref{eqmainNpartBelnonls}
with $H(N)$ of type \eqref{eqHambinaryinter0}, with some symmetric (with respect to any
permutation of arguments) initial condition
$\Psi_{N,0}$, $\|\Psi_{N,0}\|_2=1$ such that
\[
\al_N(0)= \al_{N,j}(0)=1-\E \, {\tr} (\ga_{j,0} \Ga_{N,0})=1-\E \, {\tr} (\ga_{j,0} \Ga^{(j)}_{N,0})
\]
are equal for all $j$. The main example of such initial condition is of course the product
\[
\Psi_{N,0}=\prod \psi_{j,0}(x_j),
\]
where $\al_{N,j}(0)=0$ for all $j$.

In case (2) we assume additionally that the solutions $\psi_j$ belong to the space $L^{2q}(\R^d)$,
where $q=p/(p-1)$ so that
\[
M_T=\sup_{t\in [0,T]} \E \|\psi_j(t)\|_{2q}<\infty.
\]
Then, in case (1),
\[
\E \al_N(t) \le \exp\{ (6\|A\|+28 \|L\|^2)t\} \E \al_{N}(0)
\]
\begin{equation}
\label{eq2thmynonlinSch}
+7 \left(\exp\{ (6\|A\|+28 \|L\|^2)t\}-1\right)\|A\|_{HS} N^{-1/2},
\end{equation}
and in case (2)
\[
\E \al_N(t) \le \exp\{ (6\|A\| +28 \|L\|^2)t\} \E \al_{N}(0)
\]
\begin{equation}
\label{eq3thmynonlinSch}
+\left(\exp\{ (6\|A\|+28 \|L\|^2)t\}-1\right)
\left(\frac{2}{N}\|A\|+ 4 M_T^2 \|A\|_p N^{-(q-1)/q}\right).
\end{equation}

\end{theorem}

\begin{remark}
As was shown in \cite{KolQuantMFG}, if $L_j=-L_j^*$, then $\|L\|$ does not enter the estimates above.
\end{remark}

\begin{remark}
By \eqref{eqforalpha1} it follows that if $\al_{N}(0) \to 0$, as $N\to \infty$ (for instance if $\al_{N}(0)=0$),
then $\E \, {\tr} |\Ga^{(j)}_{N,t}- \ga_{j,t}| \to 0$, as $N\to \infty$.
\end{remark}

\begin{remark}
The assumption that the solution $\psi$ belongs to higher $L^p$-spaces, $p>2$ (formulated as $M_T<\infty$ above),
is standard for the analysis of
nonlinear Schr\"odinger equations, see e.g. \cite{KnowlesPickl} and examples there. The most basic examples
comes from Sobolev's embedding. If solutions are regular enough, for instance, they belong to $W^2(\R^d)$,
as in Theorems \ref{thclassichamfilt}, \ref{thclassichamfiltmix}, they are automatically in $L^p$ with a $p>2$.
\end{remark}

\begin{proof}

{\it Step 1.}

Using definition \eqref{eqforalpha}  and Ito's product rule we derive that
\[
d \al_{N,j}(t)=-{\tr} (d\Ga_N(t) \ga_j(t))-{\tr} (\Ga_N(t) \, d \ga_j(t))-{\tr} (d\Ga_N(t) \, d \ga_j(t))
\]
\begin{equation}
\label{eqforalphaconder}
=(C_j+D_j)dt+\sum_k F_{jk} dB^k(t),
\end{equation}
where
\[
D_j={\tr} \bigl[ \sum_k (\frac12 L^*_k L_k\Ga_N +\frac12 \Ga_N L_k^*L_k-L_k \ga_N L_k^*)\ga_j
+\Ga_N (\frac12 L^*_j L_j \ga_j +\frac12 \ga_j L_j^*L_j-L_j \ga_j L_j^*)
\]
\[
-(\Ga_N L_j^*+L_j \Ga_N -\Ga_N \, {\tr} (\Ga_N(L_j^*+L_j)))(\ga_j L_j^*
+L_j \ga_j -\ga_j \, {\tr} (\ga_j(L_j^*+L_j)))\bigr],
\]
is the part depending on $L$ and
\[
C_j= i \, {\tr} ([H_j+A_j(\bar \eta),\ga_j]\Ga_N)
 +i \, {\tr} (\ga_j [H(N),\Ga_N])
 \]
is the part that depends on $H_j$ and $A$ (again we omit argument $t$ at the solutions $\Ga_N$ and $\ga_j$).

Since we are interested in $\E \al_{N,j}(t)=\E |\al_{N,j}(t)|$,
the terms with $dB^k(t)$ are irrelevant, as long as
all coefficients $F_{jk}$ are uniformly bounded, and they clearly are.

\begin{remark}
As was shown in \cite{KolQuantMFG}, $D_j$ vanishes for the case $L_j=-L_j^*$.
\end{remark}

{\it Step 2.}

Looking at $D_j$ we first observe that the terms with $k\neq j$ vanish, because
\[
{\tr} \, (L_k^*L_k \Ga_N \ga_j)={\tr} \, (\Ga_N L_k^*L_k \ga_j)
={\tr} \, (L_k \Ga_N L_k^*\ga_j)={\tr} \, (\Ga_N \ga_jL_k^*L_k).
\]
Thus we are left with
\[
D_j={\tr} \bigl[  \frac12 L^*_j L_j\Ga_N \ga_j +\frac12 \Ga_N L_j^*L_j\ga_j-L_j \Ga_N L_j^*\ga_j
+\frac12 \Ga_N L^*_j L_j \ga_j +\frac12 \Ga_N\ga_j L_j^*L_j-\Ga_NL_j \ga_j L_j^*
\]
\[
-\Ga_N L_j^*\ga_j L_j^*-\Ga_N L_j^*L_j \ga_j-L_j \Ga_N \ga_j L_j^*-L_j \Ga_N L_j\ga_j
\]
\[
+(\Ga_N L_j^*+L_j \Ga_N)\ga_j \, {\tr} (\ga_j(L_j^*+L_j))
+\Ga_N \, {\tr} (\Ga_N(L_j^*+L_j))(\ga_j L_j^*+L_j \ga_j)
\]
\[
-\Ga_N \ga_j\, {\tr} (\Ga_N(L_j^*+L_j))\, {\tr} (\ga_j(L_j^*+L_j))\bigr].
\]

Further cancelation yields the following expression:
\[
D_j=-{\tr} (\ga_j L_j \Ga_N L_j^*+\ga_j L_j^* \Ga_N L_j +\ga_j L_j^*\Ga_N L_j^*+\ga_j L_j \Ga_N L_j)
\]
\[
+ {\tr} (\ga_j\Ga_N L_j^*+\ga_j L_j \Ga_N) \, {\tr} (\ga_j(L_j^*+L_j))
+ {\tr} (\ga_j \Ga_N L_j+ \ga_j L_j^* \Ga_N) \, {\tr} (\Ga_N(L_j^*+L_j))
\]
\[
-{\tr} (\Ga_N \ga_j)\, {\tr} (\Ga_N(L_j^*+L_j))\, {\tr} (\ga_j(L_j^*+L_j)).
\]

By Lemma \ref{lemmaonspectraces} from Appendix A,
\[
|D_j| \le  28 \|L\|^2 \, {\tr}\, ((1-\ga_j) \Ga_N).
\]

{\it Step 3.}

Let us now deal with $C_j$.
It will be convenient to introduce the orthogonal projectors  $q_j=q_j(t)=\1-\ga_j(t)$ in $L^2(X)$,
which are also identified
with the orthogonal projectors in $L^2(X^N)$ by making them act on the $j$th variable, and
the averaging operator
\begin{equation}
\label{eqdefaverkn}
\hat m_N =\hat m_N(t)=\frac{1}{N} \sum_{j=1}^N q_j(t),
\end{equation}
on $L^2(X^N)$. In terms of these operators one can  rewrite \eqref{eqforalpha}
in the following equivalent form
\begin{equation}
\label{eqforalphacon}
\al_{N,j}(t)=1-{\tr}(\ga_j(t) \Ga_N(t))={\tr}(q_j(t) \Ga_N(t))=(\Psi_N(t), q_j(t) \Psi_N(t))=(\Psi_N, q_j \Psi_N),
\end{equation}
so that (by the i.i.d. property of $q_j$)
\begin{equation}
\label{eqforalphaconex}
\E \al_{N,j}(t)=\E \, {\tr}(q_j \Ga_N)=\E \, {\tr}(m_N \Ga_N).
\end{equation}

Then we can rewrite $C_j$ as follows:

\[
C_j =i \, {\tr} ([H_j+A_j(\bar \eta),\ga_j]\Ga_N)
 +i \, {\tr} ([\ga_j,H(N)] \Ga_N)
\]
\[
=-i \, {\tr} ([H_j+A_j(\bar \eta),q_j]\Ga_N(t))
 +i \, {\tr} ([H(N),q_j] \Ga_N)
 \]
 \[
=i \, {\tr} ([H(N)-H_j-A_j^{\bar \eta},q_j] \Ga_N) \\
=i \, {\tr} ([\frac{1}{N}\sum_{m\neq j} A_{mj}-A_j(\bar \eta),q_j] \Ga_N).
\]

Note that all $H_k$ vanish from the final expression, because $[H_k, q_j]=0$ for all $k\neq j$
(since $H_k$ and $q_j$ act on different variables).

Consequently,
\[
|C_j|
\le 2 \left|{\tr} \left(\left(\frac{1}{N}\sum_{m\neq j} A_{mj}-A_j(\bar \eta)\right)q_j \Ga_N\right) \right|
\]
\begin{equation}
\label{eqforalphaconder1}
\le \frac{2}{N} \, \left|{\tr} \left(\left(\sum_{m\neq j} A_{mj}-(N-1)A_j(\bar \eta)\right) q_j \Ga_N\right)\right|
+\frac{2}{N} \, |{\tr} ( A_j(\bar \eta) q_j \Ga_N)|.
\end{equation}

Let us introduce the random functions $\de^j_N$.
In case (1) of the Theorem, they are defined as
\[
\de^j_N(z,w;t)=\de^j_N(z,w) =  \frac{1}{N-1}\sum_{m\neq j} \bar \psi_m(z) \psi_m(w)-\overline{\eta(z,w)}
=\frac{1}{N-1}\sum_{m\neq j}  \overline{\ga_m(z,w)}-\overline{\eta(z,w)}.
\]
In case (2) of the theorem, they are defined as
\[
\de^j_N(z;t) =\de^j_N(z)=  \frac{1}{N-1}\sum_{m\neq j} (|\psi_j(z)|^2-\xi(z)).
\]
In both cases, by the law of large numbers $\de^j_N$ tend to $0$,
as $N\to \infty$, and $\E \de^j_N=0$ for any $j$.

Then we can write
 \[
 A_j(\bar \eta)=\frac{1}{N-1}\sum_{m\neq j}A_j(\overline{\ga_m})-A_j(\de^j_N)
 \]
 and therefore
 \begin{equation}
\label{eqforalphaconder3}
\E |C_j| \le 2\E (I+II+III),
\end{equation}
with
\[
I= \frac{1}{N}  \left|{\tr} \left(\sum_{m\neq j} (A_{mj}-A_j(\overline{\ga_m}))q_j \Ga_N\right)\right|,
\]
\[
II=|{\tr} \left(A^{\de_N} q_j \Ga_N\right)|, \quad
III=\frac{1}{N} |{\tr} ( A_j(\bar \eta) q_j \Ga_N)|.
\]

{\it Step 4.}

We have
\[
III \le \frac{1}{N} {\tr} |q_j\Ga_N| \, \| A_j^{\bar \eta}\|
\le \frac{1}{N} \| A_j(\bar \eta)\|.
\]

By \eqref{eqinterterm}, \eqref{eqinterterm1} and because
\[
\|\eta\|_{L^2(X^2)}= {\tr} (\eta^2) ={\tr }\, \eta = 1,
\]
it follows that
\[
III \le \frac{1}{N} \|A\|_{HS}
\]
in case (1) of the theorem and
\[
III \le \frac{1}{N} \|A\|
\]
in case (2) of the theorem.

{\it Step 5.}

Estimating II, we analyse separately cases (1) and (2) of the theorem.
In case (1) we first note
  \begin{equation}
\label{estimatedelta}
\E |\de^j_N(z,w)|^2=\text{Var} \, (\de^j_N(z,w))
 = \frac{1}{N-1} \text{Var} (\ga_j(z,w))
 \le  \E |\ga_j(z,w)|^2.
\end{equation}

Therefore,
\[
II \le \|A\|_{HS} \|\de^j_{N,t}\|_{L^2(X^2)} \le \frac{1}{N-1} \|A\|_{HS} \text{Var} (\ga_j(z,w)).
\]
Consequently,
\[
\E \, II \le \frac{1}{\sqrt{N-1}}\|A\|_{HS} \left(\E \int_{X^2} |\ga_{j,t}(z,w)|^2 dzdw  \right)^{1/2}
\le \frac{1}{\sqrt{N-1}}\|A\|_{HS}.
 \]

In case (2) of the theorem, we use the well known estimate for the moments of the sum of i.i.d. random
variables  (see e.g. \cite{Bahr}) to estimate
\[
\E |\de^j_N(z)|^q \le \frac{2}{(N-1)^q} \sum_{m\neq j} \E \bigl||\psi_j(z)|^2-\xi(z) \bigr|^q.
\]
Therefore,
\[
\E |\de^j_N(z)|^q \le \frac{2^{q+1}}{(N-1)^{q-1}}\E |\psi_1(z)|^{2q}.
\]
Consequently,

\[
\E \, II  \le \E \|A(\bar \de^j_N)\|\le \|V\|_p \E \left(\int |\de^j_N(z)|^q dz\right)^{1/q}
\]
\begin{equation}
\label{eqforalphaconder4}
\le  \|A\|_p \left(\E \int |\de^j_N(z)|^q dz\right)^{1/q}
 \le 2^{(q+1)/q} (N-1)^{-(q-1)/q} \|A\|_p M_{2q}^2
  \le 4 N^{-(q-1)/q} \|V\|_p M_{2q}^2.
\end{equation}

{\it Step 6.}

Dealing with $I$ we plan to use the cancellation formula  \eqref{eqmaincancel}.
To this end, we write
\[
I=\frac{1}{N}  |(\Psi_N, \sum_{m\neq j} (A_{mj}-A_j^{\overline{\ga_m}}) q_j \Psi_N)|
\]
\[
\le \frac{1}{N}  \sum_{m\neq j}  |(\Psi_N, (q_m+\ga_m)(A_{mj}-A_j^{\overline{\ga_m}}) (q_m+\ga_m)q_j \Psi_N)|.
\]
By \eqref{eqmaincancel}, the term containing two multipliers $\ga_m$ vanishes, so that
$I\le I_1+I_2$ with
 \[
I_1= \frac{1}{N}  \sum_{m\neq j}  |(\Psi_N, q_m (A_{mj}-A_j^{\overline{\ga_m}}) q_m q_j \Psi_N)|,
 \]
 \[
I_2 =\frac{1}{N}  \sum_{m\neq j}  |(\Psi_N, \ga_m(A_{mj}-A_j^{\overline{\ga_m}}) q_m q_j \Psi_N)|.
 \]
 For the first term we get the estimate
 \[
I_1 \le \frac{1}{N} \sum_{m\neq j}
\| q_m \Psi_N\| \, \| q_m \Psi_N\| \, \|A_{mj}-A_j^{\overline{\ga_m}}\|
\]
\[
\le \frac{2}{N} \sum_{m\neq j}  (\Psi_N, q_m \Psi_N) \|A\|
\le \frac{2}{N}  \|A\| \sum_{m\neq j}  \al_{N,m}(t).
 \]
 Consequently
 \[
  \E I_1 \le 2 \|A\| \E \, \al_N(t).
  \]
Turning to $I_2$ we write
\[
I_2\le\frac{1}{N}  \sum_{m\neq j}
|(\Psi_N, q_j\ga_m(A_{mj}-A_j^{\overline{\ga_m}}) q_m q_j \Psi_N)|
\]
\[
+\frac{1}{N}  \sum_{m\neq j}
|(\Psi_N, \ga_j\ga_m(A_{mj}-A_j^{\overline{\ga_m}}) q_m q_j \Psi_N)|.
\]
The first term is estimated as $I_1$ above and in the second term the operator $A_j^{\overline{\ga_m}}$ cancels,
since it commutes with $q_m$. Thus we obtain
 \begin{equation}
\label{eqforalphaconder49}
\E I_2 \le  2 \|A\| \E \al_N(t)
+\frac{1}{N}  \E \sum_{m\neq j}  |(\Psi_N, \ga_j\ga_m A_{mj} q_m q_j\Psi_N)|.
\end{equation}

The second term here is estimated in Proposition \ref{GenKnowl2} of Appendix B.
 This gives the following estimate:

\[
\E I_2 \le \|A\| (3\E \al_N(t)+\frac{1}{N})).
\]

{\it Step 7.}

Putting all estimates above together we get (for $N>1$) that
 \begin{equation}
\label{eqforalphaconder5}
\E (|D_j|+|C_j|) \le 28 \|L\|^2+6 \|A\| \E \, \al_N(t)
+\|A\|_{HS} \left(\frac{2}{\sqrt{N-1}}+\frac{4}{N}\right)
\end{equation}
in case (1) and
\[
\E (|D_j|+|C_j|) \le 28 \|L\|^2+6 \|A\| \E \, \al_N(t)
+\frac{2}{N}\|A\|+ 4 M_T^2 \|A\|_p N^{-(q-1)/q}
\]
in case (2).

Applying Gronwall's lemma yields \eqref{eq2thmynonlinSch} and \eqref{eq3thmynonlinSch}.

\end{proof}

\subsection{Extension with control}
\label{extcont}

Assume now that
the individual Hamiltonian $H$ has a control component, that is, it can be written as $H+u\hat H$ with two self-adjoint
operators $H$ and $\hat H$ and $u$ a real control parameter taken from a bounded interval $[-U,U]$.
 Suppose that, for the idealized limiting evolution, $u$ is chosen as a
 certain function of an observed density matrix $\ga_j(t)$: $u=u(t,\ga_j(t))$.
Then in the original $N$ particle evolution $u$ will be chosen based on the approximation $\Ga^{(j)}_N(t)$ to $\ga_j(t)$,
that is as $u=u(t,\Ga^{(j)}_{N,t})$. Thus the controlled and observed $N$-particle evolution will be given by equation
 \eqref{eqmainNpartBelnonls} with the nonlinear controlled Hamiltonian $H_u(N)$ instead of $H$:

\begin{equation}
\label{eqHambinaryinter1co}
 H_u(N)f(x_1, \cdots , x_N)=\sum_{j=1}^N (H_j+u(t,\Ga^{(j)}_N) \hat H_j)f(x_1, \cdots , x_N)
 + \frac{1}{N}\sum_{i<j\le N} A_{ij}f(x_1, \cdots , x_N),
\end{equation}
where $u(t,\ga)$ is some continuous function.

The corresponding density matrix $\Ga_N(t)=\Psi_N(t)\otimes \overline{\Psi_N(t)}$ satisfies the
equation (again omitting argument $t$)
\[
d \Ga_N
=-i [H_u(N),\Ga_N] dt
+\sum_j (L_j\Ga_N L_j^* -\frac12 L^*_j L_j \Ga_N -\frac12 \Ga_N L^*_j L_j)\, dt
\]
\begin{equation}
\label{eqmainNpartBeldensnonl}
+\sum_j (\Ga_N L_j^*+L_j \Ga_N-\Ga_N \, {\tr} (\Ga_N(L_j^*+L_j))) dB^j(t).
\end{equation}

The limiting evolution \eqref{eqmainnonlin} generalizes now to the equation
\[
d\psi_j(x) =-i[H +u(t,\ga_j)\hat H +A(\bar \eta)
-\langle L_S \rangle_{\psi} L_A]  \psi_j(x) \, dt
\]
\begin{equation}
\label{eqmainnonlinco}
-\frac12 (L-\langle  L_S\rangle_{\psi_j})^* (L-\langle L_S \rangle_{\psi_j}) \psi_j(x)\,dt
+(L-\langle L_S \rangle_{\psi_j(t)})\psi_j(t) \, dB^j(t),
\end{equation}
with $\eta=\E \psi \otimes \psi$.
The equation for the corresponding density matrix $\ga_j=\ga_j(t)=\psi_j(t) \otimes \bar \psi_j(t)$ writes down as
\[
d\ga_j=-i[H+u(t,\ga_j)\hat H+A(\bar \eta), \ga_j] \, dt
+(L\ga_j L^* -\frac12 L^*L \ga_j -\frac12 \ga_j L^*L)\, dt
\]
\begin{equation}
\label{eqmainnonlinpartBel1dens}
+(\ga_j L^*+L \ga_j-\ga_j \, {\tr} (\ga_j(L+L^*))) dB^j(t),
\quad \eta(y,z)=\eta(y,z;t)=\E \ga_j(y,z;t).
\end{equation}

\begin{theorem}
\label{thmynonlinSchcont}
Under assumptions of Theorem \ref{thmynonlinSch} assume additionally that
the function $u(t,\ga)$ with values in a bounded interval $[-U,U]$
is Lipschitz in the sense that
\begin{equation}
\label{eq1thmynonlinSchco}
|u(t,\ga)-u(t,\tilde \ga)|\le \ka \, {\tr} |\ga -\tilde \ga|.
\end{equation}
Then the conclusions of Theorem \ref{thmynonlinSch} remain true
if in equations \eqref{eq2thmynonlinSch} and \eqref{eq3thmynonlinSch}
 one substitutes $\|A\|$ by $\|A\|+\kappa \|\hat H\|$.
\end{theorem}

 This is an easy extension of Theorem \ref{thmynonlinSch} that uses
\eqref{ineqKnPistoch} to estimate additional terms with control $u$. We refer to
\cite{KolQuantMFG} for details of the proof.

\subsection{Derivation of mean-field limit: counting measurement}
\label{seccounting}

In case of counting observation we reduce our attention to the case of unitary coupling operators $L_j$
(the difficulties with the general case are discussed in \cite{KolQuantLLN}). In that case,
the analog of equation \eqref{eqmainNpartBeldensnonl} describing the observation of a collection of
identical quantum particles, arising from the general quantum filtering equation
\eqref{eqBeleqcountm1} or observation of counting type is the equation
\begin{equation}
\label{eqmainNpartBeldensnonlcount}
d \Ga_N(t)
=\left(-i [H_u(N),\Ga_N(t)]
+\sum_k (L_k\Ga_N(t) L_k^* -\Ga_N(t))\right)\, (dM_k(t)+dt),
\end{equation}
where $M^k_t=N^k_t-t$ are martingales and $N^k(t)$ are standard independent Poisson processes.
The quantum filtering equation in terms of the pure states takes the form
\begin{equation}
\label{eqmainNpartBelnonlcount}
d \Psi_N(t)=-iH(N) \Psi_N(t) \, dt
+\sum_j (L_j-1)\Psi_N(t))  (dM^j(t)+dt).
\end{equation}

 The corresponding analog of the limiting equation \eqref{eqmainnonlinpartBel1dens} is the equation
\begin{equation}
\label{eqmainnonlinpartBel1denscount}
d \ga_j
=-i[H+u(t,\ga_j)\hat H+A(\bar \eta), \ga_j] \, dt
+(L_j\ga_j L_j^* -\ga_j)\, (dM^j(t)+dt),
\end{equation}

Let us look again at the dynamics of $\al_{N,j}$:
\[
d \al_{N,j}=-{\tr} (d\Ga_N(t) \ga_j(t))-{\tr} (\Ga_N(t) \, d\ga_j(t))-{\tr} (d\Ga_N(t) \, d\ga_j(t)).
\]

The part containing $H,A$ is the same as for diffusion. So we are interested only in the part containing $L_j$.

Recall the Ito multiplication rule for counting processes $dN_t^j dN_t^j=dN_t^j$, implying that
$dM_t^j dM_t^j=dM_t^j+dt$.

The part of the stochastic differential (at $dM^j_t$) in the expression for $d \al_{N,j}$ is of no interest
for us, as we are looking for the expectation of $d \al_{N,j}$, which is not affected by these martingale terms.
And direct inspection shows that the part at $dt$ depending on $L_j$ vanishes.
Turning to the expectations of  $\al_{N,j}$ we have the same situation as in Theorem
\ref{thmynonlinSchcont}, but in the simpler version of the absence of $L$ in all estimates.
Consequently the following result holds.

\begin{theorem}
\label{thmynonlinSchcount}
Under the assumptions on $H,\hat H, u, A$ from Theorem
\ref{thmynonlinSchcont} let  $L$ be unitary.

 Let $\ga_j(t)$ be solutions to equations
\eqref{eqmainnonlinpartBel1denscount} with i.i.d. initial conditions
$\ga_{j,0}=\psi_{j,0}\otimes \bar \psi_{j,0}$, $\|\psi_{j,0}\|=1$.
Let $\Psi_{N,t}$ be a solution to the $N$-particle equation \eqref{eqmainNpartBelnonlcount}
with $H_u(N)$ of type \eqref{eqHambinaryinter1co}, with some symmetric initial condition
$\Psi_{N,0}$, $\|\Psi_{N,0}\|_2=1$.
Then, in case (1),
\[
\E \al_N(t) \le \exp\{6(\|A\|+\ka \|\hat H\|)t\} \E \al_{N}(0)
\]
\[
+7 \left(\exp\{6(\|A\|+\ka \|\hat H\|)t\}-1\right)\|A\|_{HS} N^{-1/2},
\]
and in case (2)
\[
\E \al_N(t) \le \exp\{6(\|A\|+\ka \|\hat H\|)t\} \E \al_{N}(0)
\]
\[
+\left(\exp\{6(\|A\|+\ka \|\hat H\|)t\}-1\right)
\left(\frac{2}{N}\|A\|+ 4 M_T^2 \|A\|_p N^{-(q-1)/q}\right).
\]
\end{theorem}

\subsection{Appendix A: a technical estimate}

\begin{lemma}
 \label{lemmaonspectraces}
 Let $\ga$ be a one-dim projector in a Hilbert space, $\Ga$ a density matrix (positive operator with unit trace)
 and $L$ a bounded operator in this Hilbert space.  Then
 \begin{equation}
\label{eq1lemmaonspectraces}
|-4\,{\tr}\, (L\ga L\Ga )+2 \,{\tr}\,(\Ga(L \ga+\ga L))\,{\tr}\,(\Ga L+\ga L)
-4 \,{\tr}\,(\Ga \ga)\, {\tr} (\Ga L)\, {\tr} (\ga L)|
\le 20 \|L\|^2 {\tr} ((1-\ga) \Ga)
\end{equation}
for a self-adjoint $L$, and
 \[
|-{\tr} (\ga L \Ga L^*+\ga L^* \Ga L +\ga L^*\Ga L^*+\ga L \Ga L)
\]
\[
+ {\tr} (\ga \Ga L^*+\ga L \Ga) \, {\tr} (\ga(L^*+L))
+ {\tr} (\ga \Ga L+ \ga L^* \Ga) \, {\tr} (\Ga_N(L^*+L))
\]
 \begin{equation}
\label{eq2lemmaonspectraces}
-{\tr} (\Ga \ga)\, {\tr} (\Ga(L^*+L))\, {\tr} (\ga(L^*+L))|
\le 28 \|L\|^2 \, {\tr}\, ((1-\ga) \Ga)
\end{equation}
 for a general $L$.
 \end{lemma}

 \begin{proof}
 By the approximation argument it is sufficient to prove the Lemma for a finite-dimensional Hilbert space $\C^n$.
 Let $\al= {\tr} ((1-\ga) \Ga)$.
 Let us choose an orthonormal basis, where $\ga$ is the projection on the first basis vector.

 By positivity of $\Ga$ it follows that
\begin{equation}
\label{eq2alemmaonspectraces}
|\Ga_{jk}|\le \al \, \text{for} \, j,k\neq 1, \,  \text{and}
\, \max(|\Ga_{j1}|,|\Ga_{1j}|)\le \sqrt \al \, \text{for} \, j\neq 1.
\end{equation}

Let $L$ be a self-adjoint matrix.
 Then the expression under the module sign on the l.h.s. of \eqref{eq1lemmaonspectraces} writes down as
 \[
  -4 (L\Ga L)_{11} +2 [(L \Ga)_{11}+(\Ga L)_{11}]({\tr}\,(\Ga L)+L_{11})-4 \Ga_{11} L_{11}\, {\tr}\,(\Ga L)
\]
\[
=-4\sum_{j,k} L_{1j}\Ga_{jk}L_{k1}+2[2L_{11}\Ga_{11}+\sum_{j\neq 1} (\Ga_{1j}L_{j1}+L_{1j}\Ga_{j1})]({\tr}\,(\Ga L)+L_{11})
-4 \Ga_{11} L_{11}\, {\tr}\,(\Ga L)
\]
\[
=-4 L_{11} \sum_{j\neq 1} (L_{1j}\Ga_{j1}+L_{j1}\Ga_{1j})-4\sum_{j\neq 1,k\neq 1} L_{1j}\Ga_{jk}L_{k1}
+2\sum_{j\neq 1} (\Ga_{1j}L_{j1}+L_{1j}\Ga_{j1})({\tr}\,(\Ga L)+L_{11})
\]
\[
=2\sum_{j\neq 1} (\Ga_{1j}L_{j1}+L_{1j}\Ga_{j1})L_{11}(\Ga_{11}-1)
-4\sum_{j\neq 1,k\neq 1} L_{1j}\Ga_{jk}L_{k1}
+2\sum_{j\neq 1} (\Ga_{1j}L_{j1}+L_{1j}\Ga_{j1})({\tr}\,(\Ga L)-L_{11}\Ga_{11})
\]
\[
=-2\sum_{j\neq 1} (\Ga_{1j}L_{j1}+L_{1j}\Ga_{j1})L_{11}\al
-4\sum_{j\neq 1,k\neq 1} L_{1j}\Ga_{jk}L_{k1}
\]
\[
+2\left(\sum_{j\neq 1} (\Ga_{1j}L_{j1}+L_{1j}\Ga_{j1})\right)^2
+2\sum_{j\neq 1} (\Ga_{1j}L_{j1}+L_{1j}\Ga_{j1})
\sum_{j\neq 1,k\neq 1} L_{kj}\Ga_{jk}.
\]
Here all terms are of order $\al$, because of \eqref{eq2alemmaonspectraces}.

More precisely,
\[
|\sum_{j\neq 1,k\neq 1} L_{kj}\Ga_{jk}|=|{\tr}[(1-\ga) L(1-\ga)\Ga(1-\ga)]|
\le \|L\| {\tr}[(1-\ga)\Ga]\le \|L\| \al.
\]
Moreover,
\[
\sum_{j\neq 1}|\Ga_{j1}|^2=\sum_{j\neq 1}|\Ga_{1j}|^2
\le \Ga_{11} \sum_{j\neq 1} \Ga_{jj}\le \Ga_{11}\al \le \al.
\]
Hence,
\[
|\sum_{j\neq 1} (\Ga_{1j}L_{j1})|^2
\le \sum_{j\neq 1} |\Ga_{1j}|^2 \sum_{j\neq 1} |L_{j1}|^2\le \|L\|^2 \al,
\]
\[
|\sum_{j\neq 1} (\Ga_{j1}L_{1j})|^2\le \|L^T\|^2 \al= \|L\|^2 \al,
\]
and thus
\begin{equation}
\label{eq3lemmaonspectraces}
|\sum_{j\neq 1} (\Ga_{1j}L_{j1}\pm L_{1j}\Ga_{j1})|\le 2\|L\| \sqrt \al.
\end{equation}
Finally,
\[
|\sum_{j\neq 1,k\neq 1} L_{1j}\Ga_{jk}L_{k1}|^2 \le \sum_{j,k} |L_{1j}|^2 |L_{k1}|^2
\sum_{j\neq 1,k\neq 1} |\Ga_{jk}|^2 \le \|L\|^4 (\sum_{j\neq 1} |\Ga_{jj}|)^2
\le \|L\|^4 \al^2,
\]
where the estimate
\[
|\Ga_{jk}|^2 \le \Ga_{jj}\Ga_{kk}
\]
for all $j,k$ was used (arising from the positivity of $\Ga$).

Putting the estimates together we get \eqref{eq1lemmaonspectraces}.

For a general $L$ we can write $L=L^s+L^a$, where $L^s=(L+L^*)/2$ is self-adjoint and
$L^a=(L-L^*)/2$ is anti-Hermitian. Plugging this into the l.h.s.
of \eqref{eq2lemmaonspectraces} leads to several cancelations, so that
 the expression under the module sign in the l.h.s. becomes equal to
\[
-4\,{\tr}\, (L^s\ga L^s\Ga )+2 \,{\tr}\,(\Ga(L^s \ga+\ga L^s))\,{\tr}\,(\Ga L^s+\ga L^s)
-4 \,{\tr}\,(\Ga \ga)\, {\tr} (\Ga L^s)\, {\tr} (\ga L^s)
\]
\[
+2{\tr} (\ga [\Ga, L^a])({\tr} (\Ga L_s)-{\tr} (\ga L^s)).
\]

Everything apart from the last term is already estimated  by \eqref{eq1lemmaonspectraces}.

By \eqref{eq3lemmaonspectraces} (that is valid for arbitrary $L$),
\[
|{\tr} (\ga [\Ga, L^a])|=|[\Ga,L^a]_{11}|=|\sum_{j\neq 1} (\Ga_{1j}L^a_{j1}-L^a_{1j}\Ga_{j1})|
\le 2\|L\| \sqrt \al,
\]
and
\[
|{\tr} (\Ga L_s)-{\tr} (\ga L^s)|
\]
\[
=|(\Ga_{11}-1)L^s_{11}+\sum_{j\neq 1} (\Ga_{1j}L^s_{j1}-L^s_{1j}\Ga_{j1})
+\sum_{j,k\neq 1} (\Ga_{jk}L^s_{kj})|\le 4 \|L\| \sqrt \al,
\]
which implies \eqref{eq2lemmaonspectraces}.
 \end{proof}

\begin{remark} Our proof of Lemma \ref{lemmaonspectraces} is based on some remarkable
cancellation of terms in concrete calculations
via coordinate representations. The author does not see any intuitive reasons for its validity.
Neither is it clear whether it can be extended to arbitrary density matrixes $\ga$,
not just one-dimensional projectors.
\end{remark}

\subsection{Appendix B: stochastic version of Knowles-Pickl estimates}
\label{secKP}

For completeness, following mostly \cite{KolQuantMFG} and \cite{KnowlesPickl}
(with more details for clarity),
we prove here a stochastic version of an estimate from \cite{KnowlesPickl}.

\begin{prop}
\label{GenKnowl2}
Let $A$ be a bounded self-adjoint operator in $\HC^{\otimes 2}=L^2(X^2)$, and let
$A_{jk}$ denote the operator on $\HC^{\otimes N}=L^2(X^N)$ that acts as $A$ on
 the $j$th and $k$th coordinates of an $f(x_1, \cdots, x_N)\in L^2(X^N)$. Let $\ga$ be
 a random one-dimensional projector in $\HC$ and let $\ga_j$, $j=1,2,\cdots$,  be independent copies of $\ga$,
 defined on some probability space $(\Om, \FC, \P)$,
 considered as operators in $L^2(X^N)$ that act on the $j$th variable of an $f(x_1, \cdots, x_N)\in L^2(X^N)$.
Moreover, let $\Psi_N=\Psi_N(x_1, \cdots, x_N)$ be a sequence of random unit vectors on $(\Om, \FC, \P)$
with values in $\HC^{\otimes N}$, which are symmetric in the sense that their distributions is invariant
under any permutation of its coordinates. Set $q=\1-\ga$. Then
 \begin{equation}
\label{eq2GenKnowl2}
\E  |(\Psi_N, \ga_1\ga_2A_{12} q_1 q_2\Psi_N)|
\le \|A\| \sqrt{N/(N-1)} (\E (\Psi_N, q_1 \Psi_N) +1/N),
\end{equation}
 \begin{equation}
\label{eq2GenKnowl3}
\frac{1}{N}  \E \sum_{m\neq j}  |(\Psi_N, \ga_j\ga_m A_{mj} q_m q_j\Psi_N)|
\le \|A\| (\E (\Psi_N, q_1 \Psi_N) +1/N).
\end{equation}
\end{prop}

\begin{proof}
Inequality \eqref{eq2GenKnowl3} is a direct consequence of \eqref{eq2GenKnowl2} and symmetry.
Let us prove \eqref{eq2GenKnowl2}.
The key objects in the method are the average operator
\[
\hat m =\hat m_N= \frac{1}{N} \sum_{j=1}^N q_j
\]
and the projectors
\[
P_k=\sum_{I\subset \{1, \cdots , N\}: |I|=k} \prod_{i\in I} q_i \prod_{i\notin I} \ga_i,
\]
for $k\in \{0, \cdots, N\}$,
and $P_k=0$ for other integer $k$.

Clearly $P_k$ are orthogonal projectors such that
\[
P_kP_l=\de_k^lP_k, \quad \sum_k P_k=\1,
\]
and
\[
\hat m = \frac{1}{N} \sum_{j=1}^N q_j=\frac{1}{N}  \sum_{j=1}^N \sum_{k=0}^N q_jP_k
=\frac{1}{N}  \sum_{k=0}^N kP_k,
\]
because the composition of $q_j$ with each term of the sum entering the expression for $P_k$
either coincides with this term or equals zero.

For functions $f:\Z \to \C$ with support contained
in $\{0, \cdots, N\}$ let us define
\begin{equation}
\label{eqdefknowleshomom}
\hat f=\sum_k f(k) P_k.
\end{equation}
This notation is consistent with the notation for $\hat m$ that arises from the function $m(k)=k/N$.

The mapping $f\to \hat f$ is an algebraic homomorphism in the sense that $\widehat{fg}=\hat f \hat g$,
and all $\hat f$ commute with all $\ga_j, q_j$.

Of importance are the powers of $\hat m$:
\begin{equation}
\label{eqaverageqpow}
\hat m^j=\sum_{k=0}^N \left(\frac{k}{N}\right)^j P_k
\end{equation}
for any $j\ge 0$. It follows from the properties of the projectors that
\[
\hat m^j \hat m^l=\hat m^{j+l}
\]
for all $l,j\ge 0$.
Since $\hat m$ is not invertible, the inverse power is not defined.
However, since
\[
\hat m P_k=\frac{k}{N}P_k,
\]
for any $k$, $\hat m^{-1}$ can be defined on the image of $\1-P_0$.
Let us thus denote (with some abuse of notation)
\begin{equation}
\label{eqaverageqpowneg}
\hat m^{-j}=\sum_{k=1}^N \left(\frac{k}{N}\right)^{-j} P_k,
\end{equation}
for any $j>0$.
It is seen directly that
\[
\hat m^j \hat m^{-j}= \hat m^{-j}  \hat m^j=\1-P_0
\]

 and thus
\begin{equation}
\label{eqaverageqpowneg1}
\hat m^j \hat m^{-j} q_i=q_i
\end{equation}
for all $i\in \{1, \cdots, N\}$ and $j \neq 0$, because $P_0 q_i=0$.

Let us introduce the operators
\[
P_k^r=\sum_{I\subset \{r+1, \cdots , N\}: |I|=k} \prod_{i\in I} q_i \prod_{i\notin I} \ga_i
\]
for $k\in \{0, \cdots, N-r\}$,
setting $P_k^r=0$ if $k\notin \{0, \cdots, N-r\}$.
These operators appear by composing $\ga_i$ and $q_i$ with $P_k$. For instance,
\[
\ga_1 \cdots \ga_r P_k =P_k^r=P_k \ga_1 \cdots \ga_r.
\]

More generally, let $Q_r=\om_1 \cdots \om_r$, where each $\om_l$ is either $q_l$ or $\ga_l$.
It is seen directly that
\begin{equation}
\label{eqaverageqpowneg2}
Q_r P_k=P_kQ_r=Q_r P^r_{k-n},
\end{equation}
where $n$ is the number of $q$s in the product $Q_r$.
In particular, both sides vanish if $k\notin \{n, \cdots ,N+n-r\}$.

The key property of $Q_r$ is as follows.

\begin{lemma}
\label{knowlpilemma}
Let two products $Q_r^1$ and $Q_r^2$ be given with the numbers $n_1$
and $n_2$ of $q$s respectively, and $A_r$ an operator acting on the first $r$ variables in $L^2(X^N)$.
Then
\[
Q_r^1 A_r \hat f Q_r^2 =Q_r^1 \widehat{\tau_nf}  A_r Q_r^2,
\]
where $\tau_n f(k) =f(k+n)$ and $n=n_2-n_1$.
\end{lemma}

\begin{remark} This result is actually needed only for $r=n=2$.
\end{remark}

\begin{proof}
For any $Q_r$,
\[
Q_r \hat f=\hat f Q_r=\sum_k f(k) Q_r P_k=\sum_k f(k) Q_r P^r_{k-n}
=\sum_k f(k) P^r_{k-n} Q_r.
\]
Consequently,
\[
Q_r^1 A_r \hat f Q_r^2 =Q_r^1 A_r \sum_k f(k) P^r_{k-n_2} Q_r^2=Q_r^1 A_r \sum_k f(k+n_2) P^r_k Q_r^2.
\]

Since $A_r$ and $P_k^r$ commute
and using $Q_r^1 P_{k+n_1}=Q_r^1 P^r_k$ (which holds by \eqref{eqaverageqpowneg2}),
it follows that
\[
Q_r^1 A_r \hat f Q_r^2
 =Q_r^1 \sum_k f(k+n_2) P^r_k A_r Q_r^2
 \]
 \[
 =Q_r^1 \sum_k f(k+n_2) P_{k+n_1} A_r Q_r^2
=Q_r^1 \sum_k f(k+n_2-n_1) P_k A_r Q_r^2,
\]
as required.
\end{proof}

Using \eqref{eqaverageqpowneg1} and Lemma \ref{knowlpilemma}, we can now write
\[
(\Psi_N, \ga_1\ga_2A_{12} q_1 q_2\Psi_N)
=(\Psi_N, \ga_1\ga_2A_{12}  \hat m^{1/2} \hat m^{-1/2} q_1 q_2\Psi_N)
\]
\[
=(\Psi_N, \ga_1\ga_2   \widehat{\tau_2m^{1/2}} A_{12} \hat m^{-1/2} q_1 q_2\Psi_N)
\]
yielding the estimate
\begin{equation}
\label{eqPicklthirdterm}
|(\Psi_N, \ga_1\ga_2A_{12} q_1 q_2\Psi_N)|
\le \|A_{12}  \widehat{\tau_2m^{1/2}} p_1p_2 \Psi\| \, \| \hat m^{-1/2} q_1q_2\Psi_N\|,
\end{equation}
and therefore
\begin{equation}
\label{eqPicklthirdterm1}
\E |(\Psi_N, \ga_1\ga_2A_{12} q_1 q_2\Psi_N)|
\le \left(\E \|A_{12}  \widehat{\tau_2m^{1/2}} p_1p_2 \Psi\|^2\right)^{1/2}
 \, \left( \E\| \hat m^{-1/2} q_1q_2\Psi_N\|^2\right)^{1/2}.
\end{equation}

To estimate the second term in \eqref{eqPicklthirdterm1} we write
\[
\begin{aligned}
\E \| \hat m^{-1/2} q_1q_2\Psi_N\|^2 & =\E (\Psi_N, \hat m^{-1} q_1q_2\Psi_N)
 =\frac{1}{N(N-1)} \E (\Psi_N, \sum_{l\neq r} \hat m^{-1} q_lq_r\Psi_N) \\
& \le \frac{1}{N(N-1)} \E (\Psi_N, \sum_{l,r} \hat m^{-1} q_lq_r\Psi_N)
 =\frac{N}{N-1} \E (\Psi_N, \hat m^{-1} \hat m^2 \Psi_N) \\
& =\frac{N}{N-1} \E (\Psi_N, \hat m \Psi_N)
\le \frac{N}{N-1} \E  (\Psi_N, q_1 \Psi_N).
\end{aligned}
 \]
 To estimate the first term in \eqref{eqPicklthirdterm} we note that
 \[
\widehat{\tau_2m}=\sum_{k=0}^N\frac{k+2}{N} P_k
=\hat m +\frac{2}{N} \sum_{k=0}^NP_k=\hat m +\frac{2}{N},
\]
and consequently
\[
\begin{aligned}
& \E \|A_{12}  \widehat{\tau_2m^{1/2}} p_1p_2 \Psi\|^2
=\E (\Psi_N,  \widehat{\tau_2m^{1/2}} p_1 p_2 A_{12}^2  p_1 p_2 \widehat{\tau_2m^{1/2}} \Psi_N) \\
& \le \|A\|^2 \E \|\widehat{\tau_2m^{1/2}} \Psi_N\|^2
=\|A\|^2 (\Psi_N, \widehat{\tau_2m} \Psi_N) \\
&\le \|A\|^2 \left[\E (\Psi_N, \hat m \Psi_N) +\frac{2}{N}\right]
=\|A\|^2 \left[\E (\Psi_N, q_1 \Psi_N) +\frac{2}{N}\right].
\end{aligned}
\]

Putting these estimates together we get that
\[
\E |(\Psi_N, \ga_1\ga_2A_{12} q_1 q_2\Psi_N)|
\le \|A\| \sqrt{N/(N-1)} \sqrt{\E (\Psi_N, q_1 \Psi_N)(\E (\Psi_N, q_1 \Psi_N) +2/N)}
\]
\[
\le  \|A\|_{HS} \sqrt{N/(N-1)} (\E (\Psi_N, q_1 \Psi_N) +1/N),
\]
as required.
\end{proof}

\section{Fractional quantum mechanics of open systems}
\label{secfraceq}

As was shown above, the standard Belavkin equations of quantum filtering can be obtained as the scaled limits of
the sequences of discrete observations. The main assumption for these approximating processes was that the time
between successive measurement was either constant (discrete Markov chain approximation) or was exponentially distributed
(continuous time Markov chain approximation). Of course there is no a priori reasons for these assumptions.
 In fact, for several processes from natural sciences it turns out to be more appropriate to model times between successive
 events by random variables from the domains of attraction of a stable law, that is via continuous time random walk (CTRW).

Before formulating the main result of these alternative approximations to quantum continuous measurement, we
remind briefly the main points on generalised (position dependent) CTRW approximations.

Suppose $T_1^h,T_2^h, \cdots $ is a sequence of i.i.d.
random variables in $\R_+$ such that the distribution of each $T_i^h$
is given by a probability measure $\mu_{time}^h (dt)$
on $\R_+$, that depend on a positive (scaling) parameter $h$. Let
\begin{equation}
\label{eqinversewalk}
N_t^h=\max \{ n: \sum_{i=1}^n T_i^h \le t\}.
\end{equation}

Let $U_h$ be a transition operator of a discrete time Markov chain
 $O^h_n(x)$ in a metric space $M$
depending on a positive parameter $h$, so that
\begin{equation}
\label{transitoperMar}
U_hf(x)=\E O^h_1(x)=\int f(y) \mu^h(x, dy),
\end{equation}
with some family of stochastic kernels  $\mu^h(x, dy)$ such that $U_h$
is a bounded operator in the space $C(M)$.

Suppose the sequence $T_1^h,T_2^h, \cdots $ is independent of $O^h_n(x)$.
Then the time-changed Markov chain
\begin{equation}
\label{ctrwposdep}
 O^h_{N_t^h}(x)
\end{equation}
is a {\it generalized scaled (position dependent)
 continuous time random walk} (CTRW) arising from
$\mu^h(x, dy)$ and $\mu^h_{time}(dt)$.

Standard CTRWs corresponds to the situation, when transitions $U^h$ comes from adding to each position
some i.i.d. random variables. Saying in other words, they are random walks
with time between jumps distributed like $T_j^h$. On the other hand, for our purposes we need
the transitions of the type
\[
U_hf(x)=\E O^h_1(x)=\sum_{j=1}^J  f(Y_j^h(x)) p_j^h(x),
\]
with a family of continuous maps $Y_j^h:\R^d\to \R^d$
and a probability law $\{p_1^h(x), \cdots, p_J^h(x)\}$.

The CTRW were introduced in \cite{MW}. They found numerous applications in physics.
The scaling limits of these CTRW were analysed by many authors, see e.g. \cite{KKU}, \cite{MS}.
The scaling limit for the position dependent CTRW was developed in \cite{Kol7}. In particular,
let
\[
 \si_y=\max \{t: S_t \le y\}
 \]
 denotes the inverse process to the L\'evy subordinator $S(t)$ generated by the operator
 \[
 L_{\be}f(x)=\int_0^{\infty}(f(x+y)-f(x))y^{-1-\be} dy.
 \]
It was found that, if $T_i^h=h^{1/\be} T_i$,
with i.i.d. positive random variables $T_i$  from the domain
of attraction of a $\be$-stable law, $\be\in (0,1)$, that is,
 such that
\begin{equation}
\label{eq1propCTRW}
 \P (T_i>m)= \frac{1}{\be m^{\be}}
\end{equation}
for sufficiently large $m$,
and $(U_h-1)/h$ converge to the generator $A$ of a Feller semigroup $F_t$ of the corresponding process $O_t$,
then the time-changed Markov chain \eqref{ctrwposdep} converges to the time-changed process $O_{\si_t}(x)$.
Moreover, the expectations $f_t(x)= \E F_{\si_t}(x)=\E f(O_{\si_t}(x))$ of the limiting time-changed process
solve the fractional in time differential equation
\begin{equation}
\label{eqfraceq}
 D^{\be}_{0+\star}f_t(x)=Af_t(x), \quad f_0(x)=f(x),
\end{equation}
where $A$ acts on the variable $x$ and
$D^{\be}_{0+\star}$ is the Caputo-Djerbashian derivative of order $\be$ acting on the variable $t$:
\begin{equation}
\label{eqfraceq2}
 D^{(\be)}_{0+\star}f_t=\int_0^t (f_{t-s}-f_t)s^{-1-\be} ds+ \frac{f_0-f_t}{\be t^{\be}}.
\end{equation}
For a general background on fractional derivatives we can refer, e.g., 
to books \cite{BalScalabook} or \cite{Kolbook19}.

However, the rates of convergence
appeared first only in \cite{KolRates}.

By Theorems \ref{nonLinsdesemmix}, \ref{thBeleqdif},
and Theorem 4 of \cite{KolRates}, we derive the following.

 \begin{theorem}
 \label{mainth}
 Under the assumptions of Theorem \ref{thBeleqdif} let the Markov chain with transitions \eqref{eqMarkchain2}
be modified in such a way that the laws of transitions remain unchanged, but the time between transitions be
taken as scaled random variable from the domain
of attraction of a $\be$-stable law.

Then the corresponding generalized CTRW processes \eqref{ctrwposdep}
  built from the transition operator
\eqref{eqMarkchain2} converge to the process $O^{\rho}_{\si_t}$ obtained from the process
$O^{\rho}_t$  with the semigroup $\Phi^{mix}_t$  of Theorem \ref{nonLinsdesemmix} 
via subordination by the inverse stable process $\si_t$, so that
\[
\sup_\rho|\E f( O^h_{N_t^h}(\rho))-\E f (O_{\si_t}(\rho))|
=\|\E \, U_h^{N_t^h}f - \E , \Phi^{mix}_{\si_t}f\|
\]
\begin{equation}
\label{eq0CTRWrates}
\le C (1+t+t^{-1}) \ln (1/h)^{-1}  \|f\|_{C^4(S^+(\HC))},
\end{equation}
for all $t>0$ and a constant $C$.

Moreover, the functions $f_t(x)=\E (\Phi^{mix}_{\si_t}f)(x)$
satisfy the fractional Caputo-Djerbashian equation
\eqref{eqfraceq} with the generator $A=\AC_{mix}$ given by \eqref{eqmixgenerun}.
 \end{theorem}

Equations \eqref{eqfraceq} with the generator $A=\AC_{mix}$
represent the fractional analogs of the process of quantum stochastic filtering. They
describe the evolutions of fractional quantum mechanics of open systems.
They are different from the fractional
Schr\"odinger equations suggested in  \cite{Lask02} and extensively studied recently.

Equations \eqref{eqfraceq} with the generator $A=\AC_{mix}$
 describe the process of fractional continuous quantum filtering on the level
of the evolution of averages. On the 'micro-level' of SDEs these equations correspond
to stopping the solutions of these SDEs at a random time $\si_t$ given by the inverse of a L\'evy subordinator.
By the known link between time-changed integral and time-changed integrand (see e.g. \cite{Koba} and references therein),
the corresponding SDEs can be written in terms of time-changed noise. Say,
for the diffusive case, if $\rho(t)$ solves SDE \eqref{Lindstochnorm1}, then the time-changed
process $R(t)=\rho(\si(t))$, whose distributions are governed by equation \eqref{eqfraceq}
with $\AC$ from \eqref{eqdifgener}, solves the {\it fractional quantum filtering SDE}
\[
dR(t)= \left(-i[H,R(t)]\, dt+\LC_L R (t)\right)\, d\si_t
\]
\begin{equation}
\label{Lindstochnormfrac}
+[LR(t)+R(t) L^*-R(t)\, {\tr} \, (LR(t)+R(t) L^*) ] dB(\si_t).
\end{equation}

\section{Examples and properties of solutions}
\label{secexamprop}

This section is a review of basic examples and properties of quantum filtering SDEs and their solutions.
Namely, we discuss the models of the continuous observation of position or momentum of classical
quantum particles and finite-dimensional models. In particular, we present results on long time behavior
of solutions for some cases of scattering and ergodic behaviors, which can be interpreted as a watch-dog effect
and/or dynamic uncertainty principle. For classical Hamiltonian we also provide the representation of solutions
in terms of rigorously defined Feynmann path integral.
 Unlike previous sections, we present only results referring
for proofs to the original papers.

\subsection{Classical quantum particle under continuous observation of its position and/or momentum}

The fundamental example of quantum filtering equations of diffusive type represent equations describing
 classical quantum particles with vanishing or quadratic potential under continuous observation of its
  position and momentum (even more generally, coupling operators $L$ are linear combinations
  of position and momentum operators). These equations in linear form for pure states
 can be written as
\begin{equation}
\label{eqqufiBlinPosMom}
d\chi(t) =-iH\chi(t) dt -\frac12 (\al_1^2 x^2 +\al_2^2 \De)\chi(t)\,dt
+\al_1 x \chi(t) dY_1(t)-\al_2 i\frac{\pa}{\pa x} \chi(t) \, dY_2(t),
\end{equation}
where $\al_1,\al_2$ are real constants, $\chi(t) \in L^2(\R^d)$,
$Y_1(t),Y_2(t)$ are standard $d$-dimensional Brownian motions,
$\De$ is Laplacian acting on the position variable of $x$ and $H=-\frac12 h \De+\om x^2$
with $\om \ge 0$ is the standard Hamiltonian
of a quantum oscillator or a free particle.

The main point about these equations is the existence of an explicit Green function for the Cauchy problem that
has a Gaussian form, whose (though complex) coefficients can be written explicitly. This fact allows for rather
detailed analysis of various properties of these equations. Moreover, using perturbation technique,
one can use these explicit solutions to tackle the situations with $H$ perturbed by (regular in some sense) potentials $V$.
Though well-posedness of equation \eqref{eqqufiBlinPosMom} can be obtained via the general technique of
  \cite{Mora13} (see also \cite{Kol25b}), explicit calculations provide more deep insights.

As an example, let us consider the case with only position measurement.
For vanishing potential the corresponding filtering equation \eqref{eqqufiBlins} takes the form

\begin{equation}
\label{eqqufiBlinGa}
d\chi(t) =\frac12 (ih\De-\al^2 x^2) \chi(t) \,dt+\al x \chi(t) dY(t),
\end{equation}
where $Y(t)=(Y_1, \cdots, Y_d)(t)$ is a $d$-dimensional Brownian motion (BM),
$\al \in \R$, $h>0$  are constants.

The Green function (fundamental solution for the Cauchy problem) for this equation was first calculated
in \cite{Kol95}. It has the following form (that can be checked by inspection):

\begin{equation}
\label{eqqufiBlinGaSo}
u_G(t,x,y)=\exp\{-\frac{\om}{2} (x^2+y^2) +\be xy -a x -b y -\ga \}
\end{equation}
where
\[
\om=\si \coth (\si Gt), \quad
\be=\si (\sinh(\si G t))^{-1}, \quad C=\sqrt{\be/(2\pi)}.
\]
\[
a= \al (\sinh(\si G t))^{-1}\int_0^t  \sinh(\si G s) dB(s)
\]
\[
b=\si G \int_0^t \frac{a(s)}{\sinh(\si G s)} ds,
\quad \ga=\frac12 \int_0^t a^2(s)ds
\]
with
$\si=\sqrt{2\al^2/ih}=\sqrt{2\al^2/h} \exp\{-i\pi/4\}$.

It follows, that for small $t$,
\[
\om=\frac{1}{iht} +\frac23 \al^2 t +O(t^3), \quad
\be=\frac{1}{iht}-\frac13 \al^2 +O(t^3),
\]
\begin{equation}
\label{asympstochosc}
a \sim \frac{\al}{t} \xi(t), \quad \xi(t)=\int_0^t s \, dB(s), \quad
b \sim \al \int_0^t \frac{\xi(s)}{s^2} ds.
\end{equation}

Let us introduce
the Hilbert space $L^2_R=L^2_R(\R^d)$ of functions from $L^2(\R^d)$ with
a finite norm squared
\[
\|f\|^2_R =\|f\|^2+\sum_j\|x_jf\|^2+\sum_j \|\frac{\pa}{\pa x_j} f\|^2.
\]
We refer to \cite{Kol25b} for the proof of the following properties of the solutions.

\begin{prop}
\label{propsmoo}
The resolving operator $U_t$ for the Cauchy problem to equation \eqref{eqqufiBlinGa}
with the integral kernel \eqref{eqqufiBlinGaSo} takes $L^2(\R^d)$ to $L^2_R(\R^d)$
for any $t>0$. Moreover, for $t\in [0,T]$ with any $T$, and any $p\in [1,2)$
\begin{equation}
\label{eqpropsmoo}
\E \|U_t\|^p \le C, \quad  \E \|U_t\|^p_{L^2_R(\R^d)} \le C,
\end{equation}
with a constant $C$ depending on $T$ and $p$. The norm squared $\|U_t\|^2$ has no finite expectation.
\end{prop}

\begin{remark}
One can show the following estimate for the smoothing property of $U_t$:
\[
 \E\|U_t \|^2_{L^2_R(\R^d)\to L^2(\R^d)} \le \frac{C}{t^6}.
\]
But  this estimate is seemingly too rough to be useful.
\end{remark}

This Proposition implies the following result.
\begin{prop}
\label{propsmoo1}
There exists a unique strong solution
to \eqref{eqqufiBlinGa} for any initial condition,
the equation being satisfied generally for all $t>0$,
and, for the initial condition from $L^2_R(\R^d)$, for all $t\ge 0$.
\end{prop}

Via the standard perturbation argument this result can be extended to more general equations
\begin{equation}
\label{eqqufiBlinGaM}
d\chi(t) =(-iH -\frac12 \al^2 x^2) \chi(t) \,dt+\al x \chi(t) dY(t),
\end{equation}
with $H=-\frac12 h\De +V(x)$ with sufficiently regular $V$(and even with magnetic fields included).
 For instance, the following result is straightforward.

\begin{prop}
\label{propsmoo2}
If $V$ is bounded with bounded continuous first and second order derivatives,
there exists a unique strong solution
to \eqref{eqqufiBlinGaM} for any initial condition from $L^2(\R^d)$.
\end{prop}

Apart from having explicit Green function and having smoothing property
(first claim of Proposition \ref{propsmoo}), equation \eqref{eqqufiBlinGa}
has another nice property. It preserves Gaussian wave packets, which yields
lots of explicit solutions. Namely, denoting
\[
g_{q,p}^{\omega}(x)
 =c \exp \{ -\frac{\omega}{2}(x-q)^2+ipx \},
 \]
 a general Gaussian function, where $p,q$ (mean momentum and position)
 are real and $\om$ complex with positive real part, one deduces by inspection that
 if initial condition for the Cauchy problem of equation \eqref{eqqufiBlinGa} has this form,
 then the solution has this form as well, where $\om(t), p(t),q(t)$ evolve according the
following equations:
\[
\dot \om =-ih \om ^2+2\al^2,
\]
\[
\begin{aligned}
& dq= [hp  -\frac{2\alpha^2}{Re \, \omega} q]\, dt
+\frac{\alpha}{Re \, \omega} \, dY, \\
& dp= \frac{Im \, \om}{Re \, \omega} [2\al^2 q  \, dt - \al  \, dY(t)].
\end{aligned}
\]
These equations can be solved explicitly yielding, in particular, that
\[
\lim_{t \to \infty} \omega(t) =\sigma=\al \sqrt{2/h i} =\al \sqrt{2/h}e^{-\pi/4}=\frac{\al}{\sqrt h}(1-i)
\]
independently on the initial condition, which was observed already in \cite{Dios88} and \cite{Bel89}.

\subsection{Long-time behaviour for classical Hamiltonians}

The corresponding normalised version of  equation \eqref{eqqufiBlinGa} has the form
 \begin{equation}
\label{eqqufiBlinGanl}
d\phi=\frac12 \left(i h\Delta \phi
-\al^2(x-\langle x \rangle _{\phi})^2\phi\right)\, dt
+\al (x-\langle x \rangle _{\phi})\phi \, dB.
\end{equation}

As follows from the discussion of the linear equation above, Gaussian wave packets are
preserved under this evolution. Moreover, unlike the case of the unitary Schr\"odinger equation of a free
quantum particle, where all solutions are asymptotically free waves
$e^{ipx}$ (that is,
Gaussian packets with the infinite variance), the Gaussian solutions
of \eqref{eqqufiBlinGa} and \eqref{eqqufiBlinGanl} tend to a Gaussian function
with a fixed finite non-vanishing $\omega$. This is yet another performance of the so-called
{\it watchdog effect for continuous measurement} (or the continuous collapse of the quantum state)
already mentioned in Section \ref{secder}.
What is even more remarkable is that all solutions to \eqref{eqqufiBlinGanl} converge to Gaussian wave,
as states the following result proved by two methods in \cite{Kol95} and \cite{Kol97}.

\begin{theorem}
\label{thscattering}
Let $\phi$ be the solution of the Cauchy
problem for equation \eqref{eqqufiBlinGanl} with an
arbitrary initial function $\phi_0 \in L_2$ of unit norm. Then for a.a. trajectories of the
innovating Wiener process $B(t)$,
\[
\| \phi -\pi^{1/4}g^{\al(1-i)/\sqrt h}_{q(t),p(t)} \| =O(e^{-\gamma t}),
\]
as $t \to \infty$, for arbitrary $\gamma \in (0,1)$, where $p,q$ evolve as
for Newton particle under white noise force. For instance, if $\al=h=1$, then
\[
\begin{aligned}
& q(t)=q_W+p_Wt+W(t)+\int_0^tW(s) \, ds +O(e^{-\gamma t}) \\
& p(t)=p_W+W+O(e^{-\gamma t})
\end{aligned}
\]

for some random constants $q_W, p_W$.
\end{theorem}

It is an interesting open problem to extend this scattering result for
quantum particles perturbed by some local potential $V$.

\subsection{Path integral representation for classical Hamiltonians}

Mathematically rigorous constructions of Feynman path integral for quantum mechanics
was a very popular topic some time ago, but seems to be a bit out of fashion for the moment.
We refer to \cite{Mazz}, \cite{Kol02}, \cite{Kol01}, \cite{KolBook00}, \cite{Shav}, \cite{Smol} for reviews of many approaches
to this topic and enormous literature. One can distinguish pure analytic
and more probabilistic approaches in these constructions,
the latter being based on either diffusive or jump-type processes. The approach via jump-type processes was
initiated in \cite{Mas}. It can be looked at as a probabilistic interpretation of
perturbation series, that often form the basis for functional integration even in quantum fields,
see e.g. \cite{SF}.

Following \cite{Kol02}, we present here a path integral representation to the solutions
of filtering equations for pure states
\eqref{eqqufiBlinPosMom},  based on jump-type processes.
Again for simplification let us look at a particular case,
choosing now the case of the continuous measurement of the momentum, that is,
choosing now $\al_1=0$ in \eqref{eqqufiBlinPosMom}.
The resulting equation can be written in the form

\begin{equation}
\label{eqqufiBlinMom}
d\chi(t) =\left(\frac12 (i+\la)\De \chi(t) -iV(x) \chi(t)\right))\,dt
- i \sqrt \la\frac{\pa}{\pa x} \chi(t) dY(t),
\end{equation}
with a positive constant $\la$.

In  momentum representation, that is, in terms of the Fourier transform $u$ of $\chi$,
this equation rewrites as
\begin{equation}
\label{eqqufiBlinMomFour}
du(t,y)=\left(-\frac12 (i+\lambda)y^2 u(t,y)
-i V (-i \frac{\pa}{ \pa y})u(t,y)\right)\, dt
+\sqrt{\la}\, yu(t,y) \, dY(t).
\end{equation}

 To write its path integral representation, let us introduce some notations.
Let $PC_p (s,t)$ (abbreviated to $PC_p(t)$, if $s=0$)
 denote the set of all right continuous and piecewise-constant
paths $[s,t] \mapsto \R^d$ starting from the point $p$,
 and let $PC_p^n(s,t)$ denote the subset
of paths with exactly $n$ discontinuities.
Topologically, $PC_p^0$ is a point and
$PC_p^n=Sim^n_t \times (\R^d \setminus \{ 0 \})^n$, $n=1,2,...$, where
\[
Sim_t^n=\{ s_1,...,s_n: 0< s_1 <s_2<...<s_n \le t \}
\]
denotes the standard $n$-dimensional simplex.
To each
$\sigma$-finite measure $M$ on
$\R^d$ without an atom at the origin, there corresponds a $\sigma$-finite
measure $M^{PC}=M^{PC}(t,p)$
 on $PC_p(t)$, which is defined as the sum of measures
$M_n^{PC}$, $n=0,1,...$, where each $M_n^{PC}$ is the
product-measure on $PC_p^n(t)$
of the Lebesgue measure on $Sim_t^n$ and
of $n$ copies of the measure $M$ on $\R^d$.
If paths $Y$ are written as
\[
Y_y(s)=y+Y^{s_1...s_n}_{\delta_1...\delta_n}(s)
=\left\{
\begin{aligned}
& Y_0=y, \quad s<s_1, \\
& Y_1=y+\delta_1, \quad s_1 \le s <s_2, \\
& ... \\
& Y_n=y+\delta_1+\delta_2+....+\delta_n, \quad s_n \le s \le t,
\end{aligned}
\right.
\]
then
\[
M_n^{PC} (dY(.))=ds_1....ds_n M(d\delta_1)...M(d\delta_n).
\]
In the probabilistic language, one says that the moments and sizes of jumps of $Y(s)$ occur
in the atoms of Poisson random measure with intensity measure being $M(d\de) ds$.

The integral of functions $F(Y(.))$ of trajectories over $M^{PC}$ us given by the sum
of finite-dimensional integrals. Thus for any measurable $F(Y(.))$ and $u_0(Y(t))$,
\[
\int_{PC_y(t)} M^{PC}(dY(.)) F(Y(.))u_0(Y(t))
\]
\[
=\sum_{n=0}^{\infty}\int_{PC_y^n(t)} M_n^{PC}(dY(.)) F(Y(.))u_0(Y(t)).
 \]
\begin{equation}
\label{pathintpoi}
=\sum_{n=0}^{\infty}\int_{Sim_t^n}ds_1...ds_n \int_{\R^{dn}} M(d\delta_1)...
M(d\delta_n) \, F(y+Y^{s_1...s_n}_{\delta_1... \delta_n})
u_0(y+\delta_1+...+\delta_n).
\end{equation}

 Let the potential $V$ be the Fourier transform (possibly in the sense of distribution)
of a function $f \in L^1+L^q$, i.e. $f=f_1+f_2$ with $f_1 \in L^1$,
$f_2 \in L^q$, with  $q \in (1, d/(d-2))$, $d>2$.
Notice that this class of potentials includes the Coulomb case
$V(x)=|x|^{-1}$ in $\R^3$, because for this case $f(y)=|y|^{-2}$.

The proof of the following theorem and its various extensions and modifications (including
singular magnetic fields) can be found
in \cite{KolBook00}, \cite{Kol02}, \cite{Kol03}.

\begin{theorem}
\label{thpathint}
 Under the given assumptions on $V$ there exists a
(strong) solution $u(t,y)$ to the Cauchy problem of
equations \eqref{eqqufiBlinMomFour} with any bounded initial data $u_0$,
which is given in terms of the Feynman path integral of type \eqref{pathintpoi}.
More precisely
\[
u(t,y)=\int_{PC_y(t)}M_{Leb}^{PC}(dY(.))F(Y(.))u_0(Y(t)), \eqno (2.1)
\]
where
\[
F(Y(.))=\prod_{j=1}^n(-if(\delta_j))
\]
\[
\times
\exp \{
-\sum_{j=0}^n\left[(\lambda+i/2)Y^2_j(s_{j+1}-s_j)
-\sqrt {\la} Y_j(W(s_{j+1})-W(s_j))\right]\}.
\]
\end{theorem}

It was observed already in \cite{Menski} (by heuristic manipulation with path integral)
that introducing continuous observation
makes certain regularisation to ill-defined Feynman path integral. This fact is rigorously demonstrated
in Theorem \ref{thpathint}, which shows the well-posedness of the path integral for any $\la\neq 0$
(regularization by continuous observation). it is shown in  \cite{Kol02} that, as $\la\to 0$,
the corresponding path integral for free propagation (with $\la=0$) can be understood as certain
improper Riemann integral over the measure $M^{PC}$. Similar consideration can be performed
also for diffusive type measures on trajectories, as explained in \cite{AlKolSmol} and \cite{AlKolSmol1}.

\subsection{Finite-dimensional examples}

Though this paper is mostly devoted to infinite-dimensional quantum mechanics,
the finite-dimensional case is very important from many points of view, and we briefly touch
now upon the properties of solutions of quantum
filtering equations for finite-dimensional Hilbert space $\HC=\C^{d+1}$.
The corresponding projective space of quantum pure states is $P\C^d$.
The first question one usually asks about a diffusion process is whether
there exists an ergodic measure. When the state space is compact, a simple
sufficient condition for the existence of invariant measure is nondegeneracy.
Thus one may ask: what is the minimal number of coordinates of operator $L=(L_1, \cdots, L_n)$
in $\AC$, given by \eqref{eqdifgener}, in order to have a strictly elliptic diffusion,
and how it may look like. The following result in this direction was obtained in
\cite{KolDynQuGames}.

\begin{prop}
\label{propOnLB}
 (i) In case of a qubit $d=1$, if we choose 3 operators $L_1,L_2,L_3$
to be Pauli matrices, then diffusion operator \eqref{eqdifgener}
(up to drift terms specified by $H$) is the
Laplace-Beltrami operator on the projective space $P\C^1$. Moreover, in this case,
operator \eqref{eqdifgener}  and \eqref{eqgenlinfilt} has exactly the same expression when
written in projective coordinates. (ii) For general finite $d$, we get the Laplace-Beltrami
operator as the diffusive part of \eqref{eqdifgener} whenever $L_j$ are chosen as $d^2+2d$
matrices forming an orthnormal basis in the space of self-adjoint traceless matrices
with the Hilbert-Schmidt scalar product given by trace. For instance, in case $d=2$,
one can take 8 Gell-Mann matrices.
\end{prop}

We refer to \cite{KolDynQuGames} for the proof, where this fact was used to the effective analysis
of quantum feedback control and dynamic quantum games. Using this specific measurement arrangements,
controlling operator $H$ (with one controller or several competitive controllers) leads effectively
to the stochastic control problem on $P\C^d$ with drift as the control parameter.

Notice also that using the basis of traceless matrices as coupling operators for reservoirs
is popular both for the study of linear and nonlinear Lindblad equations,
see e.g. \cite{Alicki} and \cite{Fernen}.

What can be said about the existence of ergodic measure for
evolution \eqref{eqqufiBnonlins}
or \eqref{Lindstochnorm1} in degenerate finite-dimensional case, say,
when we have just one coupling operator $L=L_1$?
To answer this question for a general SDE in Stratonovich form
\begin{equation}
\label{genStr}
d\phi =A_0(\phi) dt +A(\phi) \circ \, dB
\end{equation}
on a compact manifold $M$,
where $B=B(t)$ is the standard one-dimensional BM and $A_0, A$ are smooth vector fields on $M$,
one introduces
the controlled dynamic system on $M$:
\begin{equation}
\label{genStrCont}
\dot \phi =A_0(\phi) +u(t) A(\phi), \quad t\ge 0,
\end{equation}
with $u(t):\R_+\to \R$ a piecewise smooth function. Let $O(\phi)$ denote the attainable set for this system
starting at $\phi$. The Stroock-Varadhan support theorem states that the support of the transition probability
for diffusion \eqref{genStr} starting with $\phi$ coincides with the closure $\overline{O(\phi)}$ of $O(\phi)$.
A subset $S\subset M$ is called a control set if $\phi \in \overline{O_{\psi}}$ for any $\phi,\psi\in S$ and $S$
possessing this property is maximal. It was shown in \cite{Klieman} that (i) points belonging to no invariant
control set are transient and (ii) compact invariant control sets $S$
are positive recurrent, that is, there exists an ergodic probability measure $\mu$ on $S$, so that
\begin{equation}
\label{genStrrec}
\lim_{t\to \infty} \E f(\Phi(t,\phi))=\int_S f(\phi) d\mu (\phi)
\end{equation}
for any bounded continuous $f$ and solutions $\Phi(t,\phi)$ to \eqref{genStr} starting from $\phi\in S$.
The following corollary from this general theory to quantum filtering equation
was derived in \cite{Kol93}.

\begin{prop}
Let SDE \eqref{genStr} on a compact manifold $M$ have the following property.
There exists a smooth hypersurface $\Ga$ in $M$ such that (i) the field $A_0$
is transversal to $\Ga$, and (ii) there exists $\phi_0\in M\setminus \Ga$ such that
all solutions of ODE $\dot \phi =A(\phi)$ starting with a point from $M\setminus \Ga$
tend to $\phi_0$ as $t\to \infty$. Then
(i) the corresponding controlled system
\eqref{genStrCont} has unique invariant control set $M_0$, and this set is $\overline{O(\phi_0)}$;
(ii) points outside $M_0$ are transient;
(ii) there exists an invariant probability measure $d\mu$ on $M$ with support on $M_0$ such that
\eqref{genStrrec} holds for any $\phi\in M$.
\end{prop}

Applying this result to quantum filtering equation yields the following working criterium of ergodicity.

\begin{prop}
\label{propmygenerg}
Let a pair of self-adjoint matrices $H$ and $L$ in $\C^{d+1}$ has no nontrivial common invariant subspaces
and the maximal in magnitude eigenvalue of $L$ is simple, $\phi_0$ being the projection of the corresponding
 eigenvector on the projective space $P\C^d$. Then for the process given by SDE
\eqref{eqqufiBnonlinssStr} on $M=P\C^d$ there exists an invariant probability measure
with support $\overline{O(\phi_0)}$ so that \eqref{genStrrec} holds for any $\phi\in M$.
\end{prop}

Various concrete examples of the application of this result including cases of explicit
calculations of invariant measures were given in \cite{Kol93}. In particular, if it holds,
 then variances of all quantum trajectories tend to one and the same finite limit,
 which can be interpreted both
 as a dynamic uncertainty principle and a watch-dog effect.
Moreover, this result
was applied in \cite{BarchPag} yielding a bunch of different concrete cases of ergodicity
for finite-dimensional quantum filtering.
A different kind of limiting behavior is proved in \cite{BenPellegrini}
in finite-dimensional case under the so-called
nondemolition condition.

Finally we refer to \cite{Mora13} for deep results on the existence of invariant measures
for unbounded operators $H,L$  in infinite-dimensional Hilbert spaces.

\section{Quantum feedback control and quantum mean field games}
\label{seccont}


As we mentioned above, application of quantum filtering to feedback control and quantum dynamic
(or differential) games in finite dimensional systems was presented in \cite{KolDynQuGames}.
Examples of the application of filtering with counting type observation to control and games are given in
  \cite{Kol92} and \cite{KolQuantMFGCount}.
Here we briefly discuss quantum mean field games with diffusive observations following mostly \cite{KolQuantMFG}.
Mean field games were initiated in \cite{Huang} and \cite{LL2006} and quickly became one of
the most popular branches of game theory, see e.g. monographs
\cite{CarDelbook18}, \cite{BenFr}, \cite{Gomesbook}, \cite{KolMalbook19} and references therein.
On the other hand, the growth of interest in various quantum technologies brought to life
another new branch of game theory, namely quantum games, initiated in \cite{EWL99}, \cite{MW00}, \cite{MeyerD99},
which again quickly became quite popular, especially in physics literature,
see e.g. reviews in \cite{GuoZhang08}, \cite{KhanReview18}, \cite{KolMal}. However,
the work on quantum games so far was mostly concerned with sequential games.
Only the exploitation of quantum filtering allowed one to create really dynamic
(stochastic differential) games that eventually led to the quantum mean-field games
as a natural merger of two popular new directions in game theory.

We will assume the readers are familiar with the main concepts of game theory like Nash equilibrium,
see any textbook on game theory, for instance \cite{KolMal}, \cite{KolMalRus}.

Let us consider the quantum dynamic game of $N$ players, where the dynamics
of the density matrix $\Ga_{N,t}$ is given by the controlled dynamics of type
\eqref{eqmainNpartBeldensnonl}, where
the control of each player can be chosen independently on the basis of its 'position' $\Ga^{(j)}_{N,t}$:

 \[
d \Ga_{N,t}
=-i\sum_j [H_j+u_j(t,\Ga^{(j)}_{N,t}) \hat H_j, \Ga_{N,t}] -\frac{i}{N} \sum_{l<j\le N} [A_{lj},\Ga_{N,t}]
\]
\begin{equation}
\label{eqmainNpartBeldenscoindiv}
+\sum_j (L_j\Ga_{N,t} L_j^* -\frac12 L^*_jL_j \Ga_{N,t} -\frac12 \Ga_{N,t} L^*_jL_j)\, dt
+\sum_j (\Ga_{N,t} L^*_j+L_j \Ga_{N,t}) dY^j_t.
\end{equation}

Assume further that control $u$ can be chosen from some bounded closed interval $U$ of the real line,
 that the initial matrix is the product of identical states,
\[
\Ga_{N,0}(x_1, \cdots, x_n; y_1, \cdots, y_N)=\prod_{j=1}^N \psi(x_j)\overline{\psi(y_j)},
\]
and that the payoff of each player on the interval $[t,T]$ is given by the expression
 \begin{equation}
\label{eqcostfunin}
P_j(t, W ; u(.)) =\int_t^T   \left( {\tr} (J_j \Ga_{N,s}) -\frac{c}{2} u_j^2(s)\right)\, ds +{\tr} (F_j\Ga_{N,T}),
\end{equation}
where $J$ and $F$ are some operators in $L^2(X)$ expressing the current and
the terminal costs of an agent, $J_j$ and $F_j$ denote their actions
on the $j$th variable, $c$ measures the cost of applying the control.

Notice for clarity that by the property of the partial trace, the payoff
\eqref{eqcostfunin} can be also written as
 \begin{equation}
\label{eqcostfunineq}
P_j(t, W ; u(.)) =\int_t^T  \left( {\tr} (J_j \Ga^{(j)}_{N,s}) -\frac{c}{2} u_j^2(s)\right)\, ds
+{\tr} (F_j\Ga^{(j)}_{N,T}).
\end{equation}
Therefore, it really depends only on the individual partial traces $\Ga^{(j)}_{N,t}$,
which can be considered as quantum analogs of the positions of individual classical particles.

In view of Theorem \ref{thmynonlinSchcont},
the limiting evolution of each player can be expected to be  described by the equations
\[
d\ga_{j,t}=-i[H+u_j(t,\ga_{j,t}) \hat H,\ga_{j,t}] \, dt-i[A^{\overline{\eta_t}}, \ga_{j,t}] \, dt
+(L\ga_{j,t} L^* -\frac12 L^*L \ga_{j,t} -\frac12 \ga_{j,t} L^*L)\, dt
\]
\begin{equation}
\label{eqmainnonlinpartBeldens}
+(\ga_{j,t} L^*+L \ga_{j,t}) dY^j_t, \quad \eta_t(x,y)=\lim_{N\to \infty}\frac{1}{N}\sum_{j=1}^N \ga_{j,t}(x,y),
\end{equation}
with payoffs given by
 \begin{equation}
\label{eqcostfuninlim}
P_j(t, W ; u(.)) =\int_t^T  \left(  {\tr} (J \ga_{j,s}) -\frac{c}{2} u_j^2(s)\right)\, ds +{\tr} (F\ga_{j,T}).
\end{equation}

\begin{remark} Let us stress that we are not stating that evolution \eqref{eqmainnonlinpartBeldens}
is in fact the limiting one for the $N$-agent quantum evolution in this general case,
because we really do not need it. We need it only in case when almost all players (actually except
for one only) are playing the same strategy $u_t^{com}$, in which case $\eta_t=\E \ga_{j,t}$ for all $j$ adhering
to the common strategy.
\end{remark}

Let us say that the pair of functions $u^{MFG}_t(\ga)=u^{MFG}(t,\ga)$ with $t\in [0,T]$
and $\ga$ from the set of density matrices in $L^2(X)$, $u\in U$,
and $\eta^{MFG}_t(x,y)$ with $x,y\in X$, $t\in [0,T]$, solve the limiting MFG problem if
(i) $u_t(\ga)$ is an optimal feedback strategy for the stochastic control problem \eqref{eqmainnonlinpartBeldens},
\eqref{eqcostfuninlim} under the fixed function $\eta_t=\eta_t^{MFG}$ and (ii) $\eta_t^{MFG}$ arises
from the solution of \eqref{eqmainnonlinpartBeldens} under fixed $u_t=u_t^{MFG}$.

We can formulate it also in another equivalent way. For a function $u_t^{com}(\ga)$
(index 'com' from 'common') suppose we can solve
the Cauchy problem for SDE \eqref{eqmainnonlinpartBeldens} with $u_t=u_t^{com}$ defining the correlations
$\eta_t(x,y)=\E \, \ga_{j,t}(x,y)$. Given these correlations we may be able to find an optimal feedback
control for the individual control problem \eqref{eqmainnonlinpartBeldens},
\eqref{eqcostfuninlim} under the fixed function $\eta_t$
 defining the individually optimal feedback
control $u_t^{ind}(w)$ (index 'ind' from 'individual').
The main MFG consistency equation is then expressed by
the equation  $u_t^{com}=u_t^{ind}$. If it is fulfilled, the pair $u_t^{com}$ and $\eta_t$
solves the limiting MFG in the sense defined above.

Formulated in this way, the MFG problem is fully classical, though the state space is the
sphere in the Hilbert space (or the space of density matrices).
 What makes this story truly quantum is the completely different
link with $N$-agent quantum game arising essentially from Theorem \ref{thmynonlinSchcont}.
 One an say that we establish the correspondence between quantum $N$-agent games and classical
$N$-agent games on some Riemannian manifold (possibly infinite-dimensional), so that the limiting
MFG forward-backward system is identical for both the quantum game and its classical counterpart.

\begin{theorem}
\label{mainresult}
Let the conditions on $H,L,A$ from Theorem \ref{thmynonlinSch} hold and let $\hat H$ be a
bounded self-adjoint operator  in $L^2(X)$.
Assume that the pair  $u^{MFG}_t(\ga)$ and $\eta^{MFG}_t(x,y)$
solves the limiting MFG problem and moreover $u^{MFG}_t$ is Lipschitz in the sense of
inequality \eqref{eq1thmynonlinSchco}.
Then the strategies
\[
u_j(t,\Ga_{N_t})=u^{MFG}_t(\Ga^{(j)}_{N,t}),
\]
where $\Ga^{(j)}_{N,t}$ is the partial trace of $\Ga_{N,t}$ with respect to all variables
except of the $j$th, form a symmetric $\ep$-Nash equilibrium for the $N$-agent quantum game
described by \eqref{eqmainNpartBeldenscoindiv} and \eqref{eqcostfunin},
with $\ep$ being of order $N^{-1/4}$.
\end{theorem}

We refer to \cite{KolQuantMFG} for the proof of this theorem. In \cite{KolQuantMFGCount} an analog of this result is
presented for counting type observations.

Let us make a comment on the information space of the players.
In the game above the players were allowed to have access to their individual partial traces $\Ga_{N,t}^{(j)}$.
In the spirit of classical MFGs one could imagine them to have access to 'empirical measures',
which in our case represent the average operators
 \begin{equation}
\label{eqdefquantempirmes}
\frac{1}{N} \sum_{j=1}^N \Ga_{N,t}^{(j)}
\end{equation}
considered as operators in $L^2(X)$.
Since the averages \eqref{eqdefquantempirmes} approach in expectation the expected correlations $\eta_t$,
allowing $u_j$   in \eqref{eqmainNpartBeldenscoindiv} to depend on this averages amounts
to the dependence on $\eta_t$ in the limit, which is already taken into account in the construction of the MFG
consistency problem. Consequently this additional information possibility will not change the result
of Theorem  \ref{mainresult}.

 Note finally that Theorem \ref{mainresult} is a conditional result asserting some property of an object,
 if it does exist. The problem of existence of the solution to quantum MFG problem is not yet sorted out
 in general infinite-dimensional setting. However, in finite-dimensional setting and the arrangement
 of coupling operators as in Proposition \ref{propOnLB}, the existence and local uniqueness
 of the limiting MFG problem is proved in \cite{KolQuantMFG}.

\section{Bibliographic notes and open problems}
\label{secbibopen}

A study of mean-field Belavkin's equations of type \eqref{LindstochnewBel} initiated in
\cite{KolQuantLLN}, \cite{KolQuantMFG} was of course a natural next step after
the consideration of nonlinear mean-field Lindblad equations. The simplest version of it can be obtained
by ignoring stochastic differential in \eqref{LindstochnewBel}. However, many other versions were under
attention of both mathematicians and physicists, see e.g. Section 3.7 in \cite{Breuer}, Section 1.4.3 in
\cite{Alicki} for initial contributions, Section 11.3 of \cite{Kolbook10}
for some basic well-posedness results, and \cite{Fernen} and references therein
for more modern development
 including detailed study of long-time behavior for various finite-dimensional systems.
Let us mention important recent contributions to mean-field Belavkin equations in
\cite{Bouard} and \cite{Chalal}.

Fractional quantum mechanics as presented in Section \ref{secfraceq} gives
a way to incorporate memory in classical Markov models.
As more conventional (and well-established) methods for introducing non-Markovian models,
one can mention, for instance, the Nakajima-Zwanzig equations and time-convolutionless
(TCL) projection operator technique, as explained, e.g. in Chapter 4 of \cite{Breuer}.

 Since a continuous feedback control via quantum filtering is still not easy to realise practically,
 the study of quantum control was so far mostly oriented on either programming control (without feedback) or
via feedback in discrete times. The literature is enormous, and we give just couple of
examples of problems arising. Most basic problem is to steer optimally one given state of the system
to another one, with control consisting of both Hamiltonian control (via say, coherent electromagnetic fields)
and measurement control (choosing an appropriate sequence of operators to measure in given sequence of time).
Characteristic paper in this direction is \cite{Pechen}, where analytic solution for optimal transfer
of one state to another in a qubit is given under this setting, see also \cite{PechRab} for related simulation
with concrete reservoirs. Another well-addressed question, see e.g. \cite{RVP},
closely related to the effective use of modern optimization algorithms
in quantum setting, consists in identifying natural classes of optimization problems
with the absence of traps, that is, of local extremum.

Mathematic analysis of finite-dimensional feedback control (asymptotic stability, controllability, etc)
has been also developing during recent years (rather slowly so far), see e.g. \cite{VanHandelControl}
or more recent papers \cite{Liang19}, \cite{Liang22} and references therein.
Usual tool for qualitative analysis is given by the Lyapunov second method
and its stochastic counterpart given by the stochastic LaSalle-type theorem.
As a notable related development let us mention the study of hybrid systems,
see \cite{Barch24} and \cite{Dios}, and quantum feedback networks
based on non-commutative fractional linear, or M\"obius, transformations,
see \cite{Gough}.

An extremely interesting effect from qualitative theory of quantum trajectory occurs
in the strong measurement regime. Namely, under certain conditions solutions to quantum
filtering equations behave as a pure jump Markov process on a finite set (so-called pointer states),
see e.g. \cite{Bauer} for physical presentations including references on real experiments and
 \cite{Ballesteros}, \cite{BenCedric} and references therein on rather deep mathematics results
 (concerning only finite-dimensional settings).
Even more subtle and surprising phenomenon of "spiking" (appearance of spikes around limiting
jump-type-process trajectories) can be observed and analysed, see \cite{Bernard} and references therein.

Though the present paper deals with mathematical developments, let us mention
the technical experimental side of observing quantum trajectories
and organising feedback quantum control. This topic has been actively investigated
in recent time, see e.g. \cite{Armen02Adaptive}, \cite{Bushev06Adaptive}, \cite{WiMilburnBook},
\cite{Breuer}, and more up-to-date progress in \cite{HarochePaperNature}, \cite{HarocheNobel},
\cite{WinelandNobel}, \cite{Murch}, \cite{Weber}, \cite{Zhang}.

To conclude let us indicate some open problems naturally arising from above:

(1) Convergence to limiting mean-field quantum filtering in case of counting observation
was obtained so far only for unitary coupling operators. What can one say about the limit
 for general bounded couplings $L$?

(2) Existence of solutions (and local uniqueness) for the limiting forward-backward systems
of quantum mean-field games was obtained only in finite-dimensional case. Which conditions would guarantee
the existence (and uniqueness in some cases) for basic infinite-dimensional settings? In fact, work on 
mean-field games and McKean-Vlasov equations in infinite dimensions is yet in its initial phase only, 
see however \cite{Ahmed16}. 

(3) The dynamic law of large numbers given by the limiting mean-field quantum filtering Belavkin's equations
naturally suggests the question about the description of fluctuations. Can one get some
kind of dynamic central limit theorem (CLT)? For classical dynamic laws of large numbers expressed by
Smoluchovskii and Boltzmann equations the corresponding CLTs were obtained in \cite{Kol10} and \cite{Kolbook10}.

(4) By including interaction in the coupling operators $L$, it should be possible to obtain more general
mean-field quantum filtering equations that would present the full stochastic counterpart
to the existing nonlinear Lindblad equations.

(5) Can one obtain the analog of the scattering result,
Theorem \ref{thscattering}, for a particle in a potential field.

(6) With quantum trajectories for infinite-dimensional quantum filtering being built, it is natural to
ask for the extension of many aspects developed for finite-dimensional systems to infinite-dimensional
situations. For instance, this concerns the results on feedback control and on strong measurement limit,
as mentioned briefly above.

(7) This paper was restricted to bounded coupling operators $L$.
The results of  \cite{Kol25b} provides the well-posedness of basic quantum filtering SDEs for mixed states
for certain cases of unbounded $L$. It should be possible to extend the derivations of Section  \ref{secder}
and the propagation of chaos results of Section \ref{secLLN} to these situations. Moreover, we dealt here with
finite-dimensional Brownian motion $Y(t)$ and $B(t)$. It is also popular in literature to analyse
infinite-dimensional noises $B(t)$, see e.g. \cite{MoraRebo} and references therein. It would be
interesting to include infinite-dimensional noises in the present theory.

{\bf Acknowledgements.} The author delivered  numerous talks on the topic of the paper on various national
and international events. Especially let me mention the presentations of the last academic year
2025-2026 on research seminars of A.S. Kholevo in MIAN RAS and of A.N. Pechen in MIAN RAS, on the international
conference "Theory of functions, Theory of operators and Quantum information theory" in Kolomna, May 2026, 
and finally on the Big Seminar of the Department of Probability under supervision of A.N. Shyriaev, October 2026. 
My gratitude to all listeners for their attention and fruitful discussions.

{\bf Funding.}
The work was supported by the Ministry of Education and Science of
 the Russian Federation as part of the program of the Moscow Center
  for Fundamental and Applied Mathematics under the Agreement No. 075-15-2025-345;

\end{document}